\documentclass[12pt,reqno]{amsart}
\usepackage{amsmath,amsfonts,amssymb,amscd,amsthm,amsbsy,color}
\usepackage[dvips]{graphicx}
\usepackage{comment}

\usepackage{wrapfig}
\usepackage{fancybox}
\usepackage[maxfloats=100]{morefloats}

\usepackage{graphicx,epsfig}

\numberwithin{equation}{section}

\newtheorem{theorem}{Theorem}
\newtheorem{meta-thm}[theorem]{Meta-Theorem}

\newtheorem{proposition}[theorem]{Proposition}

\newtheorem{remark}[theorem]{Remark}
\newtheorem{definition}[theorem]{Definition}

\newcommand\beq[1]{ \begin{equation}\label{#1} }
\newcommand{\eeq}{ \end{equation} }

\newcommand\beqa[1]{ \begin{eqnarray} \label{#1}}

\newcommand{\red}{\textcolor{black}}

\newcommand{\eeqa}{ \end{eqnarray} }
\newcommand{\beqano}{ \begin{eqnarray*} }
\newcommand{\eeqano}{ \end{eqnarray*} }
\newcommand\equ[1]{{\rm (\ref{#1})}}

\def\F{{\mathcal F}}

\def\G{{\mathcal G}}

\def\H{{\mathcal H}}

\def\integer{{\mathbb Z}}

\begin{document}

\title[Resonances in the Earth's space environment]
{Resonances in the Earth's space environment}

\author[A. Celletti]{Alessandra Celletti}

\address{
Department of Mathematics, University of Roma Tor Vergata, Via
della Ricerca Scientifica 1, 00133 Roma (Italy)}
\email{celletti@mat.uniroma2.it}

\author[C. Gales]{Catalin Gales}

\address{
Department of Mathematics, Al. I. Cuza University, Bd. Carol I 11,
700506 Iasi (Romania)}
\email{cgales@uaic.ro}

\author[C. Lhotka]{Christoph Lhotka}

\address{
Space Research Institute, Austrian Academy of Sciences
Schmiedlstrasse 6, A-8042 Graz, Austria}
\email{christoph.lhotka@oeaw.ac.at}

\thanks{Corresponding author: \sl E-mail address: \rm celletti@mat.uniroma2.it (Alessandra Celletti)}
\thanks{A.C. was partially supported by GNFM-INdAM, MIUR-PRIN 20178CJA2B ``New Frontiers of Celestial Mechanics: theory and Applications'' and
acknowledges the MIUR Excellence Department Project awarded to the
Department of Mathematics, University of Rome Tor Vergata, CUP
E83C18000100006.
 C.G. was partially supported by CNCS-UEFISCDI, project number
PN-III-P1-1.1-TE-2016-2314. C.L. was partially supported by
the Austrian Science Fund (FWF) with project number P30542-N27.
A.C. and C.G. acknowledge the EU-ITN Stardust-R}


\baselineskip=18pt              




\begin{abstract}
We study the \red{presence} of resonances in the region of space
around the Earth. We consider a massless body (e.g, a dust
particle or a small space debris) subject to different forces: the
gravitational attraction of the geopotential, the effects of Sun
and Moon. We distinguish different types of resonances: tesseral
resonances are due to a commensurability involving the revolution
of the particle and the rotation of the Earth, \red{semi-secular
resonances include the rates} of variation of the mean anomalies
of Moon and Sun, \red {while} secular resonances just
\red{depend} on the rates of variation of the arguments of
perigee and the longitudes of the ascending nodes \red{of the
perturbing bodies}.  We \red{characterize} such resonances,
giving precise statements on the regions where the resonances can
be found and \red{provide} examples of some specific
commensurability relations.
\end{abstract}


\maketitle

\tableofcontents

\section{Introduction}\label{sec:intro}

The dynamics of an object around the Earth has revealed many
interesting aspects, thanks to the wide variety of behaviors
occurring at different altitudes from the Earth. The space around
our planet is usually split into three different regions,
characterized by the acronyms LEO, MEO and GEO, where LEO stands
for Low-Earth-Orbit ranging up to 2\,000 km of altitude, MEO
stands for Medium--Earth--Orbit between 2\,000 and 30\,000 km, GEO
stands for \red{Geostationary--Earth--Orbit}. Objects in these
regions feel different effects: Earth's gravitational potential
(including both the Keplerian part and the potential due to the
non-spherical shape), gravitational influence of Sun and Moon,
effect of Solar radiation pressure, dissipation due to the
atmospheric drag (acting only in LEO), \red{tides, or reradiation
from the atmosphere}. In this work we will not be concerned with
the dynamics in LEO (which requires the analysis of the
atmospheric drag), but we are rather concentrated on the analysis
of orbits in MEO and GEO \red{neglecting all non-conservative
forces}. Even more specifically, we focus on resonant motions
occurring in MEO and GEO. Due to the high complexity of the
system, there are several kinds of resonances, which involve the
rates of variation of different variables associated to the
object, Earth, Sun and Moon, each variable having its own
time-scale.

More precisely, the angular variables associated to the dynamics
of an object around the Earth are the mean motion, the argument of
perigee and the longitude of the ascending node. Analogous
variables are used to describe the dynamics of Sun and Moon, while
the rotational motion of the Earth is described by the sidereal
time. Associated to these variables, we have three different time
scales:

\begin{description}
        \item[$(i)_t$] \sl short periodic terms, \rm involving the mean anomaly and the sidereal time, which are
        fast angles with periods of days;
        \item[$(ii)_t$] \sl semi-secular terms, \rm involving the mean anomalies of Moon and/or Sun,
        which are semi-fast angles with periods of one month and one
        year;
        \item[$(iii)_t$] \sl secular terms, \rm involving the arguments of perigee and the longitudes of the ascending \red{nodes of the
        perturbers},
        which are slow angles with periods of several years.
\end{description}

As a consequence, we have \red{a hierarchy of resonances},
involving combinations of the angles \red{evolving on different
time scales: fast angles, like the mean motion of the object and
the sidereal time of the Earth, slow angles, like the arguments of
perigee or the longitudes of the ascending node (either of the
object, Sun or Moon), intermediate angles (entering the
semi-secular resonances), like the mean motions of the Moon or the
Sun.} We will give precise definitions in the following sections,
but we anticipate here the terminology:

\begin{description}
        \item[$(i)_r$] \sl tesseral resonances, \rm occurring when there is a commensurability relation between
        the rates of variation of the mean motion of the object and of the sidereal time;
                \red{such \red{a} relation will involve also the rates of variation of the argument of perigee and
        the longitude of the ascending node of the object;}
        \item[$(ii)_r$] \sl Lunar or Solar semi-secular resonances, \rm occurring when there is a commensurability relation
        between the rates of variation of the argument of perigee
        of the object, the longitude of the ascending node of the
                object, \red{and} the mean motion of \red{the} Moon or \red{the} Sun;
        \item[$(iii)_r$] \sl Lunar or Solar secular resonances, \rm occurring when there is a commensurability relation between
        the rates of variation of the arguments of perigee and the longitudes of the ascending node of
        the object and the Moon, or the object and the Sun.
\end{description}

The most \red{deeply investigated} examples of resonances in the Earth's space
environment are the 1:1 and 2:1 tesseral resonances, \red{where}
the orbital frequency of the object is equal or twice that of
the Earth; in the 1:1 resonance we find geostationary satellites,
while in the 2:1 resonance we find GPS satellites (see, e.g.,
\cite{ASRV}, \cite{EH}, \cite{GCPE}, \cite{GC}, \cite{Klinkrad},
\cite{PCDL}, \cite{SARV}, \cite{SRT}). The 1:1 and 2:1 resonances are located,
respectively, at 42\,164 km and 26\,560 km from the center of the
Earth.

Other tesseral resonances might be equally relevant \red{for
spacecraft operations}, \red{while lunisolar, semisecular and
secular resonances, are important in designing disposal
strategies. The study of resonances} deserves much attention,
especially in the context of space debris, which are small
remnants of satellites, left in space after collisions, explosions
or when a satellite becomes non-operative. Space debris represent
a great threat due to the damages that their collision can provoke
with operational or manned satellites.

\red{Being covered by a web-like structure of resonances that
induce various effects at different time scales, the
circumterrestrial space is a dynamically complex environment.
Inside resonance regions, the dynamical behavior of an
uncontrolled object depends on the time scale at which the motion
is studied and its initial position in the phase space. This
remark motivates the need for a {\it systematic classification} of
the {\it different types of resonances}, which manifest at
different time scales, and a thorough investigation of their
locations. Such a task involves explicit definitions and formulas
for all constraints, a characterization of the main dynamical
effects of resonances from the same category, a study of the
location of resonances as a function of various parameters.}

In this work we make a comprehensive investigation of the
occurrences of the different kinds of resonances described in
$(i)_r$-$(ii)_r$-$(iii)_r$ above. The natural tool is to adopt
Hamiltonian formalism using a model that includes the geopotential
(expanded in spherical harmonics and limited to a finite number of
terms), the Lunar and Solar potentials. Within this setting, we
give a characterization of the tesseral, Lunar and Solar
semi-secular resonances, and Lunar and Solar secular resonances.
In particular, we give results on the regions in the orbital
elements space (semimajor axis, eccentricity and inclination)
\red{where the resonances can be found and how their location
changes with the variation of the semimajor axis, eccentricity and
inclination. We emphasize the fact that, even when large variations
of the eccentricity and inclination induce only a small change in
the location of the resonance, such a change may be important when
compared to the size of the separatrices, dividing resonant
and non-resonant kinds of motions in phase space. Moreover, we
give a brief description of the main dynamical effects and
phenomena of resonances from the same category. If tesseral
resonances induce a variation of the semimajor axis on time scales
of the order of hundreds of days, lunisolar resonances provoke
variations of eccentricity and inclination  on much longer time
scales, of the order of tens or hundreds of years.} The study of
resonances depends, of course, on the (integer) values of the
coefficients providing the resonance relation. The results \red{of
this work} can be used as a guideline to study the geography of
the different kinds of resonances, hence affecting the choice of
the location of satellites or rather giving information on
possible locations of disposal orbits.

This article is organized as follows. In
Section~\ref{sec:hamiltonian} we describe the model (including the
geopotential, Lunar and Solar disturbing functions) and we give
the definitions of tesseral, semi-secular and secular resonances.
The quadrupolar approximation is described in
Section~\ref{sec:quadrupolar}. The study of tesseral resonances is
performed in Section~\ref{sec:tesseral}, while semi-secular
resonances are investigated in Section~\ref{sec:semi}, and secular
resonances are described in Section~\ref{sec:secular}. \red{Some
conclusions are drawn in Section~\ref{sec:conclusions}. A short
survey of the chaos indicator used in this work (the Fast Lyapunov
Indicator) is given in the Appendix.}

\section{The model and the resonances}\label{sec:hamiltonian}

We consider a massless body, say $S$, orbiting around the Earth in
a region that includes both MEO and GEO; in our analysis, we do
not include the LEO region in which the dissipative effect due to
the atmospheric drag should be considered, thus changing the
dynamics from conservative to dissipative. Moreover, we also
neglect non-gravitational forces, like \red{reradiation of the
atmosphere, tides,} the interaction with the Solar wind, or Solar
radiation (acting also in GEO, see \cite{Lho2016}). We study the
dynamics of $S$ under the effects of the oblateness of the Earth,
and the Lunar and Solar attractions.

This study is accomplished by adopting the Hamiltonian formalism and by introducing the action--angle Delaunay variables, denoted as
$(L,G,H,M,\omega,\Omega)$. We remind that the action variables $(L,G,H)$ are related to the orbital
elements $(a,e,i)$ by the expressions
\begin{equation}\label{LGH_aei}
L=\sqrt{\mu_E a}\,,\qquad  G=L \sqrt{1-e^2}\,, \qquad H=G \cos
i\,,
\end{equation}
where $a$ is the semimajor axis, $e$ the eccentricity, $i$ the
inclination. As for the angle variables, the physical meaning is
the following: $M$ is the mean anomaly, $\omega$ is the argument of
perigee, $\Omega$ is the longitude of the ascending node.
The quantity $\mu_E$ is equal to $\mu_E=\G m_E$ with $\G$ the gravitational constant and $m_E$ the
mass of the Earth. The orbital elements of the small body are referred to the celestial equator.

\red{The corresponding Hamiltonian can be written as
\beq{ham}
\mathcal{H}=-{\mu^2_E\over {2L^2}}+\mathcal{H}_{Earth}(\Upsilon,\theta)
-\mathcal{R}_{Sun}(\Upsilon,\Upsilon_S)-\mathcal{R}_{Moon}(\Upsilon, \Upsilon_M)\ ,
\eeq
where $\theta$ denotes the sidereal time, we denote by $\Upsilon=(a, e, i, M,\omega,\Omega)$,
$\Upsilon_S=(a_S, e_S, i_S, M_S,$ $\omega_S, \Omega_S)$, $\Upsilon_M=(a_M, e_M, i_M, M_M, \omega_M, \Omega_M)$
the orbital elements of the massless body,  Sun and Moon, while $\mathcal{H}_{Earth}$, $\mathcal{R}_{Sun}$, $\mathcal{R}_{Moon}$ describe the perturbations due to the Earth, Sun and Moon, respectively.}

\subsection{The geopotential part}\label{sec:geopotential_Ham}
Denoting by $R_E$ the Earth's radius, following \cite{Kaula} we expand $\mathcal{H}_{Earth}$ as
\beq{Rearth}
\mathcal{H}_{Earth}=- {{\mu_E}\over a}\ \sum_{n=2}^\infty \sum_{m=0}^n \Bigl({R_E\over a}\Bigr)^n\ \sum_{p=0}^n F_{nmp}(i)\
\sum_{q=-\infty}^\infty G_{npq}(e)\ S_{nmpq}(M,\omega,\Omega,\theta)\ ,
\eeq
where the functions $F_{nmp}$, $G_{npq}$, $S_{nmpq}$ are given by classical relations (\cite{Kaula}), which are given
below for completeness. The expression of $F_{nmp}$ is given by
\beqa{Ffun}
F_{nmp}(i)&=&\sum_w {{(2n-2w)!}\over {w!(n-w)!(n-m-2w)!2^{2n-2w}}} \sin^{n-m-2w}i\ \sum_{s=0}^m\left(\begin{array}{c}
  m \\
  s \\
 \end{array}\right)
 \cos^si\nonumber\\
 &&\times \sum_c \left(\begin{array}{c}
  n-m-2w+s \\
  c \\
 \end{array}\right)
\left(\begin{array}{c}
  m-s \\
  p-w-c \\
 \end{array}\right)
 (-1)^{c-k}\ ,
\eeqa
with $k=[{{n-m}\over 2}]$, $w$ is summed from zero to the $\min(p,k)$,
$c$ is taken over all values such that the binomial coefficients are not
identically zero. The expression for $G_{npq}$ is given by
\beq{Gfun}
G_{npq}(e)=(-1)^{|q|}(1+\beta^2)^n\beta^{|q|}\ \sum_{k=0}^\infty P_{npqk}Q_{npqk}\beta^{2k}\ ,
\eeq
where
$$
\beta={e\over {1+\sqrt{1-e^2}}}\ ,
$$

$$
P_{npqk}=\sum_{r=0}^h \left(\begin{array}{c}
  2p'-2n \\
  h-r \\
 \end{array}\right)
{{(-1)^r}\over r!}
({{(n-2p'+q')e}\over {2\beta}})^r\ ,
$$
with $h=k+q'$ when $q'>0$ and $h=k$ when $q'<0$, and
$$
Q_{npqk}=\sum_{r=0}^h \left(\begin{array}{c}
  -2p' \\
  h-r \\
 \end{array}\right)
{1\over r!}
({{(n-2p'+q')e}\over {2\beta}})^r\ ,
$$
where $h=k$ when $q'>0$ and $h=k-q'$ when $q'<0$,
$p'=p$ and $q'=q$ when $p\leq n/2$, $p'=n-p$ and $q'=-q$ when $p> n/2$.
The expression of $S_{nmpq}$ is given by
\beq{S}
S_{nmpq}=\left[%
\begin{array}{c}
  C_{nm} \\
  -S_{nm} \\
 \end{array}%
\right]_{n-m \ odd}^{n-m \ even} \cos \Psi_{nmpq}+ \left[%
\begin{array}{c}
  S_{nm} \\
  C_{nm} \\
 \end{array}%
\right]_{n-m \ odd}^{n-m \ even} \sin \Psi_{nmpq}\ ,
\eeq
where
\beq{psi}
\Psi_{nmpq}=(n-2p) \omega+(n-2p+q)M+m(\Omega-\theta)\ .
\eeq
We next introduce the coefficients $J_{nm}$ and the quantities $\lambda_{nm}$ so that
$$
J_{nm} = \sqrt{C_{nm}^2+S_{nm}^2}   \quad \textrm{if} \ m\neq 0\ , \qquad    J_{n0} \equiv J_n= -C_{n0}    \ ,$$
$$
C_{nm}=-J_{nm} \cos(m \lambda_{nm}) \ , \qquad S_{nm}=-J_{nm} \sin(m \lambda_{nm}) \ .
$$
With this notation, we can express $S_{nmpq}$ in the form
$$ S_{nmpq}=\left\{%
\begin{array}{cc}
  -J_{nm}  \cos \Psi_{nmpq}  & \textrm{if} \ n-m\  \textrm{is even} \\
  -J_{nm}  \sin \Psi_{nmpq}  & \textrm{if} \ n-m\ \textrm{is odd}\ . \\
 \end{array}%
\right. $$
We remark that the Fourier series associated to $\mathcal{H}_{Earth}$ in \equ{Rearth} contains an infinite number of
terms. Nevertheless, the long term variation of the orbital elements is governed by the
secular and resonant terms. Moreover, we discussed in \cite{CGmajor}, \cite{CGLEO} how to
reduce the series to a finite number of terms, namely the terms which turn out to be the most relevant ones for the
description of the dynamics.

\vskip.1in

We remark that the angle $\Psi_{nmpq}$ in \equ{psi} depends on linear combinations of different quantities,
including the mean anomaly of the small particle and the sidereal \red{time, hence there} can occur
resonances of the form
\beq{T}
(n-2p+q)\dot M-m\dot\theta+(n-2p)\dot\omega+m\dot\Omega=0\ .
\eeq
\red{Setting $\ell=n-2p+q$ and $m=j$, we can re-write \equ{T} as
\beq{T2}
\ell\ \dot M-j\ \dot\theta+j\ \dot\Omega+\ell\ \dot\omega-q\ \dot\omega=0\ .
\eeq
Let us introduce the \red{resonance angle}
\beq{sigma}
\sigma_{j\ell}=\ell \ M-j\ \theta+j\ \Omega+\ell\ \omega\ ,
\eeq
so that \equ{T2} represents the rate of variation of
\beq{sigma}
\sigma_{j\ell}-q\ \omega\ .
\eeq
This discussion motivates the following definition.}

\vskip.1in

\red{
\begin{definition}\label{def:resonance}
A tesseral resonance of order $j:\ell$ with $j$, $\ell\in\integer\backslash\{0\}$
occurs whenever the orbital period of the massless particle, the rotational period of the Earth,
the rates of variation of the argument of perigee and the longitude of the ascending node
of the massless particle satisfy the relation
\beq{TR}
\ell\ \dot{M}-j\ \dot{\theta} +j\ \dot\Omega+\ell\ \dot\omega= 0\ , \qquad j,\ell \in \integer\backslash\{0\}\ .
\eeq
\end{definition}
}

\red{
\begin{remark}
We notice that for $J_2=0$, \red{we have $\dot\omega=\dot\Omega=0$ (see \equ{elements} below); hence,} the tesseral resonance reduces to a commensurability relation
between the orbital period of the massless particle and the rotational period of the Earth,
namely
$$
\ell\ \dot{M}-j\ \dot{\theta} = 0\ , \qquad j,\ell \in \integer\backslash\{0\}\ .
$$
\end{remark}
}

\red{
\noindent
For $q\not=0$, the term $q\ \dot\omega$ \red{in \equ{T2}} generates a multiplet of resonances (\cite{CGFrontiers}),
thus leading to the following definition of multiplet tesseral resonance, that extends the
Definition~\ref{def:resonance} of tesseral resonances.
}

\red{
\begin{definition}\label{def:multiplet_resonance}
A multiplet tesseral resonance of order $j:\ell:q$ for $j,\ell,q\in\integer\backslash\{0\}$ occurs whenever
the following relation is satisfied
\beq{multipletcond}
\ell\ \dot M-j\ \dot\theta+j\ \dot\Omega+\ell\ \dot\omega-q\ \dot\omega=0\ .
\eeq
\end{definition}
}

Retaining only a finite number, say $N\in\integer_+$, of terms, we can approximate $\mathcal{H}_{Earth}$ by
\beq{RE}
\mathcal{H}_{Earth}=\mathcal{H}^{sec}_{Earth}+\mathcal{H}_{Earth}^{res}+\mathcal{H}_{Earth}^{nonres}\cong
\sum_{n=2}^N \sum_{m=0}^n \sum_{p=0}^n \sum_{q=-\infty}^{\infty} \mathcal{T}_{nmpq} \ ,
\eeq
where $\mathcal{H}^{sec}_{Earth}$, $\mathcal{H}_{Earth}^{res}$, $\mathcal{H}_{Earth}^{nonres}$ are the secular, resonant and non--resonant
parts of the Earth's potential (\cite{CGmajor}); the sums have been truncated
to a suitable finite order $N\in\integer_+$,  and the coefficients $\mathcal{T}_{nmpq}$ are defined by:
\begin{equation}\label{T_nmpq_term}
\mathcal{T}_{nmpq}=-\frac{\mu_E R_E^n}{a^{n+1}}\ F_{nmp}(i)G_{npq}(e) S_{nmpq}(M, \omega, \Omega , \theta)\ .
\end{equation}

\subsubsection{The secular and resonant parts of the geopotential}\label{sec:secular1}
The secular part of the geopotential expanded in \equ{Rearth} is
obtained as the \red{average} over the fast angles $M$
and $\theta$.

Taking into account the expression of $S_{nmpq}$ in \equ{S}-\equ{psi}, we notice that
the secular terms are those associated to the indexes $m=0$ and $n-2p+q=0$.

For our purposes, it will be convenient to approximate the secular part of \equ{Rearth} as (\cite{CGmajor})
\beqa{Rsec}
\red{\mathcal{H}}_{Earth}^{sec}&\cong&\frac{\mu_E R^2_E J_{2}}{a^3} \Bigl(\frac{3}{4} \sin^2 i -\frac{1}{2}\Bigr) (1-e^2)^{-3/2} \nonumber\\
&+&\frac{2\mu_E R^3_E J_{3}}{a^4} \Bigl(\frac{15}{16} \sin^3 i -\frac{3}{4} \sin i\Bigr) e (1-e^2)^{-5/2} \sin \omega \nonumber \\
&+&\frac{\mu_E R^4_E J_{4}}{a^5} \Bigl[\Bigl(-\frac{35}{32} \sin^4 i +\frac{15}{16} \sin^2 i\Bigr) \frac{3e^2}{2}(1-e^2)^{-7/2} \cos(2\omega) \nonumber \\
&+&
\Bigl(\frac{105}{64} \sin^4 i -\frac{15}{8} \sin^2 i+\frac{3}{8}\Bigr) (1+\frac{3e^2}{2})(1-e^2)^{-7/2} \Bigr]\ .
\eeqa
For the Earth, it turns out that $J_2\gg J_n$ for all $n \in \mathbb{N}$, $n>2$; hence,
the most important harmonic of the secular Hamiltonian is that corresponding to $J_2$.
When the expansion of $\mathcal{H}_{Earth}^{sec}$ is limited to the $J_2$ term, we will use the terminology
of \sl quadrupolar \rm approximation of the Hamiltonian.

As for the resonant part of the geopotential,
given the expressions \equ{S}-\equ{psi} for the quantity $S_{nmpq}$,
the terms corresponding to a tesseral resonance of order $j:\ell$
are \red{those containing the angle $\sigma_{j\ell}$ as in \equ{sigma}
with $\ell=n-2p+q$ and $m=j$.}
The resonant part of the geopotential will then be a sum of terms with resonant arguments.

\subsubsection{Solar and Lunar disturbing functions} \label{sec:SunMoon}

The expressions for the Solar and Lunar disturbing functions are obtained as follows.
We assume that the Solar elements, say $(a_S,e_S,i_S,M_S,\omega_S,\Omega_S)$, are referred
to the equatorial frame (\cite{Kaula1962}). We also assume that the Sun moves on a Keplerian ellipse with
semimajor axis $a_S=1\, AU$, eccentricity $e_S=0.0167$, inclination $i_S=23^{\circ} 26' 21.406''$,
argument of perigee $\omega_S=282.94^{\circ}$,
longitude of the ascending node $\Omega_S=0^{\circ}$; the rate of variation of the mean anomaly is
assumed to be $\dot M_S\simeq 1^{\circ}/day$. Then, denoting by $m_S$ the mass of the Sun,
the disturbing function due to the Sun can be written as
\beqa{RSUN}
\mathcal{R}_{Sun}&=&\mathcal{G} m_S\sum_{l=2}^{\infty}\sum_{m=0}^l \sum_{p=0}^l \sum_{h=0}^l \sum_{q=-\infty}^\infty \sum_{j=-\infty}^\infty {a^l\over a_S^{l+1}}
\ \epsilon_m\, {{(l-m)!}\over {(l+m)!}}\nonumber\\
&&\times\ \F_{lmph}(i,i_S) \mathcal{H}_{lpq}(e)\, \mathcal{G}_{lhj}(e_S)\ \cos(\varphi_{lmphqj})\ ,
\eeqa
where
\beqano
\F_{lmph}(i,i_S)&\equiv&F_{lmp}(i)\ F_{lmh}(i_S)\ ,\nonumber\\
\varphi_{lmphqj}&\equiv& (l-2p)\omega+(l-2p+q)M-(l-2h)\omega_S-(l-2h +j)M_S+m(\Omega-\Omega_S)\ ;
\eeqano
the quantities $\epsilon_m$ take the values
$$
    \epsilon_m = \left\{ \begin{array}{cl} 1 & \text{if } m = 0\ , \\ 2 & \text{if } m \in\integer\backslash\{0\}\ , \end{array} \right.
$$
the functions $ \mathcal{H}_{lpq}(e)$ and $\mathcal{G}_{lhj}(e_S)$
correspond to the Hansen coefficients $X_{l-2p+q}^{l,l-2p}(e)$,
$X_{l-2h+j}^{-(l+1),l-2h}(e_S)$, while
$F_{lmp}(i)$ and $F_{lmh}(i_S)$ are as in \equ{Ffun}.\\

The contribution due to the Moon is conveniently formulated taking the elements
of the satellite with respect to the equator and the elements of the Moon with respect to the ecliptic plane
(\cite{Lane1989}). With this choice, the inclination $i_M$ of the Moon
is nearly constant, the rates of variation of the argument of perihelion
$\omega_M$ of the Moon and that of the longitude of the ascending node $\Omega_M$ are nearly linear.

We assume that the Moon moves on a Keplerian ellipse with semimajor axis $a_M=384\,748\, km$, eccentricity
$e_M=0.0549$ and inclination $i_M=5^{\circ}15'$. Let $m_M$ be the mass of the Moon.
Then, the Lunar disturbing function takes the form (\cite{CGPR2017}, \cite{Lane1989}):
\beqa{RMOON}
    \mathcal{R}_{Moon}
    \nonumber
    & = &\mathcal{G} m_M\sum\limits_{l \geq 2} \sum\limits_{m = 0}^l \sum\limits_{p = 0}^l \sum\limits_{s = 0}^l
        \sum\limits_{q = 0}^l \sum\limits_{j = -\infty}^{+\infty} \sum\limits_{r = -\infty}^{+\infty}
        (-1)^{m+s}\ (-1)^{k_1} \frac{ \epsilon_m \epsilon_s}{2 a_M} \frac{(l - s)!}{(l + m)!}
        \left( \frac{a}{a_M} \right)^l \\
    \nonumber
    &\times& F_{lmp} (i) F_{lsq} (i_M)
        \mathcal{H}_{lpj} (e) \mathcal{G}_{lqr} (e_M) \\
    \label{eq:lane}
    &\times& \left\{ (-1)^{k_2} U_l^{m, -s}
        \cos \left( \bar{\theta}_{lmpj} + \bar{\theta}_{lsqr}^\prime - y_s \pi \right)
        + (-1)^{k_3} U_l^{m, s} \cos \left( \bar{\theta}_{lmpj} - \bar{\theta}_{lsqr}^\prime - y_s \pi \right)  \right\}\ ,\nonumber\\
\eeqa
where $y_s=0$ for $s$ even, $y_s=1/2$ for $s$ odd, $k_1=[m/2]$, $k_2 = t (m + s - 1) + 1$, $k_3 = t (m + s)$ with $t=(l-1)$ mod 2,
the terms $\bar{\theta}_{lmpj}$, $\bar{\theta}_{lsqr}^\prime$ are given by
\beqano
    \bar{\theta}_{lmpj} & = & (l - 2p) \omega + (l - 2p + j) M + m \Omega\ , \nonumber\\
    \bar{\theta}_{lsqr}^\prime & = & (l - 2q) \omega_M + (l - 2q + r) M_M + s (\Omega_M - \pi/2)\ ,
\eeqano
the functions $U_l^{m,s}$ have the following expressions (compare with \cite{CGPR2017})
\beqano
U_l^{m,s}&=&\sum_{r =\max(0,-(m+s))}^{\min(l-s,l-m)} (-1)^{l-m-r}
\left(\begin{array}{c}
  l+m \\
  m+s+r \\
\end{array}\right)\
\left(\begin{array}{c}
  l-m \\
  r \\
\end{array}\right)\
\cos^{m+s+2r}({\varepsilon\over 2})\sin^{-m-s+2(l-r)}({\varepsilon\over 2})\ ,\nonumber\\
\eeqano
where $z=\cos^2 ({\varepsilon\over 2})$ and $\varepsilon$ denotes the obliquity of the ecliptic.
The functions $\mathcal{H}_{lpj}(e)$ and $\mathcal{G}_{lqr}(e_S)$ represent the Hansen coefficients
$X_{l-2p+j}^{l,l-2p}(e)$, $X_{l-2q+r}^{-(l+1),l-2q}(e_M)$.

\vskip.1in

Given the expansions above of the Solar and Lunar perturbing functions, we are led
to introduce the following definitions of Solar and Lunar secular resonances.

\begin{definition} \label{def:secres}
A Solar gravity secular resonance occurs whenever there exists
an integer vector $(k_1,k_2,k_3,k_4)\in\integer^4\backslash\{0\}$, such that
\beq{secressun}
k_1\dot\omega+k_2\dot\Omega+k_3\dot\omega_S+k_4\dot\Omega_S=0\ .
\eeq
A Lunar gravity secular resonance occurs whenever there exists an integer vector
$(k_1,k_2,k_3,k_4)\in\integer^4\backslash\{0\}$, such that
\beq{secresmoon}
k_1\dot\omega+k_2\dot\Omega+k_3\dot\omega_M+k_4\dot\Omega_M=0\ .
\eeq
\end{definition}

We stress that the above definition of secular resonances is as general as possible. However, given the fact that
we will consider Solar
and Lunar expansions truncated to the second order in the ratio of the semi--major axes,
in view of \eqref{RSUN} and \eqref{RMOON}, the Hamiltonian $\H$ is independent of $\omega_M$ and $\omega_S$.
Therefore, for all resonances studied here, one has $k_3=0$. Moreover, since $\dot{\Omega}_S \simeq 0$,
the relations \eqref{secressun} and \eqref{secresmoon} may be rewritten in the particular form:
$$
(2-2p)\dot\omega+m\dot\Omega=0\ , \qquad m,p=0,1,2\,,
$$
and
$$
(2-2p)\dot\omega+m\dot\Omega+\kappa \dot\Omega_M=0\ , \qquad m,p=0,1,2, \quad \kappa=-2,-1,0,1,2\,,
$$
respectively.

\vskip.1in

According to the classification of the harmonic terms of the expansions
\eqref{RSUN} and \eqref{RMOON}, we define the Solar and Lunar semi--secular resonances as follows (compare with \cite{HughesI}).

\begin{definition}
A Solar semi--secular resonance occurs whenever
$$
(l-2p)\dot\omega+m\dot\Omega-(l-2h+j) \dot{M}_S=0\ , \qquad l \in \mathbb{Z}_+\,,\ m,p,h=0,1,2,...,l\,,\ j \in \mathbb{Z}.
$$
We have a Lunar semi--secular resonance whenever
\beqano
(l-2p)\dot\omega+m\dot\Omega \pm [(l-2q) \dot{\omega}_M+(l-2q+r) \dot{M}_M+s\dot{\Omega}_M]=0\ ,\nonumber\\
\qquad\qquad\qquad\qquad\qquad\qquad\qquad\qquad\ l \in \mathbb{Z}_+\,,\ m,p,q,s=0,1,2,...,l\,,\ r \in \mathbb{Z}.
\eeqano
\end{definition}

By taking a quadrupolar approximation of the expansions \eqref{RSUN} and \eqref{RMOON}, namely considering $l=2$,
it follows that the possible resonances have the form:
$$
\alpha \dot\omega+ \beta \dot\Omega-\gamma \dot{M}_S=0\ , \qquad \alpha \in \{\pm 2, 0\}\,, \quad\beta \in \{\pm 2, \pm 1, 0\}\,, \quad
\gamma \in \mathbb{Z}\backslash\{0\}
$$
for the Sun and
\begin{equation}\label{lunar2_semi_secular_res}
\begin{split}
& \alpha \dot\omega+ \beta \dot\Omega +\alpha_M \dot\omega_M+ \beta_M \dot\Omega_M-\gamma \dot{M}_M=0\ , \qquad \alpha\,, \alpha_M \in \{\pm 2, 0\}\,,\nonumber\\
& \hspace{7cm} \beta, \beta_M \in \{\pm 2, \pm 1, 0\}, \quad \gamma \in \mathbb{Z}\backslash\{0\}\nonumber
\end{split}
\end{equation}
for the Moon. We remark that Lunar semi--secular resonances occur
at very low altitudes and therefore their interest is limited.

\section{Quadrupolar approximation}\label{sec:quadrupolar}

The quadrupolar approximation is obtained by summing the Keplerian part $-\mu_E^2/(2L^2)$ and the
secular part (namely, the term of the series \eqref{RE} for which $n=2$, $m=0$, $p=1$, $q=0$).
Since  $F_{201}(i)=0.75 \sin^2 i -0.5$ and $G_{210}(e)=(1-e^2)^{-3/2}$ (see \eqref{Ffun} and \eqref{Gfun}),
from \eqref{Rsec} we obtain the Hamiltonian
\beq{quad}
\H_{Kepler+J_2}=-\frac{\mu_E^2}{2 L^2}+{{R_E^2 J_2 \mu_E^4}\over {4}}\ {{1}\over {L^3G^3}}\ \Bigl(1-3{H^2\over G^2}\Bigr)\ .
\eeq
Since the angles in $\H_{Kepler+J_2}$ are ignorable, it follows that $L$, $G$ and $H$ are constant, while the Delaunay
angle variables $M$, $\omega$ and $\Omega$ \red{evolve} linearly in time with rates \red{of variation}:
\beqa{elements}
\dot{M}  &=&  \frac{\mu_E^2}{L^3}-\frac{3 R_E^2 J_2 \mu_E^4}{4} \frac{1}{L^4 G^3} \Bigl(1-3 \frac{H^2}{G^2}\Bigr)\,,\nonumber\\
\dot{\omega} &=&  \frac{3 R_E^2 J_2 \mu_E^4}{4} \frac{1}{L^3 G^4} \Bigl(-1+5 \frac{H^2}{G^2}\Bigr)\,,\nonumber\\
\dot{\Omega} &=& - \frac{3 R_E^2 J_2 \mu_E^4}{2} \frac{H}{L^3 G^5}\ .
\eeqa
By using the following data: $m_E=5.972\cdot 10^{24}$ $kg$, $\G=6.67408\cdot 10^{-11}$ $m^3/(kg\, s^2)$,
$R_E=6378.137\cdot 10^3$ $m$, $J_2=1082.6261\cdot 10^{-6}$, and using the
conversion factor $60\cdot 60\cdot 24\cdot 360/(2\pi)$, one obtains that
$$
{{(\G m_E)^{1\over 2}}\over {R_E^{3\over 2}}}=6135.7^o/day\ ,\qquad
{{3J_2(\G m_E)^{1\over 2}}\over {4R_E^{3\over 2}}}=4.98^o/day\ .
$$
Hence, we can \red{approximate} the
relations \equ{elements} in terms of the orbital elements as:
\begin{equation}
\begin{split} \label{MomeagaOmega_var}
& \dot{M}  \simeq 6135.7 \Bigl(\frac{R_E}{a}\Bigr)^{3/2}- 4.98 \Bigl(\frac{R_E}{a}\Bigr)^{7/2} (1-e^2)^{-3/2} (1-3 \cos^2 i) \ ^o/day\,,\\
& \dot{\omega} \simeq  4.98 \Bigl(\frac{R_E}{a}\Bigr)^{7/2} (1-e^2)^{-2} (5 \cos^2 i-1) \ ^o/day\,,\\
& \dot{\Omega} \simeq - 9.97 \Bigl(\frac{R_E}{a}\Bigr)^{7/2} (1-e^2)^{-2}  \cos i \ ^o/day\,. \\
\end{split}
\end{equation}
Therefore, we are led to summarize as follows the main effects of $J_2$: a slow change of the rate of the mean anomaly, a precession of the perigee and a secular regression of the orbital node.


\red{Our results can be generalized to higher orders in a straightforward way:} let us write the quadrupolar Hamiltonian in compact form as
$$
\red{\mathcal{H}}_Q=-{{\mu_E^2}\over {2L^2}}+J_2\ F_2(L,G,H)
$$
for a suitable function $F_2$ whose expression can be obtained from \equ{quad}. If we include other
spherical harmonic coefficients, e.g. $J_3$ and $J_4$, pertaining to the secular part of the Hamiltonian, we are led to
add to $\red{\mathcal{H}}_Q$ terms of the form
$$
\red{\mathcal{H}}_{add}=J_3\ F_3(L,G,H,\omega)+J_4\ F_4(L,G,H,\omega)
$$
for suitable functions $F_3$ and $F_4$.
Hence, for the new system the rate of variation of the mean anomaly is given by
$$
\dot M={{\partial (\red{\mathcal{H}}_Q+\red{\mathcal{H}}_{add})}\over {\partial L}}\ ,
$$
which makes $\dot M$ depending also on $\omega$. This might generate a superposition of
resonances and, therefore, a characterization of chaos in a model which is not just the quadrupolar
approximation.

\section{A characterization of tesseral resonances}\label{sec:tesseral}

\red{
In this Section we study tesseral resonances, which involve the rates of variation of the mean anomaly of the object, the
rate of the sidereal time and the rates of variation of the argument of perigee and the longitude of the ascending
node of the object. The main results of this Section are the following: in Proposition~\ref{ref:pro1} we fix $a$, $e$ and find that the resonance relation gives
an expression for the inclination; in Proposition~\ref{ref:pro2} we fix $e$, $i$ and find an equation for the
semimajor axis; in Proposition~\ref{ref:pro3} we fix $a$, $i$ and find an expression for the eccentricity.
}

Within the quadrupolar approximation, using \equ{MomeagaOmega_var}
a tesseral resonance of order $j:\ell$, $j,\ell\in\integer\backslash\{0\}$, occurs whenever
the semimajor axis, eccentricity and inclination satisfy the following relation:
\red{
\beqa{quatess}
&&\Big\{\ell\ \Big[6135.7 ({R_E\over a})^{3\over 2}-4.98 ({R_E\over a})^{7\over 2} (1-e^2)^{-{3\over 2}}(1-3\cos^2 i)
+4.98 ({R_E\over a})^{7\over 2} (1-e^2)^{-2}(5\cos^2 i-1)\Big] \nonumber \\
&&-j\ 9.97({R_E\over a})^{7\over 2} (1-e^2)^{-2}\red{\cos i}\Big\}
\ {1\over {360}}\ {365.242196\over {366.242196}}=j\ ,\nonumber\\
\eeqa where the factor $365.242196/366.242196$ is introduced to
transform from mean Solar to sidereal days}.

\begin{remark}
We notice that if $J_2=0$, then $\dot\omega=\dot\Omega=0$ and the first in
\equ{elements} reduces to
$$
\dot M=6135.7\ ({R_E\over a})^{3\over 2}\ .
$$
Hence, tesseral resonances can occur for any inclination and eccentricity, provided
the semimajor axis satisfies the relation
$$
6135.7\ ({R_E\over a})^{3\over 2}\ {1\over {360}}\ {{365.242196}\over {366.242196}}={j\over \ell}\ .
$$
Therefore, the effect of $J_2$ is to introduce a dependence of the tesseral resonances on eccentricity
and inclination as well as a precession of $\omega$ and $\Omega$.
\end{remark}

The following result provides the condition under which the inclination satisfies a tesseral resonance
of order $j:\ell$, whenever the semimajor axis and eccentricity are fixed.

\begin{proposition}\label{ref:pro1}
Within the quadrupolar approximation \equ{quad}, for a fixed value of $a$ and $e$, denoting by $A$, $B$,
\red{$C$}, $\tau$
the quantities
\red{
\beqano
A&=&6135.7({R_E\over a})^{3\over 2}-4.98 ({R_E\over a})^{7\over 2} (1-e^2)^{-{3\over 2}}
-4.98 ({R_E\over a})^{7\over 2} (1-e^2)^{-2}\ ,\nonumber\\
B&=& 14.94({R_E\over a})^{7\over 2} (1-e^2)^{-{3\over 2}}
+24.90({R_E\over a})^{7\over 2} (1-e^2)^{-2}\ ,\nonumber\\
C&=&9.97 ({R_E\over a})^{7\over 2} (1-e^2)^{-2}\ ,\nonumber\\
\tau&=&{{365.242196}\over {366.242196}}\ ,
\eeqano
}
a tesseral resonance
of order $j:\ell$ occurs for values of the inclination such that
\red{
$$
        \cos i=\frac{C\ j \pm \sqrt{C^2\ j^2-4B\ell\left(\ell A - 360\ j/\tau\right)}}{2B\ell}\ ,
$$
}
provided the following conditions are satisfied:
\red{
\beqano
&&C^2\ j^2-4B\ell\left(\ell A - 360\ j/\tau\right)\ge 0 \ , \nonumber \\
&&\Bigl|\frac{C\ j \pm \sqrt{C^2\ j^2-4B\ell\left(\ell A - 360\ j/\tau\right)}}
        {2B\ell}\Bigr| \leq 1\ .
\eeqano
}
\end{proposition}

The proof of Proposition~\ref{ref:pro1}, as well as those of Propositions~\ref{ref:pro2} and \ref{ref:pro3} below,
is elementary, since it suffices to insert in Definition~\ref{def:resonance} the value for $\dot M$
in \equ{MomeagaOmega_var}.

\red{If we} fix the eccentricity and the inclination, we obtain the following solutions for the semimajor axis.

\begin{proposition}\label{ref:pro2}
Within the quadrupolar approximation \equ{quad}, for a fixed value of $e$ and $i$, denoting by
$A$, $B$, \red{$C$}, $\rho$
the quantities
\red{
\beqano
A&=&6135.7\ ,\nonumber\\
        B&=& 4.98(1-e^2)^{-{3\over 2}}\ (1-3\cos^2i) \red{-}4.98(1-e^2)^{-2}\ (5\cos^2i-1)\ ,\nonumber\\
        C&=& 9.97(1-e^2)^{-2}\red{\cos{\left(i\right)}}\ ,\nonumber\\
\rho&=&\Bigl({R_E\over a}\Bigr)^{1\over 2} \ , \nonumber\\
\red{\tau}&=&\red{{{365.242196}\over {366.242196}}}\ ,
\eeqano
}
a tesseral resonance of order $j:\ell$ occurs for values of the semimajor axis such that
\red{
$$
\ell\ A\rho^3-(\ell B+j C)\rho^7={360\over \tau} j\ .
$$
}
\end{proposition}

Finally, if we fix the semimajor axis and the inclination, we obtain the following result for the eccentricity.

\begin{proposition}\label{ref:pro3}
Within the quadrupolar approximation \equ{quad}, for a fixed value of $a$ and $i$, denoting by
\red{$\varepsilon$}, $A$, $B$, \red{$C$}, \red{$D$}, the quantities
\red{
\beqano
        \red{\varepsilon}&=&(1-e^2)^{-1}\ ,\nonumber\\
A&=&6135.7({R_E\over a})^{3\over 2}\ ,\nonumber\\
B&=& 4.98({R_E\over a})^{7\over 2}\ (1-3\cos^2i)\ ,\nonumber\\
C&=& 4.98({R_E\over a})^{7\over 2}\ (5\cos^2i-1)\ ,\nonumber\\
D&=& 9.97({R_E\over a})^{7\over 2}\red{\cos i }\ , \nonumber\\
\red{\tau}&=&\red{{{365.242196}\over {366.242196}}}\ ,
\eeqano
}
a tesseral resonance of order $j:\ell$ occurs for values of the eccentricity \red{(namely of $\varepsilon$)
which are solutions of the following equation:}

\red{
$$
(jD-\ell C)\varepsilon^2+\ell B\varepsilon^{3\over 2}+({{360}\over \tau}j-\ell A)=0\ .
$$
}
\end{proposition}


The dependency of the location of $j:\ell$ resonances on the
orbital parameters is shown in Figure~\ref{f:tes}. The grey shaded
regions provide the values in which \equ{quatess} can be solved,
contours are marked for different values of orbital parameters as
shown in the plot legends \red{at the top}. The value of the
semi-major axis $a$ weakly depends on the choice of $e$ and $i$.
In Figure~\ref{f:tes} the values $a_{geo}$ and $a_0$ denote the
nominal values of the semimajor axes corresponding, respectively,
to the 1:1 and 5:1 resonances, when the oblateness, the
eccentricity, and the inclination are set to zero. As
Figure~\ref{f:tes} shows, the deviations from the nominal values
are less than $10$ km, and they are higher when increasing the
order of the resonance.

\begin{figure}
\centering
\includegraphics[width=0.45\textwidth]{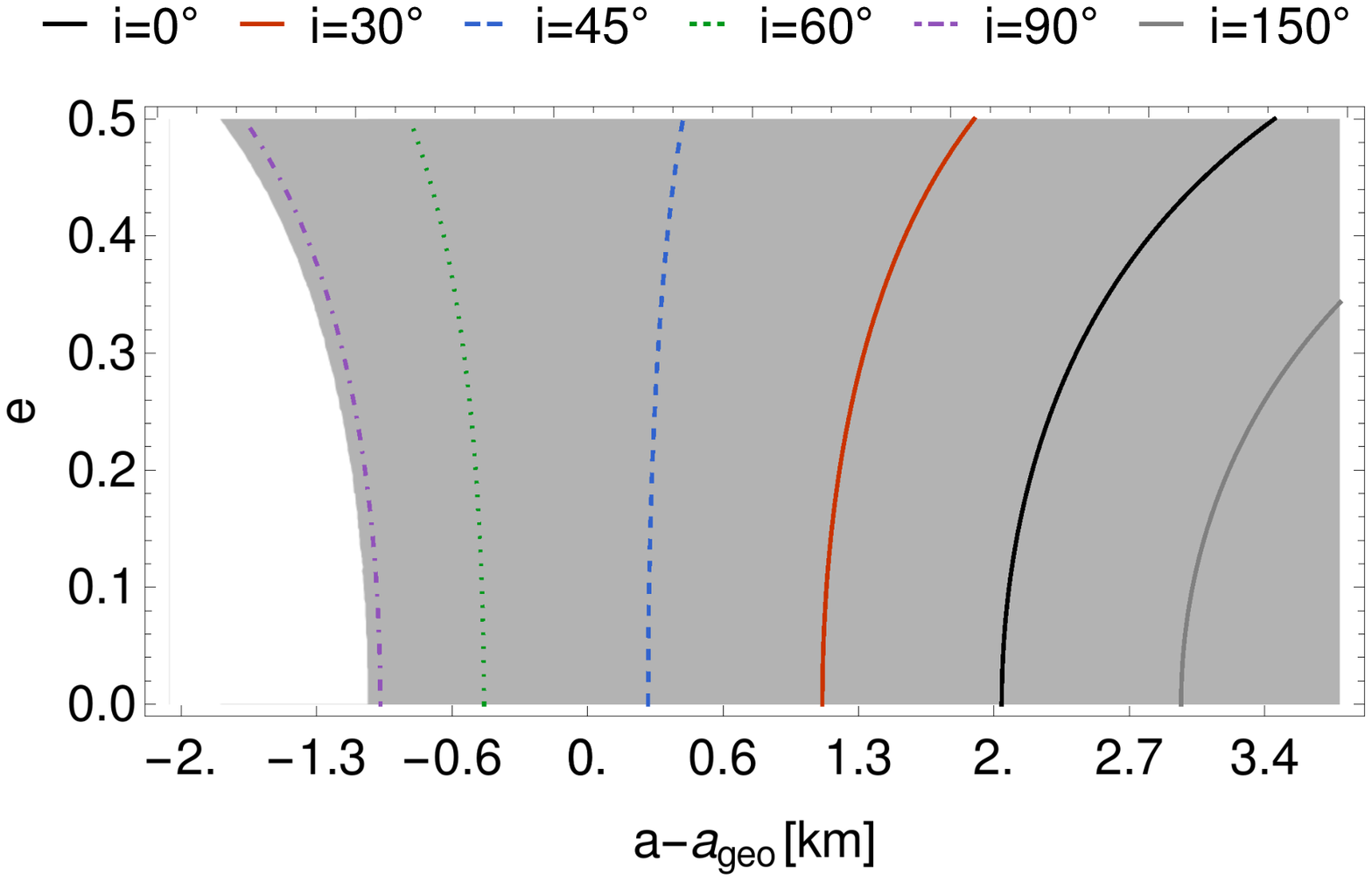} \hspace{0.1cm}
\includegraphics[width=0.45\textwidth]{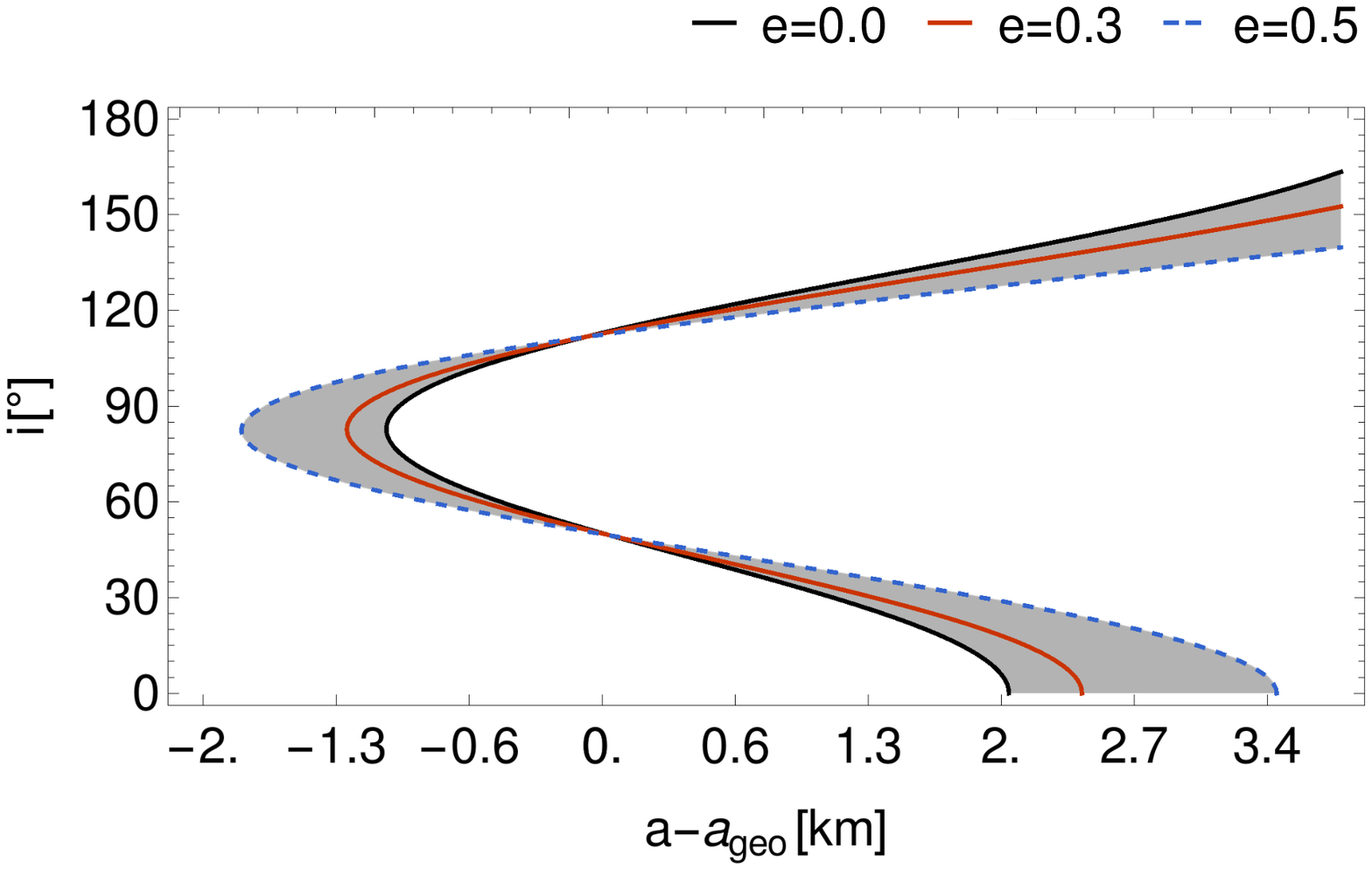}

\vspace{0.1cm}

\includegraphics[width=0.45\textwidth]{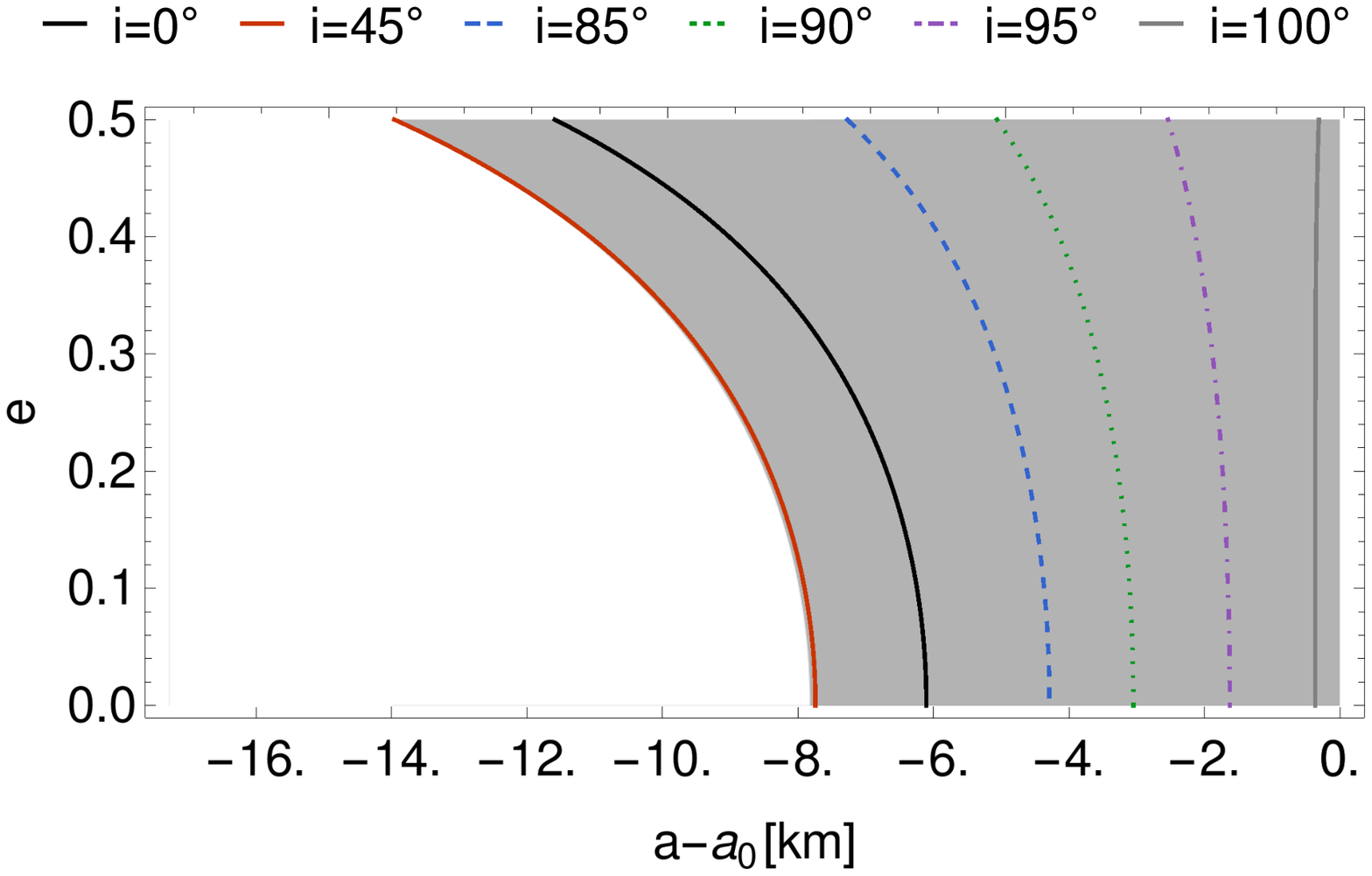} \hspace{0.1cm}
\includegraphics[width=0.45\textwidth]{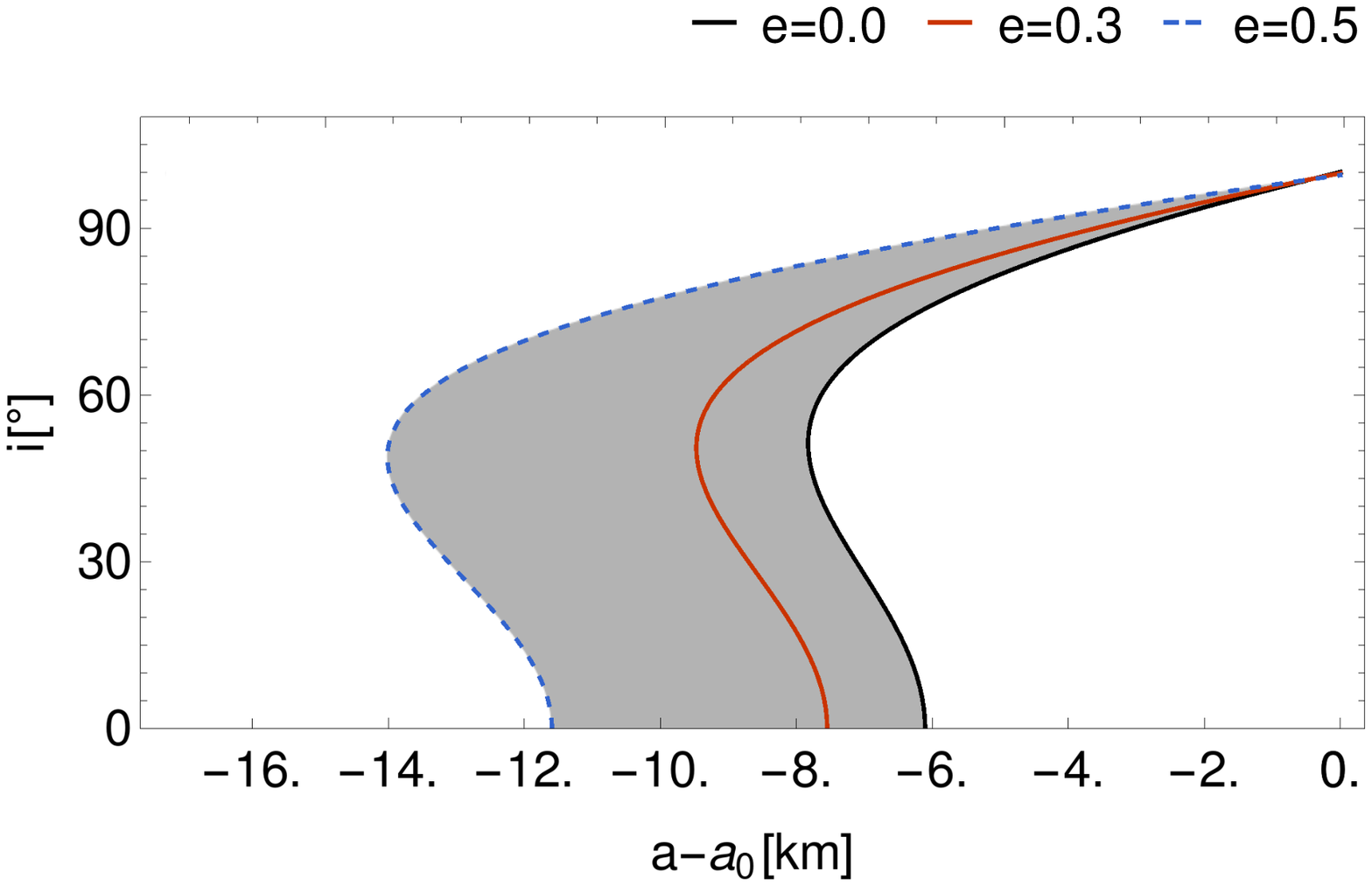}
\caption{Location of tesseral resonances in $(a,e)$-space (left) and
$(a,i)$-space (right) for the $1:1$ (top) and the $5:1$ resonances (bottom).
Values on abscissae are deviations from the reference value (defined as $J_2=e=i=0$)
given in $km$.}
\label{f:tes}
\end{figure}

The dependency of the location of the resonance on the $J_2$
gravity field expansion coefficient is shown in Table~\ref{t:tes}.
The influence of the flattening of the Earth on the resonant value
of the semi-major axis $a$ ranges from less than 1 km (resonances
$3:4$) up to a few kilometers, e.g. for the resonance $5:1$.

\begin{table}
$$
\begin{array}{|c|c|r|r|c|}
\hline
\hline
j & \ell & a_{J_2}[km] & a[km] & a_{J_2}-a [km]  \\
\hline
 3 & 4 & 51 \ 079.116 & 51 \ 078.254 & 0.860 \\
 4 & 5 & 48 \ 928.085 & 48 \ 927.185 & 0.900 \\
 1 & 1 & 42 \ 165.214 & 42 \ 164.170 & 1.044 \\
 5 & 4 & 36 \ 337.192 & 36 \ 335.980 & 1.212 \\
 4 & 3 & 34 \ 807.020 & 34 \ 805.755 & 1.265 \\
 3 & 2 & 32 \ 178.652 & 32 \ 177.284 & 1.369 \\
 5 & 3 & 29 \ 996.159 & 29 \ 994.691 & 1.468 \\
 2 & 1 & 26 \ 563.420 & 26 \ 561.762 & 1.658 \\
 5 & 2 & 22 \ 892.157 & 22 \ 890.233 & 1.924 \\
 3 & 1 & 20 \ 272.591 & 20 \ 270.419 & 2.172 \\
 4 & 1 & 16 \ 735.493 & 16 \ 732.862 & 2.631 \\
 5 & 1 & 14 \ 422.996 & 14 \ 419.943 & 3.053 \\
\hline
\end{array}
$$
\caption{Values of $j:\ell$ tesseral resonances with ($a_{J_2}$)
and without ($a$) the dependency on $J_2$.} \label{t:tes}
\end{table}

\begin{figure}[h]
\centering
\hglue0.1cm
\includegraphics[width=0.49\textwidth]{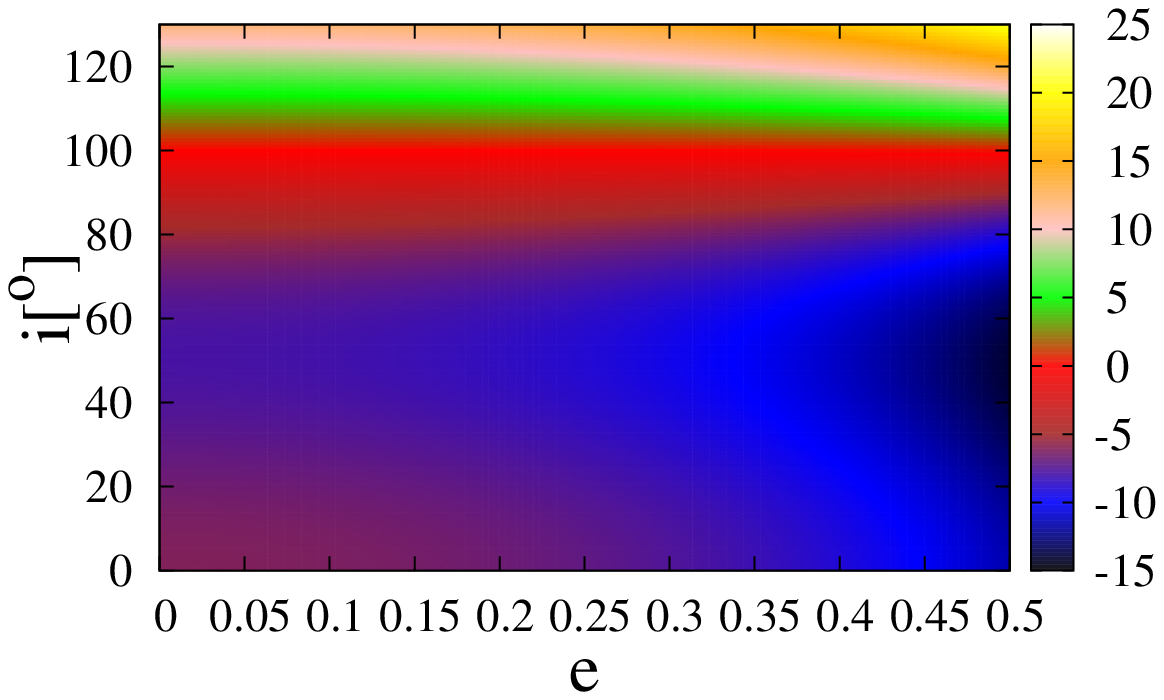}
\includegraphics[width=0.49\textwidth]{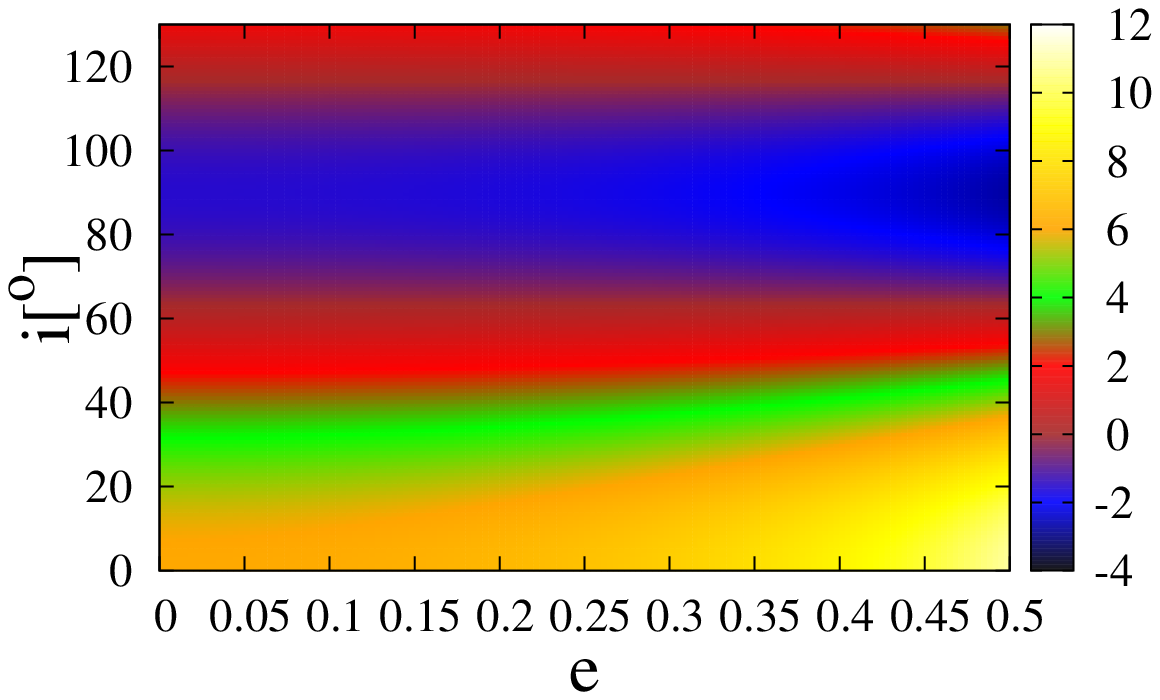}
\vglue0.5cm
\caption{Resonance 5:1. Shift in semi--major axis (color bar in
km) of the location of the equilibria as a function of $e$ and
$i$, due to the effect of $J_2$.
 {\it Left}: Shift of the location of the exact resonance
 $\dot{\sigma}_{51}=0$ from the nominal value $a_{5:1}=14 \, 419.9\ $ km
 (see also the Table~\ref{t:tes}). {\it Right}: Shift of the
location of the exact resonance of a harmonic term whose argument
is $\sigma_{51}+(k+1) \omega $, $k \in \mathbb{Z}$, from the
location of the resonance associated to another harmonic term
whose argument is $\sigma_{51}+ k \omega $.} \label{shift}
\end{figure}

\begin{figure}
\centering
\includegraphics[width=0.49\textwidth]{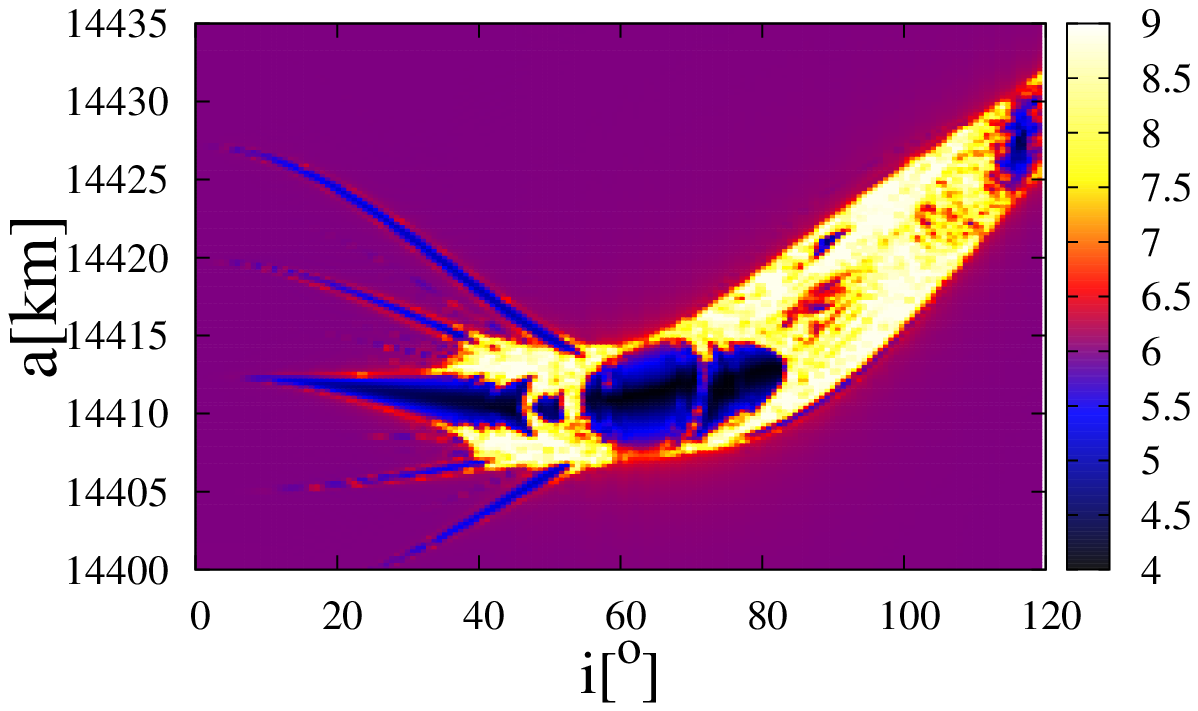}
\includegraphics[width=0.49\textwidth]{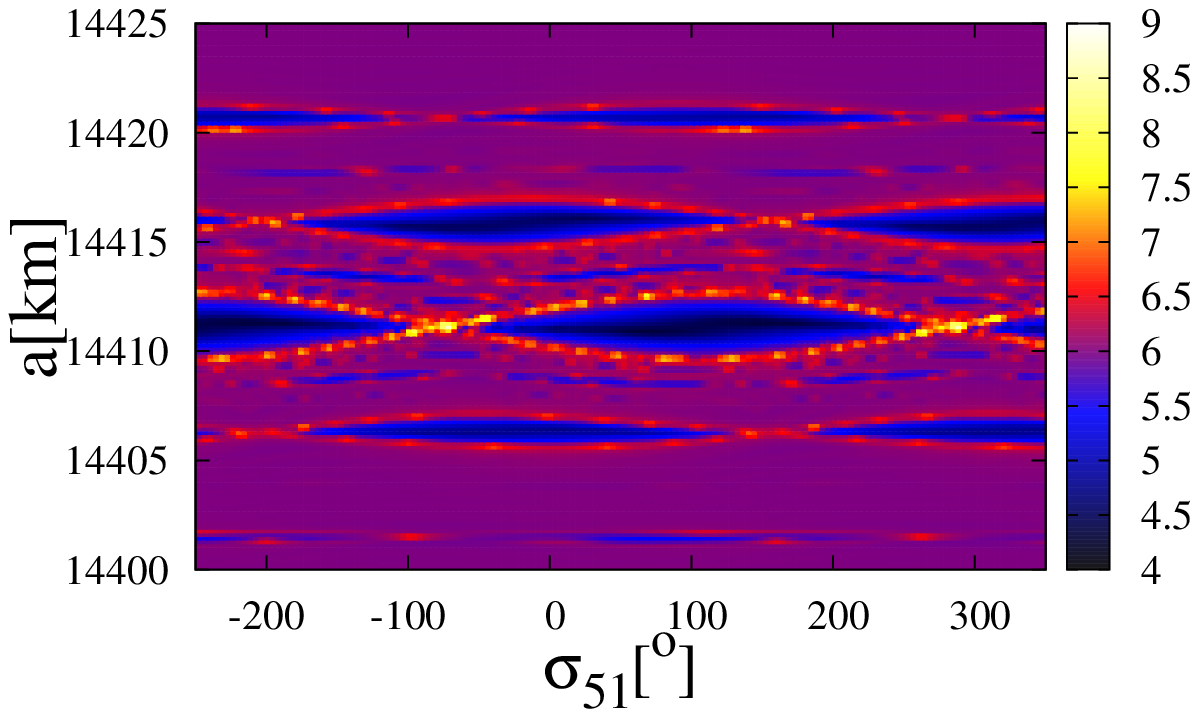}
\includegraphics[width=0.49\textwidth]{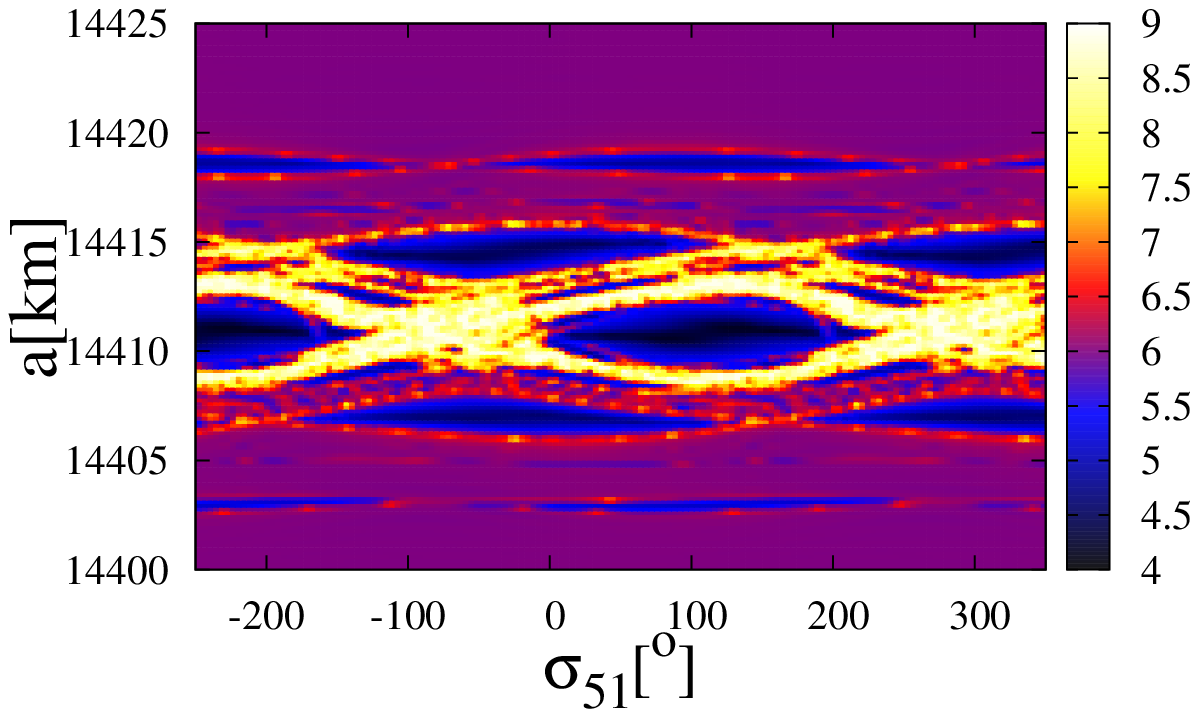}
\includegraphics[width=0.49\textwidth]{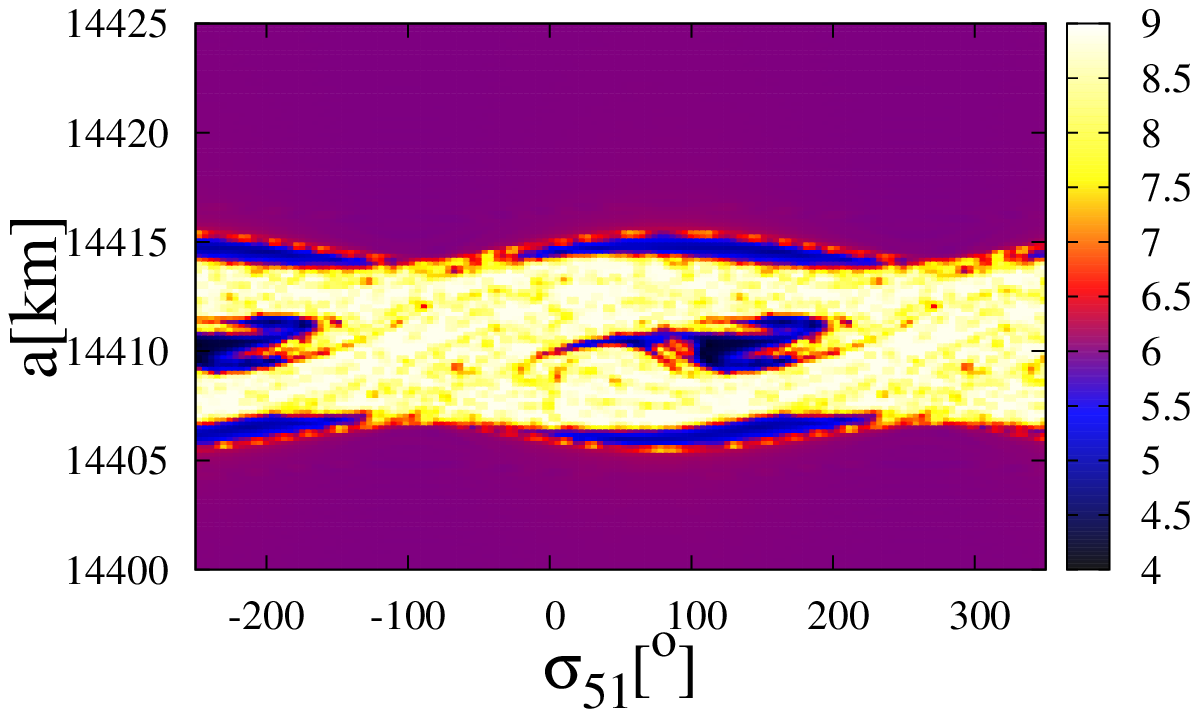}
\includegraphics[width=0.49\textwidth]{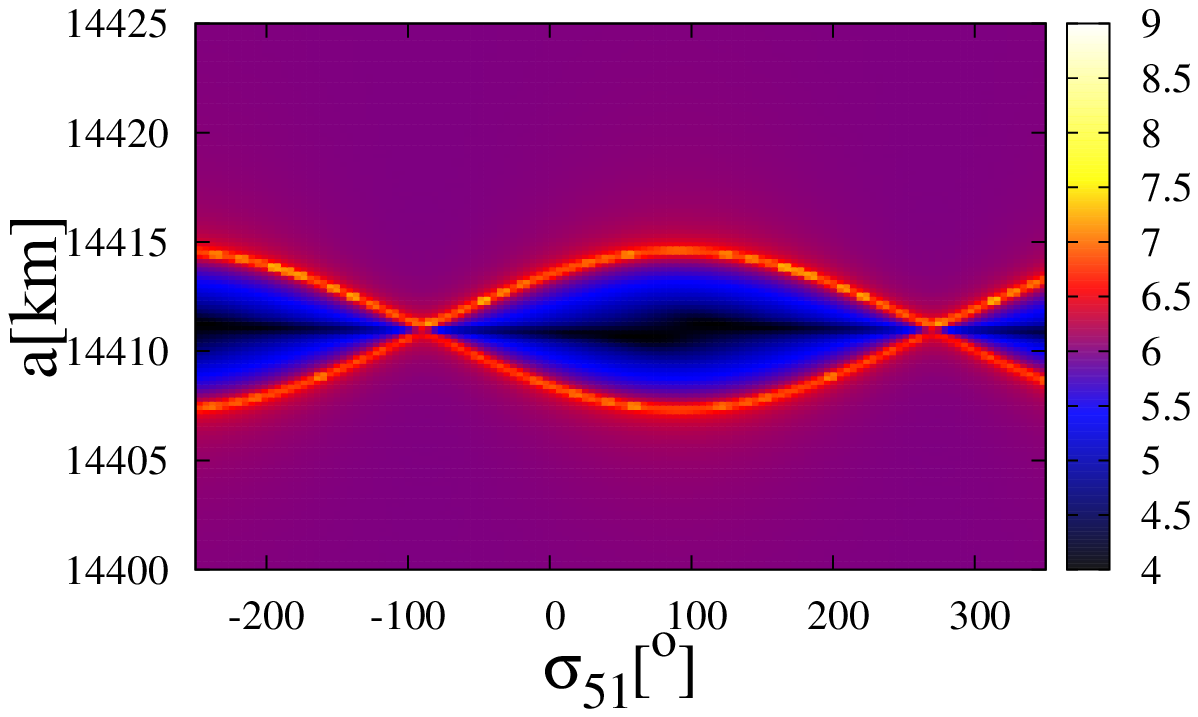}
\includegraphics[width=0.49\textwidth]{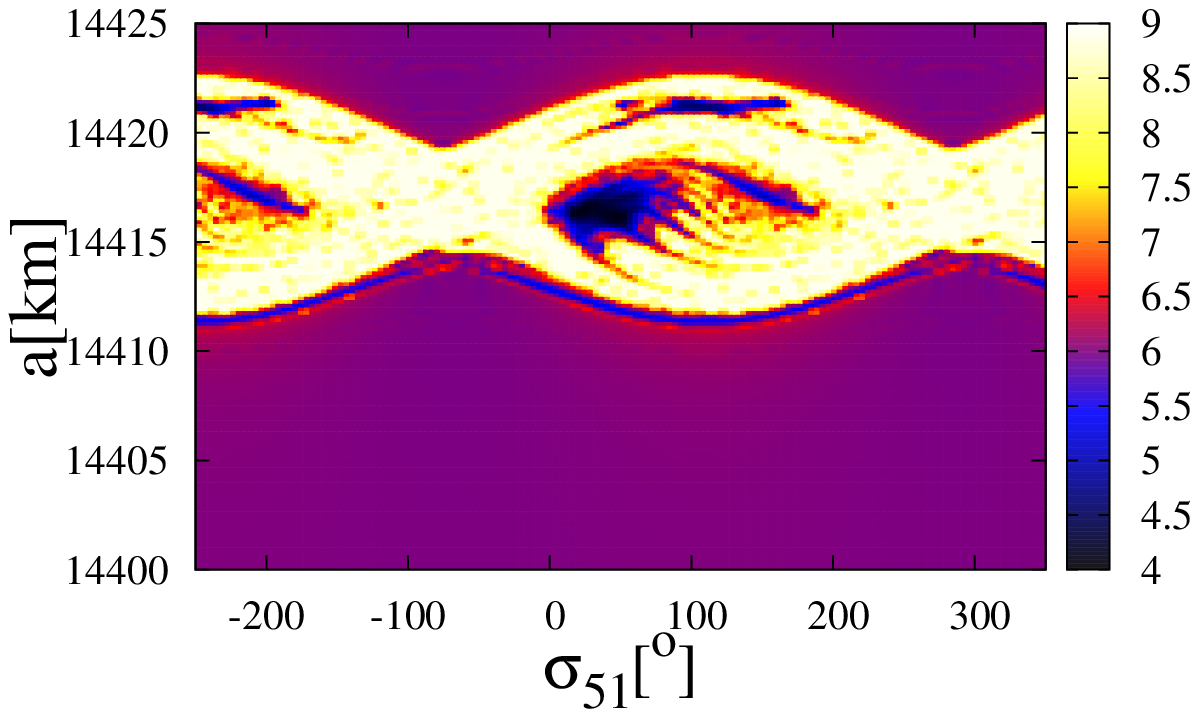}
\vglue0.8cm
\caption{FLI for the 5:1 tesseral resonance as a function of
semi-major axis and inclination (top left) and respectively
semi-major axis and resonant angle (all other plots). The initial
values are $e=0.3$, $\omega=0^o$, $\Omega=0^o$ for all panels. The
top left panel is obtained for the initial resonant angle
$\sigma_{51}=105^o$, while the plots in the
        $(a-\sigma_{51})$--plane are computed for $i=32^o$ (top \red{right}),
        $i=38^o$ (\red{middle left}), $i=50^o$ (\red{middle right}),
        $i=63.4^o$ (\red{bottom left}) and $i=90^o$ (bottom right). } \label{f:tescart}
\end{figure}

Due to the effect of the secular part\footnote{Since the
coefficient $J_2$  is much larger than any other zonal harmonic
coefficient, the secular part is dominated essentially by the
$J_2$ harmonic terms. From a numerical viewpoint, it is enough to
consider just the influence of the $J_2$ harmonic terms in order
to catch the main effects of the secular part.}, the frequencies
$\dot{\omega}$ and $\dot{\Omega}$ are not zero (compare with
$\eqref{MomeagaOmega_var}$). As a consequence, as already pointed
out in Definition~\ref{def:multiplet_resonance}, each $j:\ell$
resonance splits into a multiplet of resonances and, hence, each
harmonic term of a specific resonance, with big enough magnitude,
yields equilibria located at different distances from the center
of the Earth. Roughly speaking, the values provided in
Table~\ref{t:tes} give just a hint on the location of the
resonances, including the minor ones, different from 1:1 and 2:1.
Actually, the location and width of the minor resonances, as well
as the regular and chaotic behavior of the corresponding resonant
regions, are affected by the interaction between the secular part
and the resonant harmonic terms. The dynamics of a $j : \ell$
tesseral resonance may be analytically described by estimating the
location of the equilibria of specific components of the multiplet
and the width of the associated resonant islands.

Thus, when studying a tesseral resonance, by using
\eqref{MomeagaOmega_var} we can estimate the location of the exact
resonance for each component of the multiplet
\eqref{multipletcond}. For example, Figure~\ref{shift} provides
the shift of the equilibria from the nominal value $a_{5:1}=14\,419.9\, km$, expressed in
kilometers, as a function of eccentricity and inclination, for the
$5:1$ resonance. Let us \red{recall the definition}
$\sigma_{j\ell}\equiv \ell M-j\theta+\ell\omega+j\Omega$. The left
panel of Figure~\ref{shift} shows the shift of the location of the
equilibria associated to the exact resonance
$\dot{\sigma}_{51}=0$ from the nominal distance $a_{5:1}=14 \,419.943\  km $, given in Table~\ref{t:tes} and obtained by using
Kepler's third law. Negative values are used to show that the
equilibria are closer to the Earth than $a_{5:1}$, while positive
values are used to express that the equilibria are farther from
the Earth than $a_{5:1}$. The right panel gives the shift of the
location of the exact resonance of a harmonic term whose argument
is $\sigma_{51}+(k+1) \omega $, $k \in \mathbb{Z}$, from the
location of the resonance associated to another harmonic term
whose argument is $\sigma_{51}+ k \omega $. We recall that
$\sigma_{51}$ denotes the resonant angle $\sigma_{51}=M-5
\theta+\omega+5 \Omega$. The positive/negative sign in the color bar expresses the fact that the equilibrium
point associated to the term whose argument is $\sigma_{51 }+(k+1)
\omega $ is located farther/closer to the Earth than the
equilibrium point corresponding to the harmonic term whose
argument is $\sigma_{51}+ k \omega $.

By using Figure~\ref{shift} in connection with a simple procedure
for computing the amplitude of the resonant island for a specific
multiplet, as described in \cite{CGminor}, we can give an
analytical explanation concerning the location of the equilibrium
points, an estimate of the amplitude of the resonant island (or
islands) and a description of the dynamics of space debris in the
resonant regions for each tesseral resonance.

Analyzing the right panel of Figure~\ref{shift}, and taking also
into account that the amplitude of the resonant islands is of the
order of few kilometers, we can infer that for small inclinations
a splitting phenomenon takes place, namely the width of the
resonance associated to each component of the multiplet is smaller
than the distance separating these resonances. On the contrary, for
larger inclinations we have an opposite phenomenon, called
superposition of harmonics, which gives rise to a complex behavior
of the semi-major axis. Indeed, Figure~\ref{f:tescart} shows a
cartography of the 5:1 tesseral resonance in the $(a,i)$--plane
(top left) and in the $(a-\sigma_{51})$--plane (all other plots),
for  eccentricity $e=0.3$  based on \red{the computation of Fast Lyapunov Indicators (hereafter FLIs, see the Appendix for more details on the tools used to
study the cartography of the resonances)}. Thus, for $i<35^o$ it is easy to see
that the amplitude of the resonant islands vary from 0 to 4
kilometers (compare with \cite{CGminor}), while the
minimum distance between the equilibria associated to any two
harmonic terms is larger than 4 kilometers (see the top right
panel of Figure~\ref{shift}). Therefore, the harmonic terms having
large enough magnitude give rise to non--overlapping resonances,
as shown by the top \red{right} panel of Figure~\ref{f:tescart}. The
situation is opposite for larger inclinations, say
$i>35^o$ (with the exception of the value $i=63.4^o$ corresponding
to the critical inclination and discussed below), at least for moderate eccentricities
(that is $e \in [0.15, 0.5])$; the width of the resonant island
associated to the dominant term, which in this case is
$\mathcal{T}_{5520}$, is at least 4 kilometers, while
Figure~\ref{shift} (top right) shows that it is possible to have
harmonic terms $\mathcal{T}_{nmpq}$ whose associated equilibria
are shifted by less than 4 kilometers. If these terms are
comparable in magnitude with the dominant term (most of them
satisfy this assumption, provided the eccentricity is large enough
as in Figure~\ref{f:tescart}), then the resonances of the
multiplet superimpose (\red{middle left, middle right and bottom right} of
Figure~\ref{f:tescart}).

Let us remark that for any resonance an interesting fact occurs at
the critical inclination $i=63.4^o$. Within the $J_2$
approximation, this value of the inclination makes the argument of
perigee to be constant (see relation $\eqref{MomeagaOmega_var}$)
and since the argument of any two harmonic terms differs by
$k\omega$, then the shift in $a$ is zero. As a result, for
$i=63.4^o$  the pattern of the resonance is similar to a
pendulum (see \red{bottom left} panel of Figure~\ref{f:tescart}).


\section{A characterization of semi-secular resonances}\label{sec:semi}

In this Section we give some results for the semi-secular resonances associated to the
effects of Sun (Sections~\ref{sec:solarsemi} and \ref{sec:around}) and Moon (Section~\ref{sec:lunarsemi}).
\red{
The Solar semi-secular resonances involve the rates of variation of the argument of perigee,
the longitude of the ascending node and the \red{Sun's} mean motion. We consider different cases within
the quadrupolar approximation: in Proposition~\ref{pro1S} we fix $a$, $e$ and obtain a constraint
on the inclination; in Proposition~\ref{pro2S} we fix $e$,  $i$ and find an expression for the semimajor axis;
in Proposition~\ref{pro3S} we fix $a$, $i$ and find an expression for the eccentricity.
A model for the description of the dynamics in the neighborhhod of the Solar semi-secular resonances is introduced
in Section~\ref{sec:around}. Lunar semi-secular resonances are discussed in Section~\ref{sec:lunarsemi};
in particular, we give bounds for the existence of solutions when $a$, $e$ are fixed (Proposition~\ref{pro:1L}),
for a given semimajor axis and $e$, $i$ satisfying a constraint (Proposition~\ref{pro:2L}), for a given
eccentricity and $a$, $i$ satisfying a constraint (Proposition~\ref{pro:3L}).
}

\subsection{Solar semi-secular resonances}\label{sec:solarsemi}
Solar semi-secular resonances are characterized by a relation of the form
\beq{SSSR}
\alpha\dot\omega+\beta\dot\Omega-\gamma\dot M_S=0\ ,
\eeq
where $\alpha=0,\pm 2$, $\beta=0,\pm 1,\pm 2$, $\gamma\not=0$.
We can set $\dot M_S=1$ $^o$/day. Notice that for $\alpha=\beta=\gamma=2$, one obtains the
so-called \sl evection \rm resonance. In the following we present results similar to those of Section~\ref{sec:tesseral}
(see Propositions~\ref{ref:pro1}, \ref{ref:pro2}, \ref{ref:pro3}). We focus on resonant orbits with
low or moderate eccentricities, let say $e\leq 0.5$.

\begin{proposition}\label{pro1S}
Within the quadrupolar approximation \equ{quad}, consider the resonance relation \equ{SSSR} with
fixed $\alpha$, $\beta$, $\gamma$. For \red{given} values of $a$, $e$, let us introduce the quantity
$$
A=4.98\Bigl({R_E\over a}\Bigr)^{7\over 2}(1-e^2)^{-2}\ .
$$
$(a)$ If $\alpha=0$, then \equ{SSSR} admits one solution provided
$$\Bigl| \frac{\gamma}{2 \beta A} \Bigr| \leq 1\,.$$
$(b)$ If $\alpha \neq 0$, let us introduce the quantity
\beq{delta}
\Delta=\beta^2 A^2+5\alpha (\alpha A+\gamma)A\ .
\eeq
Then, we have the following cases:

$(i)_b$ if $\Delta<0$ or $|{{\beta A\pm\sqrt{\Delta}}\over {5\alpha A}}|>1$, then \equ{SSSR} admits
no solutions;

$(ii)_b$ if $\Delta\geq 0$ and just one of the following two conditions is satisfied
\beq{C2}
-(5\alpha+\beta) A\leq \sqrt{\Delta}\leq (5\alpha-\beta)A\,,
\eeq
\beq{C3}
-(5\alpha-\beta) A\leq \sqrt{\Delta}\leq (5\alpha+\beta)A\ ,
\eeq
then \equ{SSSR} admits one solution;

$(iii)_b$ if $\Delta\geq 0$ and both \equ{C2} and \equ{C3} are satisfied,
then \equ{SSSR} admits two solutions.
\end{proposition}

\begin{proof}
Equation \equ{SSSR} can be written as
$$
5\alpha A\cos^2i-2\beta A\cos i-(\alpha A+\gamma)=0\ .
$$
If $\alpha=0$ then $\beta \neq 0$ and one deduces $(a)$. Otherwise, if $\alpha \neq 0$ then \equ{SSSR} admits
the solutions
$$
\cos i={{{\beta A\pm\sqrt{\Delta}}}\over {5\alpha A}}
$$
with $\Delta$ as in \equ{delta}. The solutions are real if $\Delta\geq 0$,
otherwise there are no solutions. Moreover, taking into account that $|\cos i|\leq 1$,
one obtains $(ii)_b$-$(iii)_b$.
\end{proof}

\begin{figure}
\centering
\includegraphics[width=0.45\textwidth]{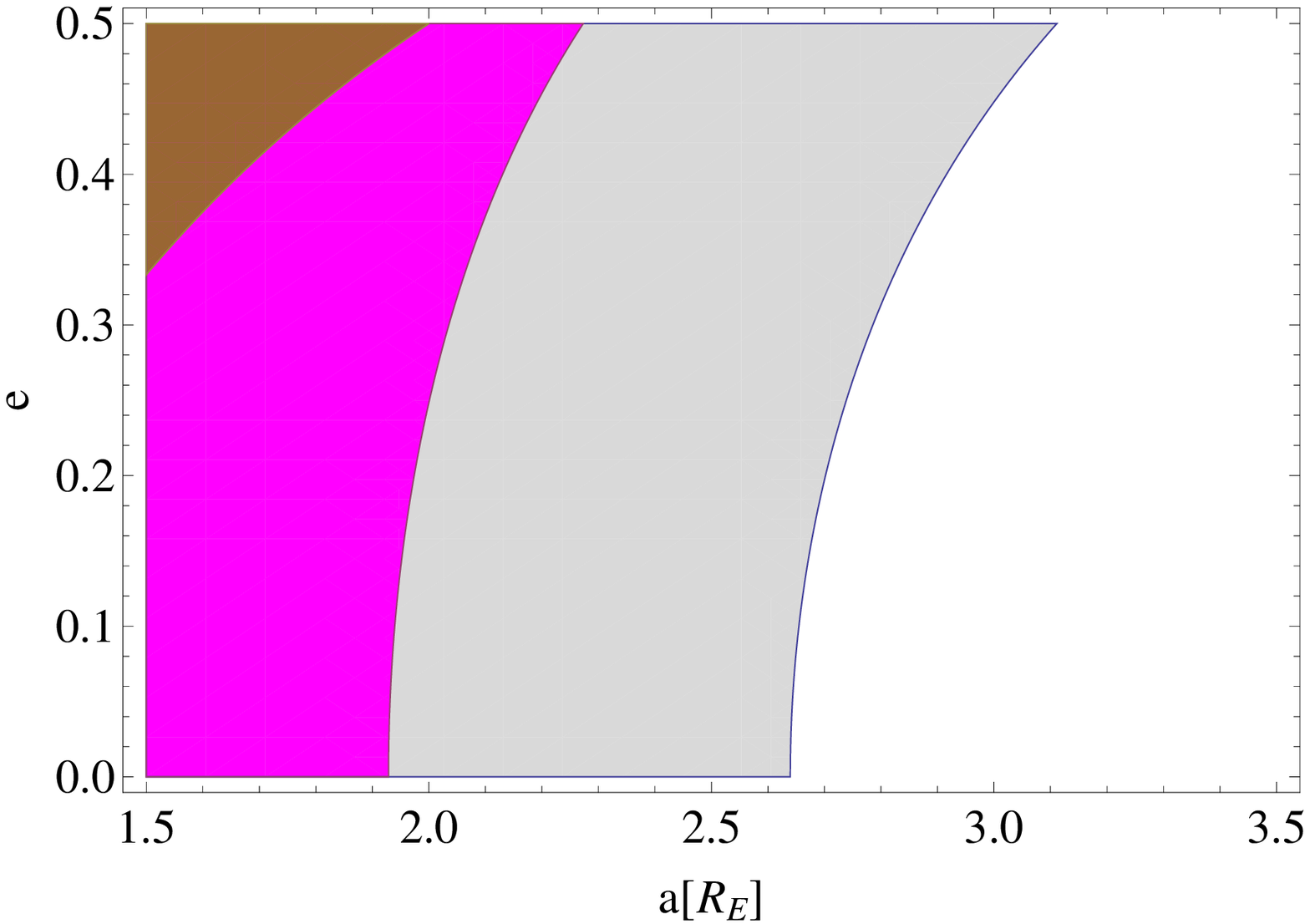} \hspace{0.1cm}
\includegraphics[width=0.45\textwidth]{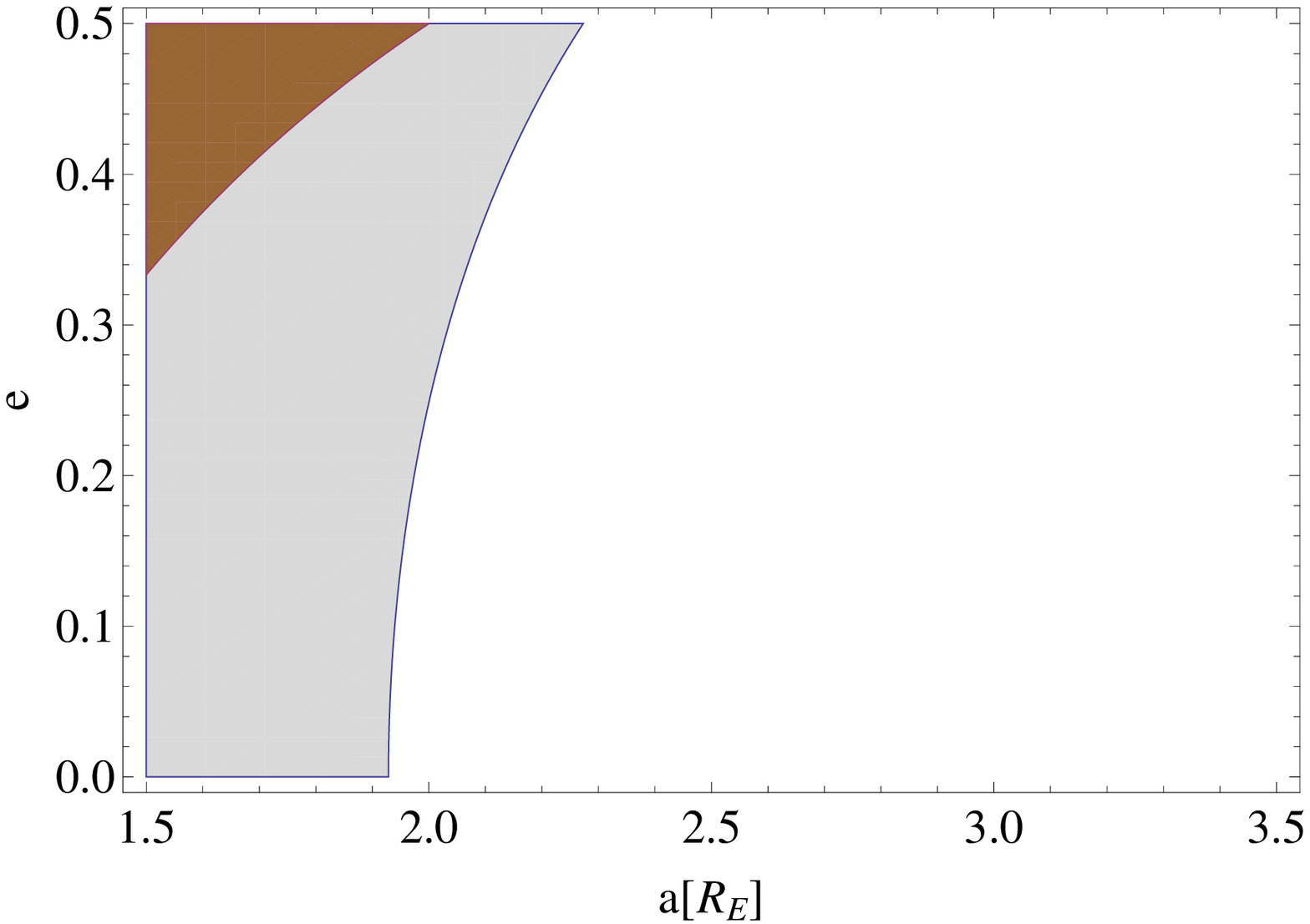}\\
\vspace{0.1cm}
\includegraphics[width=0.45\textwidth]{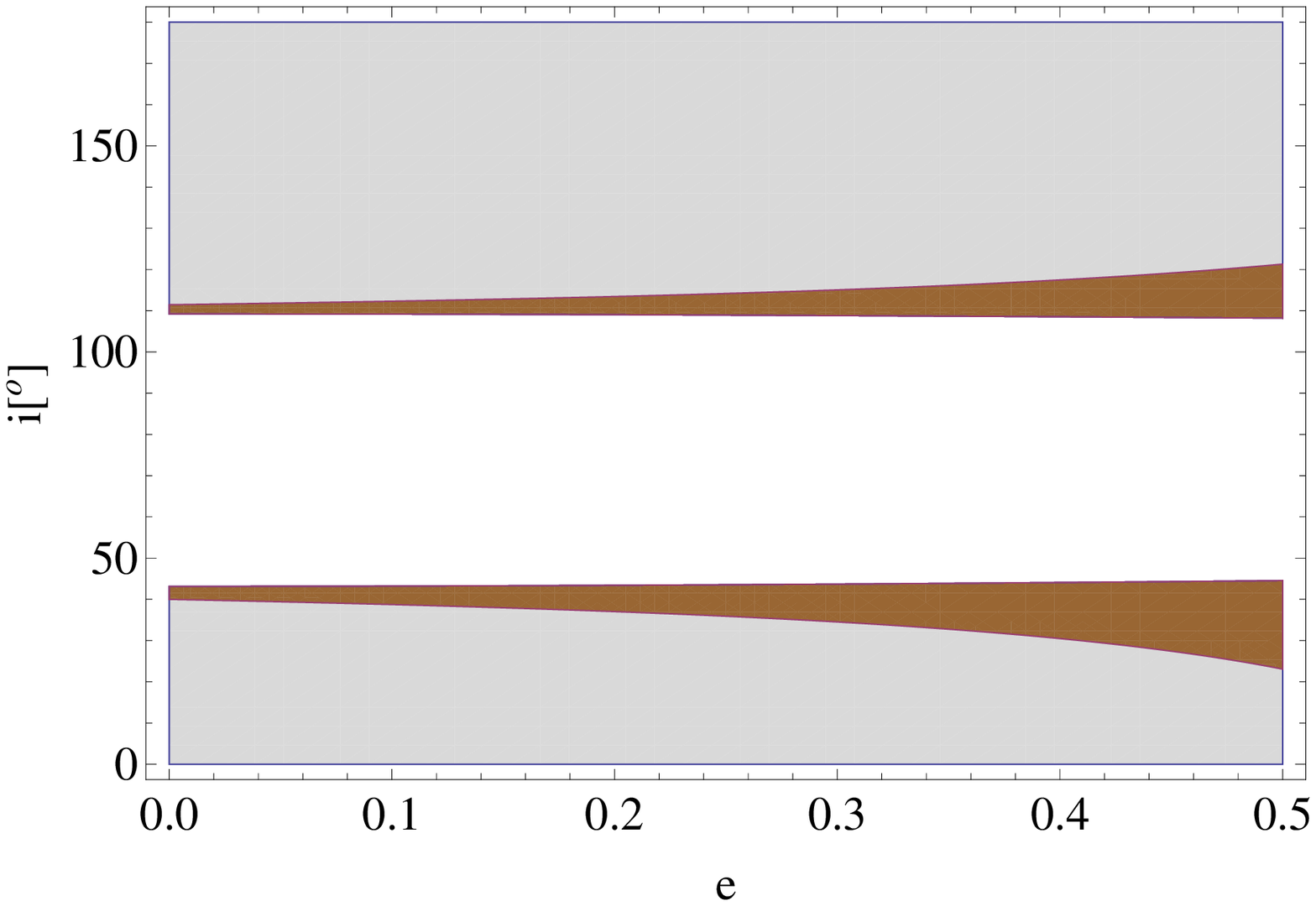} \hspace{0.1cm}
\includegraphics[width=0.45\textwidth]{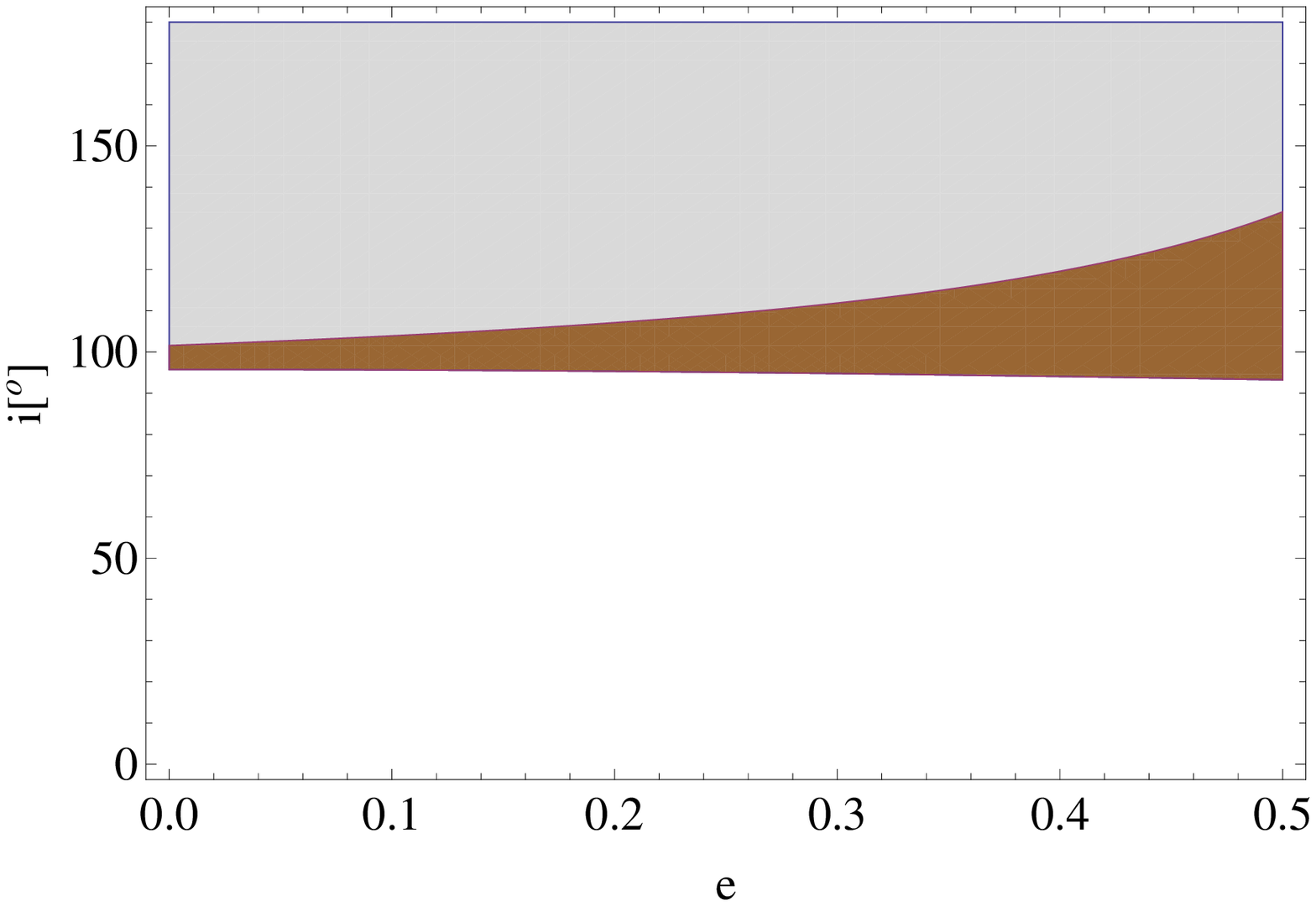}\\
\vspace{0.1cm}
\includegraphics[width=0.45\textwidth]{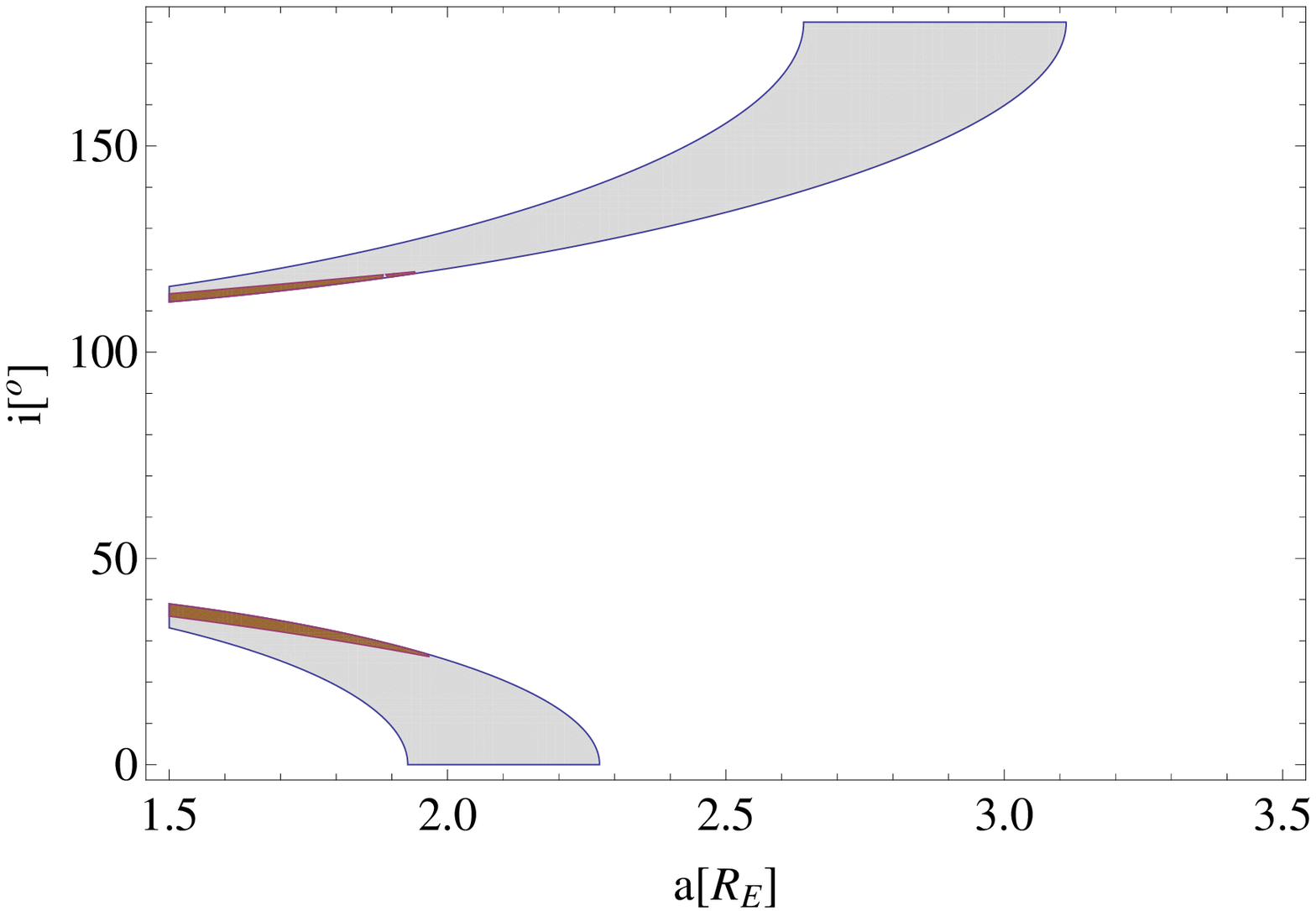} \hspace{0.1cm}
\includegraphics[width=0.45\textwidth]{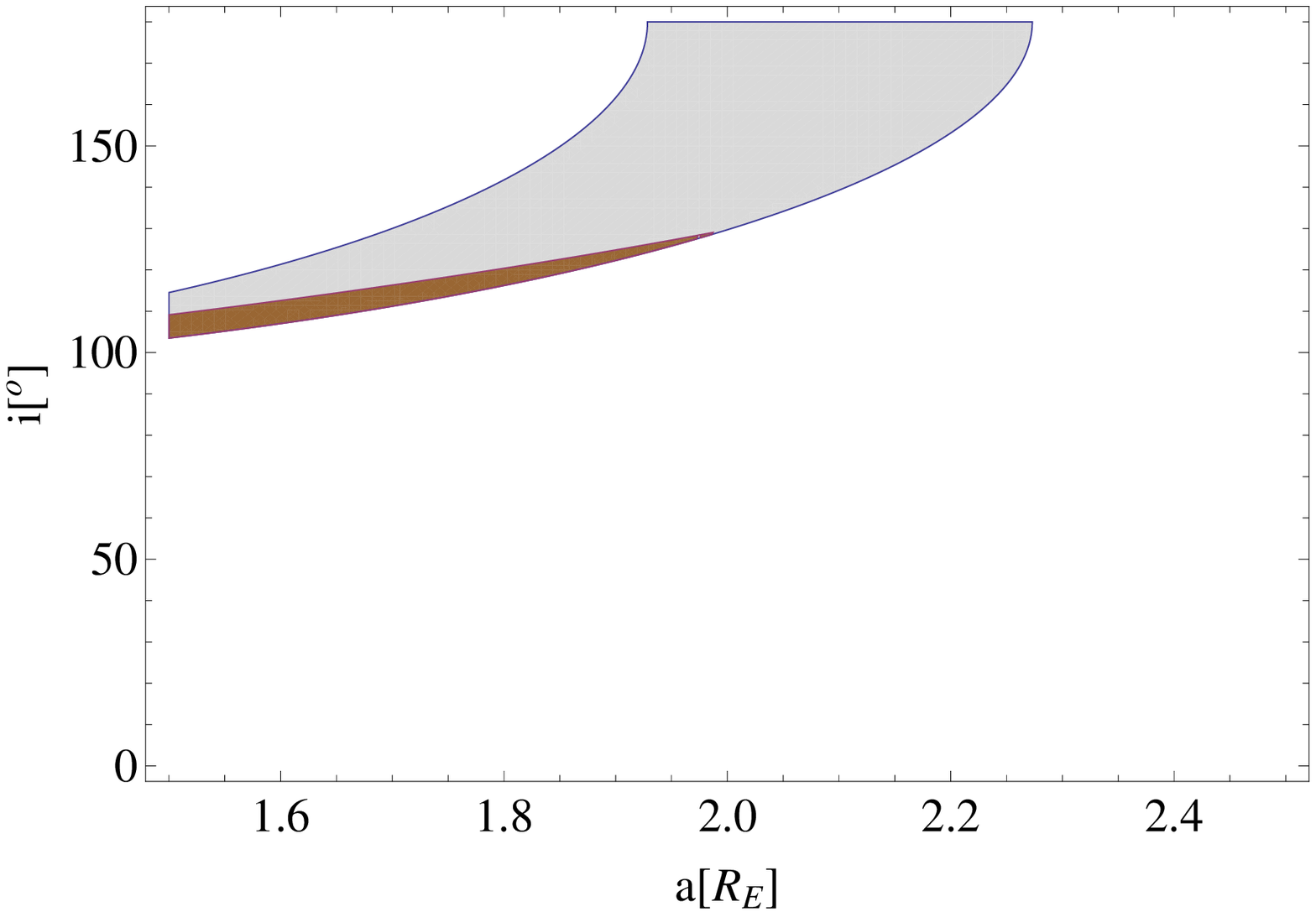}
\caption{The regions in the $(a,e)$ plane (top panels), $(e,i)$
plane (middle panels), $(a,i)$ plane (bottom panels), where
equation \equ{SSSR} admits solutions. Semi-secular resonances with
$\alpha=2$, $\beta=2$, $\gamma=2$ (left column) $\alpha=0$,
$\beta=2$, $\gamma=2$ (right column). Legend: white -- no
solutions, light grey -- one solution, purple -- two solutions,
brown--colliding orbits.} \label{f:prop11}
\end{figure}

As an example, let us consider the resonances with $\alpha=0$,
$\beta=\gamma=2$ and  $\alpha=\beta=\gamma=2$, respectively. Then,
analyzing the conditions stated in Proposition~\ref{pro1S}, we
compute the regions in the $(a,e)$ plane which admit solutions
(see Figure~\ref{f:prop11}).

In a similar way \red{to Proposition~\ref{pro1S},} one can prove the following result.

\begin{proposition}\label{pro2S}
Within the quadrupolar approximation \equ{quad}, for \red{given} values of $e$ and $i$ the resonance relation \equ{SSSR} with
fixed $\alpha$, $\beta$, $\gamma$ admits solutions for
$$
a=R_E ({A\over\gamma})^{2\over 7}
$$
with
$$
A=4.98\alpha(1-e^2)^{-2}(5\cos^2i-1)-9.97\beta(1-e^2)^{-2}\cos i\ ,
$$
provided
$$
4.98\alpha(1-e^2)^{-2}(5\cos^2i-1)-9.97\beta(1-e^2)^{-2}\cos i>\gamma\ .
$$
\end{proposition}

For the same resonances, that is $\alpha=0$, $\beta=\gamma=2$ and
respectively  $\alpha=\beta=\gamma=2$  we compute in
Figure~\ref{f:prop11} the regions in the $(e, i)$ plane where
equation \equ{SSSR} admits solutions.

\vskip.1in

Finally, we have the following result.

\begin{proposition}\label{pro3S}
Within the quadrupolar approximation \equ{quad}, for \red{given} values of $a$ and $i$ the resonance relation \equ{SSSR} with
given $\alpha$, $\beta$, $\gamma$ admits solutions for
$$
e=\sqrt{1-\sqrt{A\over\gamma}}
$$
with
$$
A=(4.98\alpha(5\cos^2 i-1)-9.97\beta\cos i)({R_E\over a})^{7\over 2}\ ,
$$
provided
$$
(4.98\alpha(5\cos^2i-1)-9.97\beta\cos i)\ ({R_E\over a})^{7\over 2}<\gamma\ .
$$
\end{proposition}

For the same semi--secular resonances as above, Figure~\ref{f:prop11} shows the regions in the $(a,i)$ plane where the equation  \equ{SSSR} admits solutions.

Summarizing the results obtained by applying
Propositions~\ref{pro1S}, \ref{pro2S}, \ref{pro3S}, we may
conclude that for the resonance with $\alpha=\beta=\gamma=2$ the
equation \equ{SSSR} might have 0, 1 or 2 solutions, depending on
the values of $a$, $e$ and $i$; for instance, fixing the
eccentricity, let say $e=0.3$, then, we find that for $a>2.8 R_E$
there are no solutions, for $a \in [2.044 R_E, 2.798 R_E]$ there
is one solution, while for $a<2.044 R_E$ one finds two solutions
(see Figure~\ref{f:prop11}). On the contrary, for the resonance
with $\alpha=0$, $\beta=\gamma=2$, the case $(iii)_b$ of
Proposition \ref{pro1S} cannot be fulfilled, that is one might
find at most one solution.

The location of the resonances with $\alpha=\beta=\gamma=2$ and
$\alpha=0$, $\beta=\gamma=2$, respectively, is analyzed
numerically in Figure~\ref{f:sss}, which complements the
results given in Figure~\ref{f:prop11}.
Grey shaded regions define the
zone in which \equ{SSSR} can be solved, \red{collisional regimes are highlighted in brown.}
Figure~\ref{f:sss} numerically validates the Propositions~\ref{pro1S}, \ref{pro2S},
\ref{pro3S} (compare with Figure~\ref{f:prop11} where the results
are obtained by simply representing the regions described by these
Propositions). In Figure~\ref{f:sss}, contours are shown for
different values of orbital parameters as shown in the plot
legends.


\begin{figure}
\centering
\includegraphics[width=0.49\textwidth]{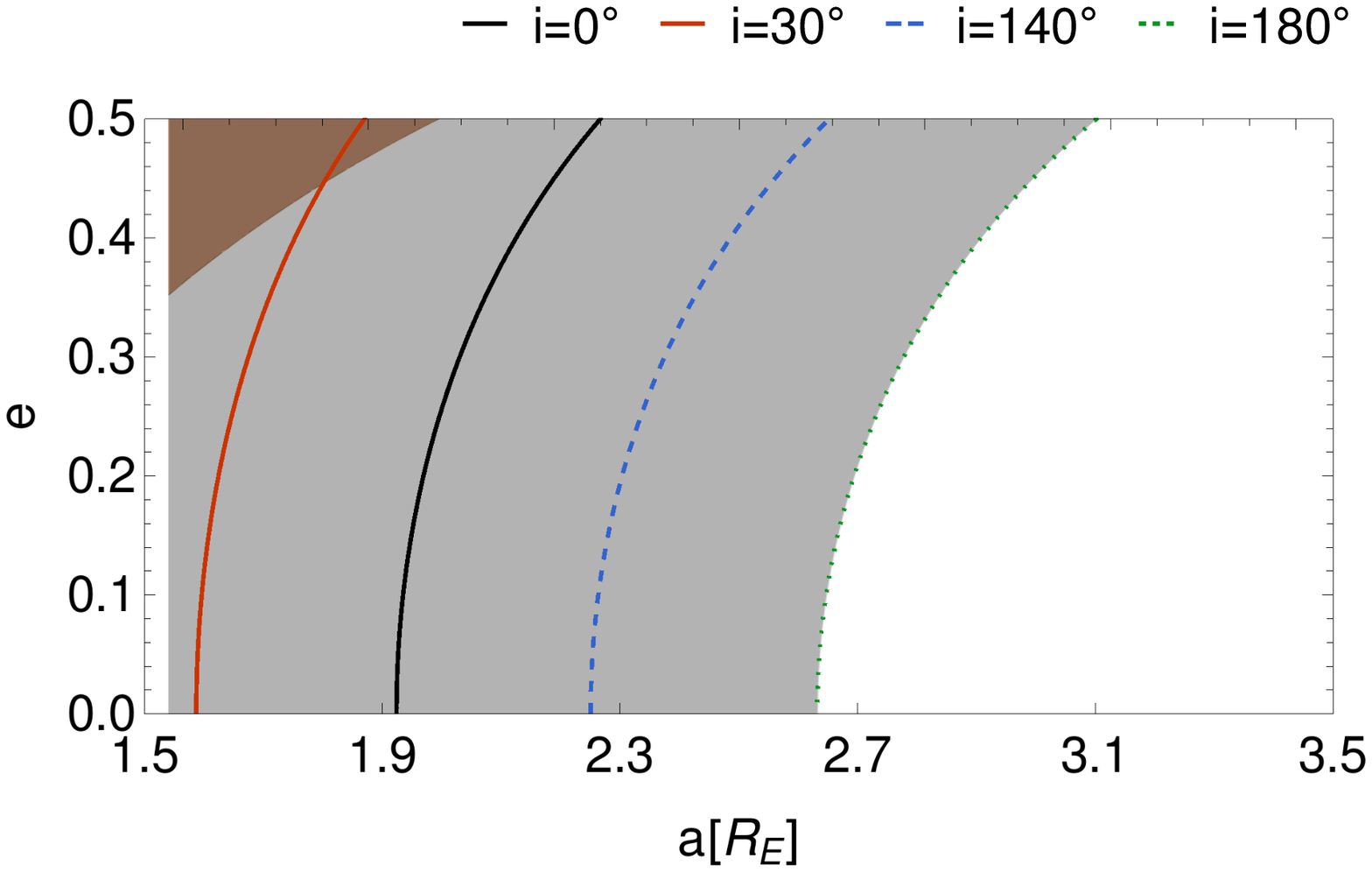}
\includegraphics[width=0.49\textwidth]{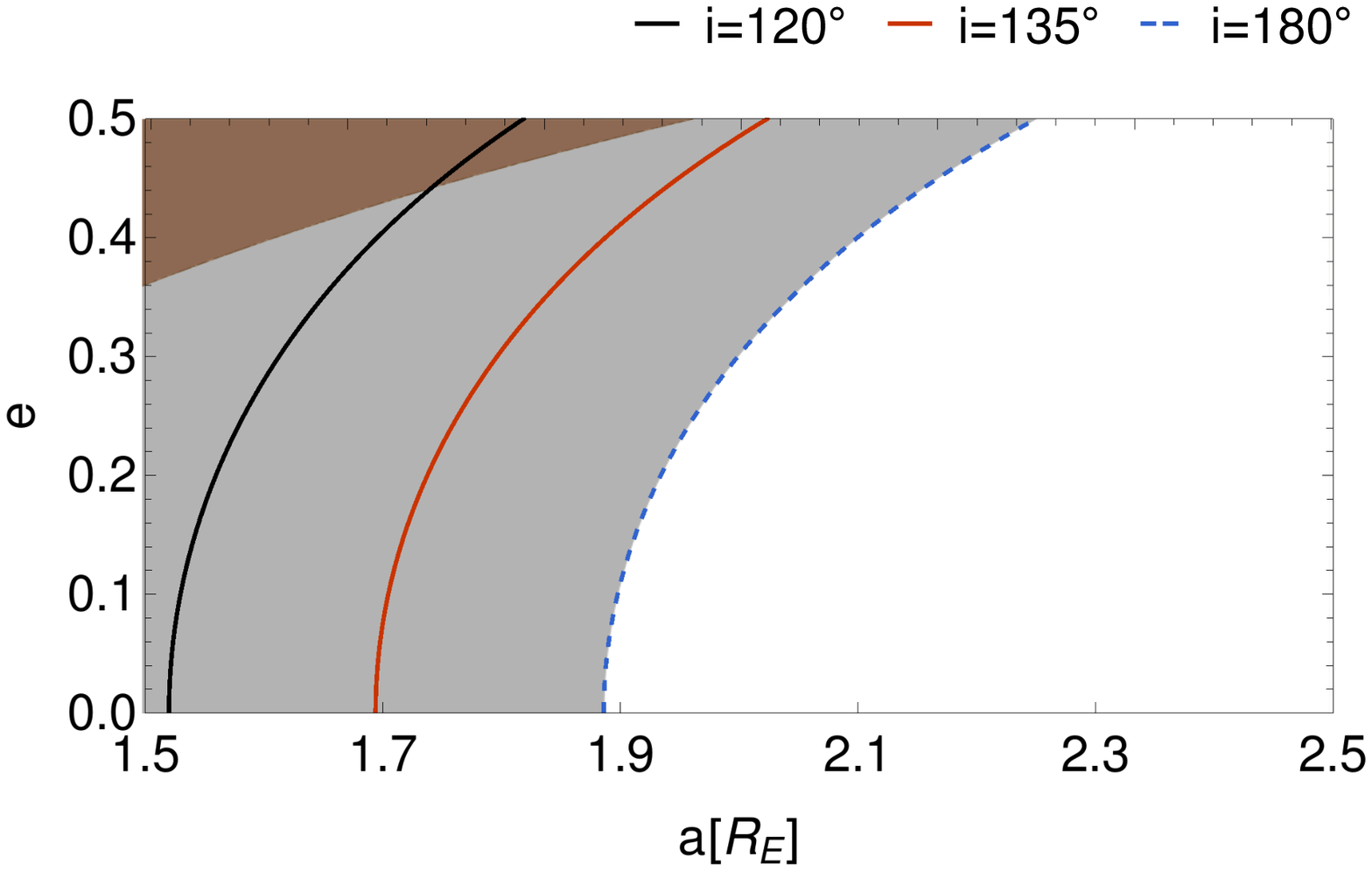}

\vspace{0.3cm}

\includegraphics[width=0.49\textwidth]{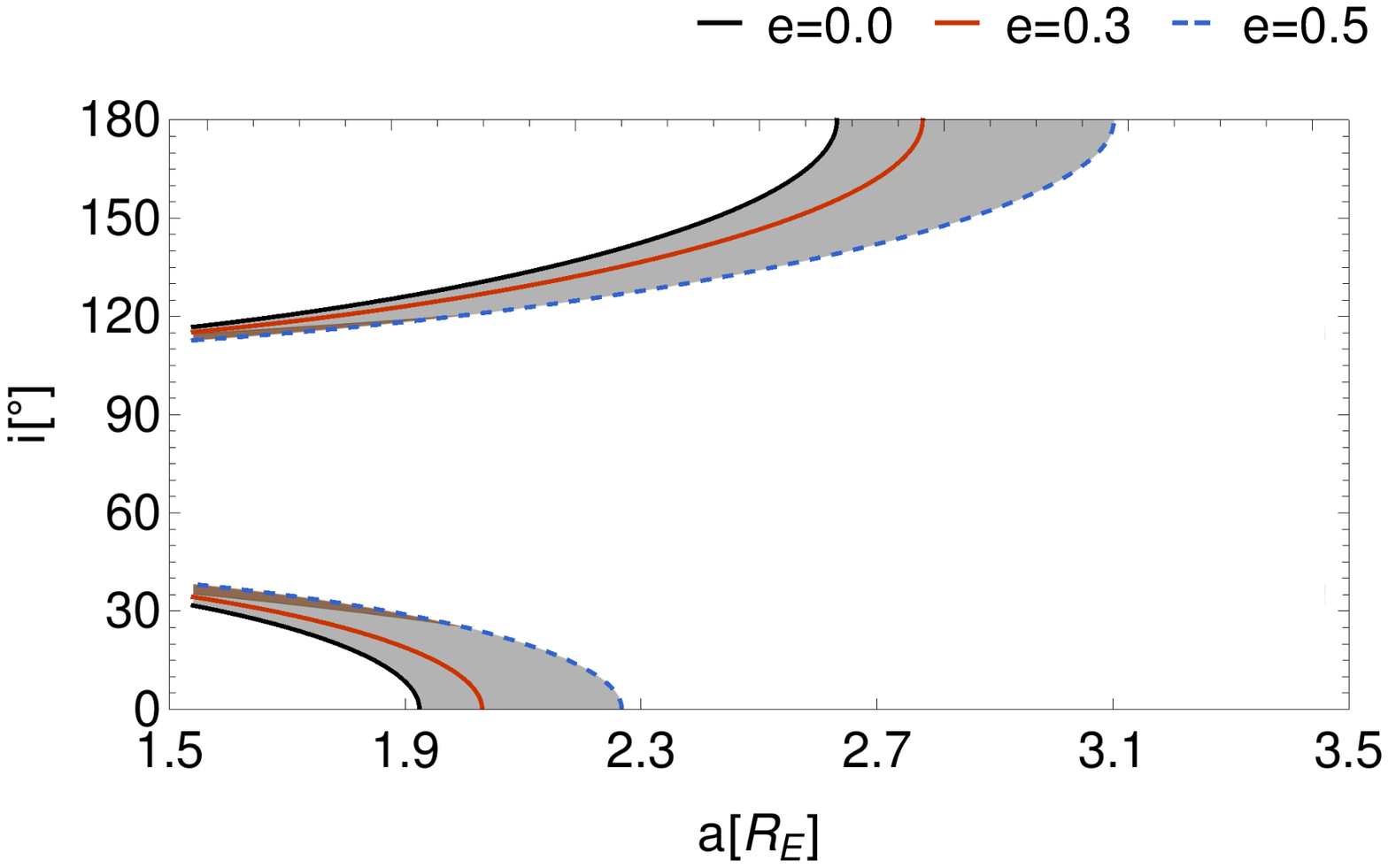}
\includegraphics[width=0.49\textwidth]{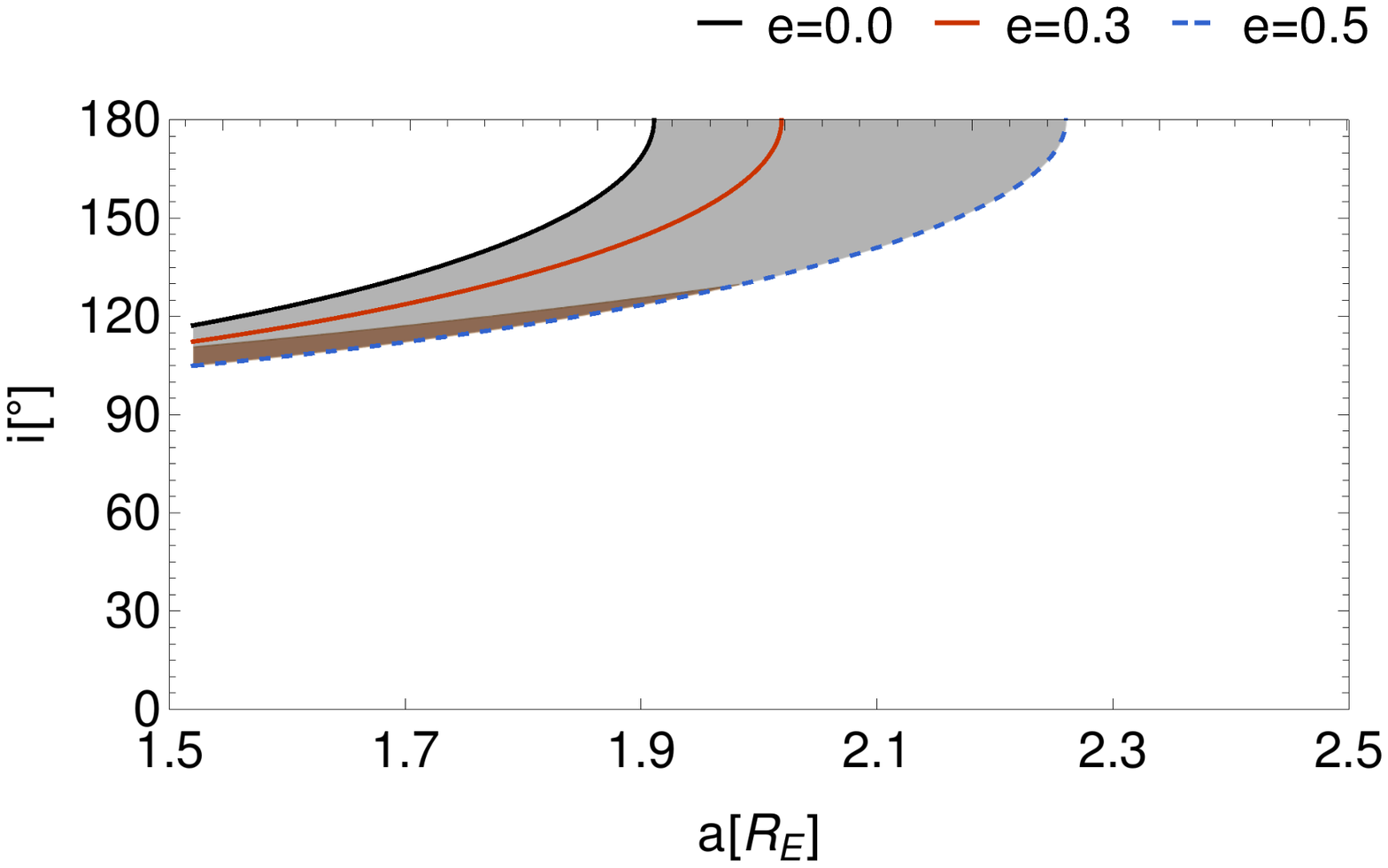}
\caption{Location of semi-secular resonances in $(a,e)$-space
(top) and $(a,i)$-space (bottom) for the $\alpha=2$, $\beta=2$,
$\gamma=2$ (left) and the $\alpha=0$, $\beta=2$, $\gamma=2$
    resonance (right). \red{The full solution space is shown in grey,
    }
    \red{brown regions define parameters that lead to
    collision with Earth.}} \label{f:sss}
\end{figure}


\begin{figure}
\centering
\includegraphics[width=0.45\textwidth]{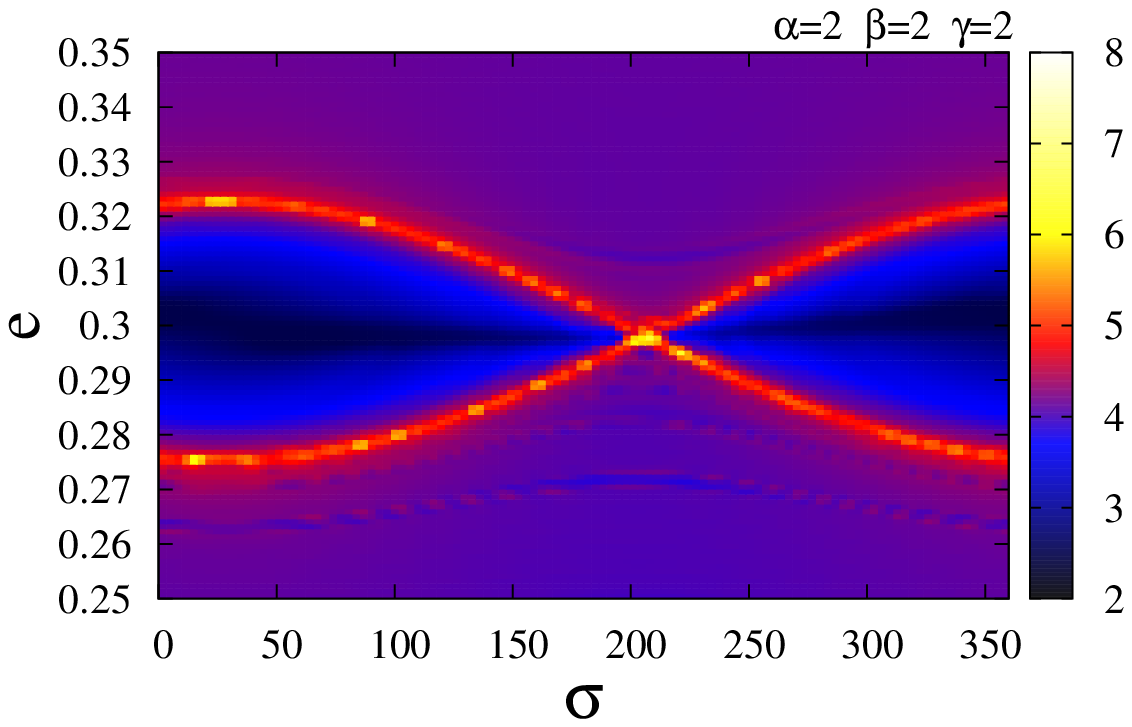}
\includegraphics[width=0.45\textwidth]{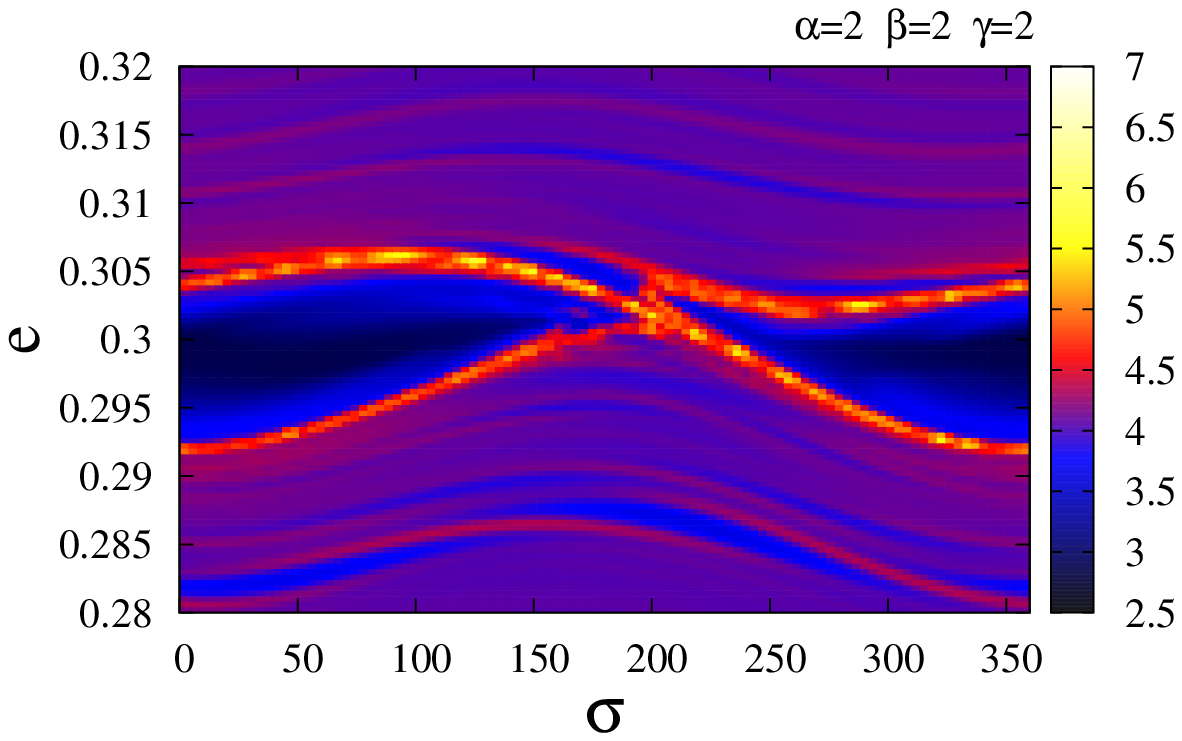}\\
\includegraphics[width=0.45\textwidth]{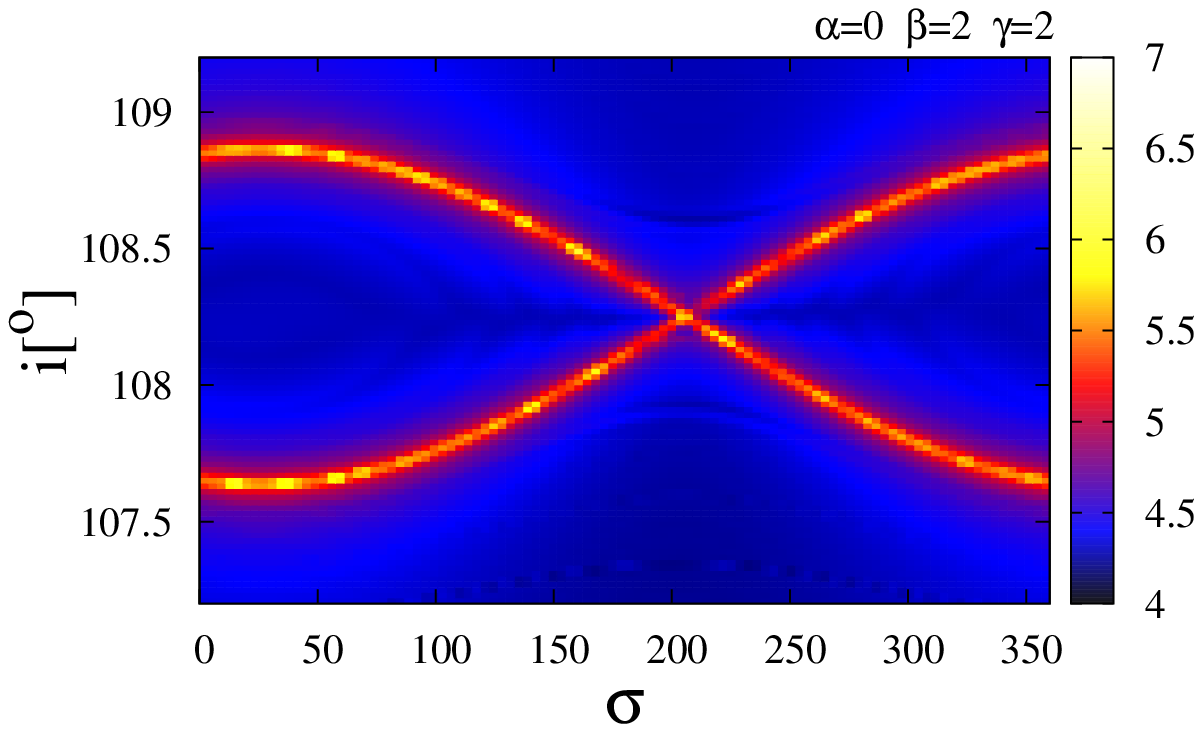}
\includegraphics[width=0.45\textwidth]{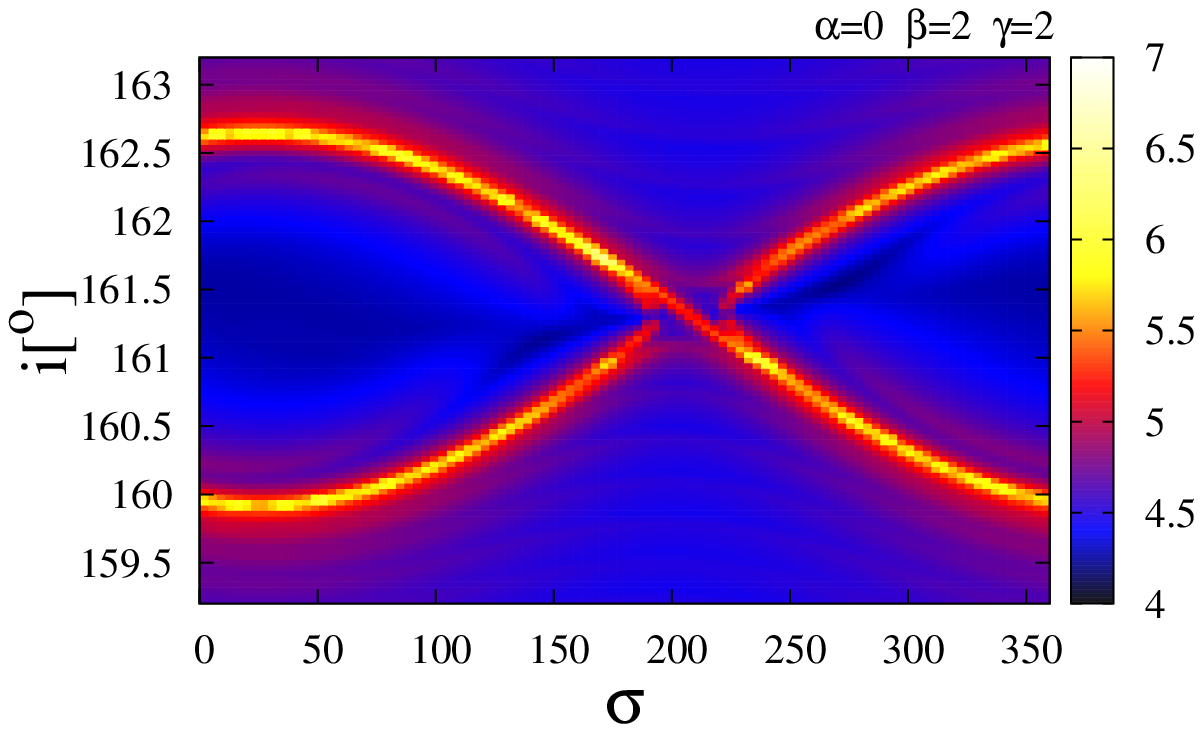}
\vspace{0.7cm}
\caption{Top panels: FLI for the evection resonance with
$\alpha=\beta=\gamma=2$  as a function of the resonant angle
$\sigma=2 \omega+2 \Omega -2 M_S$. The results are obtained for:
$a=1.91 R_E$, $i=19.04^o$, $\Omega=0^o$ and the value of $M_S$ at
the epoch J2000, that is $M_S=357.5256^o$ (left panel), and
respectively $a=2.3 R_E$, $i=135.95^o$, $\Omega=0^o$ and
$M_S=357.5256^o$ (right panel). Bottom panels: FLI for the
resonance with $\alpha=0$, $\beta=\gamma=2$ as a function of the
resonant angle $\sigma=2 \Omega -2 M_S$.  The results are obtained
for: $a=1.392 R_E$, $e=0.05$, $\omega=0^o$ and the value of $M_S$
at the epoch J2000, that is $M_S=357.5256^o$ (left panel), and
respectively $a=1.91 R_E$, $e=0.05$, $\omega=0^o$ and
$M_S=357.5256^o$ (right panel). } \label{f:cart_semi_sec}
\end{figure}

The predicted position of the semi-secular resonances is confirmed
by a cartographic study based on the computation of the FLIs. For
instance, by solving the equation~\equ{SSSR} for
$\alpha=\beta=\gamma=2$, we find the solutions $i=19.04^o$ and $i=
123.04^o$  for $e=0.3$ and $a=1.91 R_E$, and respectively the
unique solution $i=135,95^o$ for $e=0.3$ and $a=2.3 R_E$.
Figure~\ref{f:cart_semi_sec} shows the FLI values for $i=19.04^o$
and $a=1.91 R_E$ (left panel) and respectively $i=135.95^o$ and
$a=2.3 R_E$ (right panel). For these parameters, pendulum--like
plots are obtained; the separatrix divides the phase space into
regions where the resonant angle $\sigma=2 \omega+2 \Omega -2 M_S$
librates or circulates. As far as the resonance $\alpha=0$,
$\beta=\gamma=2$, is considered, we compute the FLIs as in
Figure~\ref{f:cart_semi_sec} and we infer a similar dynamical
behavior as for $\alpha=\beta=\gamma=2$ in the sense that the
phase--space is still similar to a pendulum. However, the
resonances with $\alpha=0$ lead to variations of the inclination,
while the eccentricity remains constant as it is shown in
Section~\ref{sec:around}. In fact, a more detailed study of the
dynamics of Solar semi-secular resonances is provided in
Section~\ref{sec:around}, where both types of commensurabilities,
with $\alpha=0$ and $\alpha \neq 0$, are analyzed for small as
well as moderate values of the eccentricity.

\subsection{Around the Solar semi-secular resonance}\label{sec:around}

\begin{figure}
\centering
\includegraphics[width=0.45\textwidth]{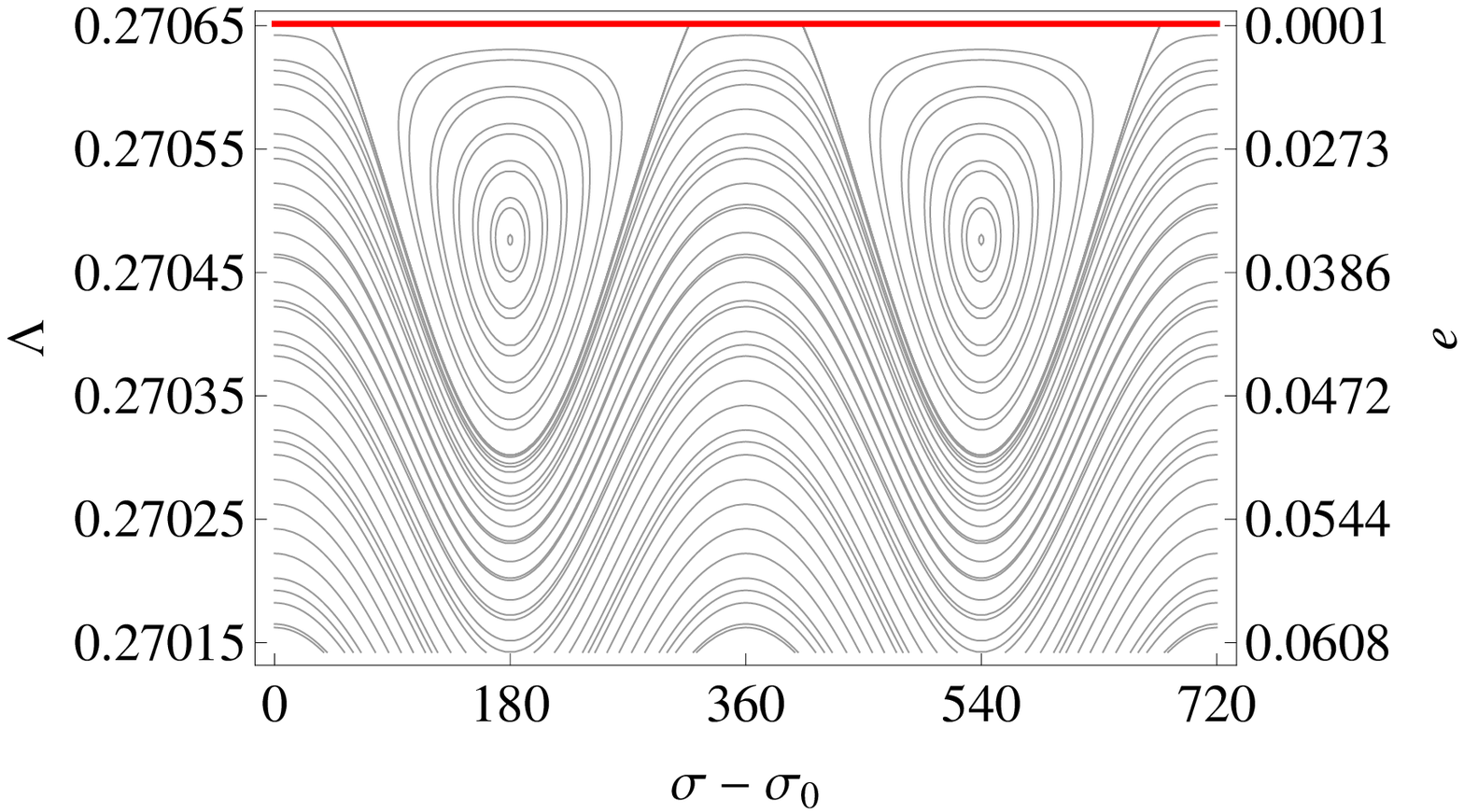}
\includegraphics[width=0.45\textwidth]{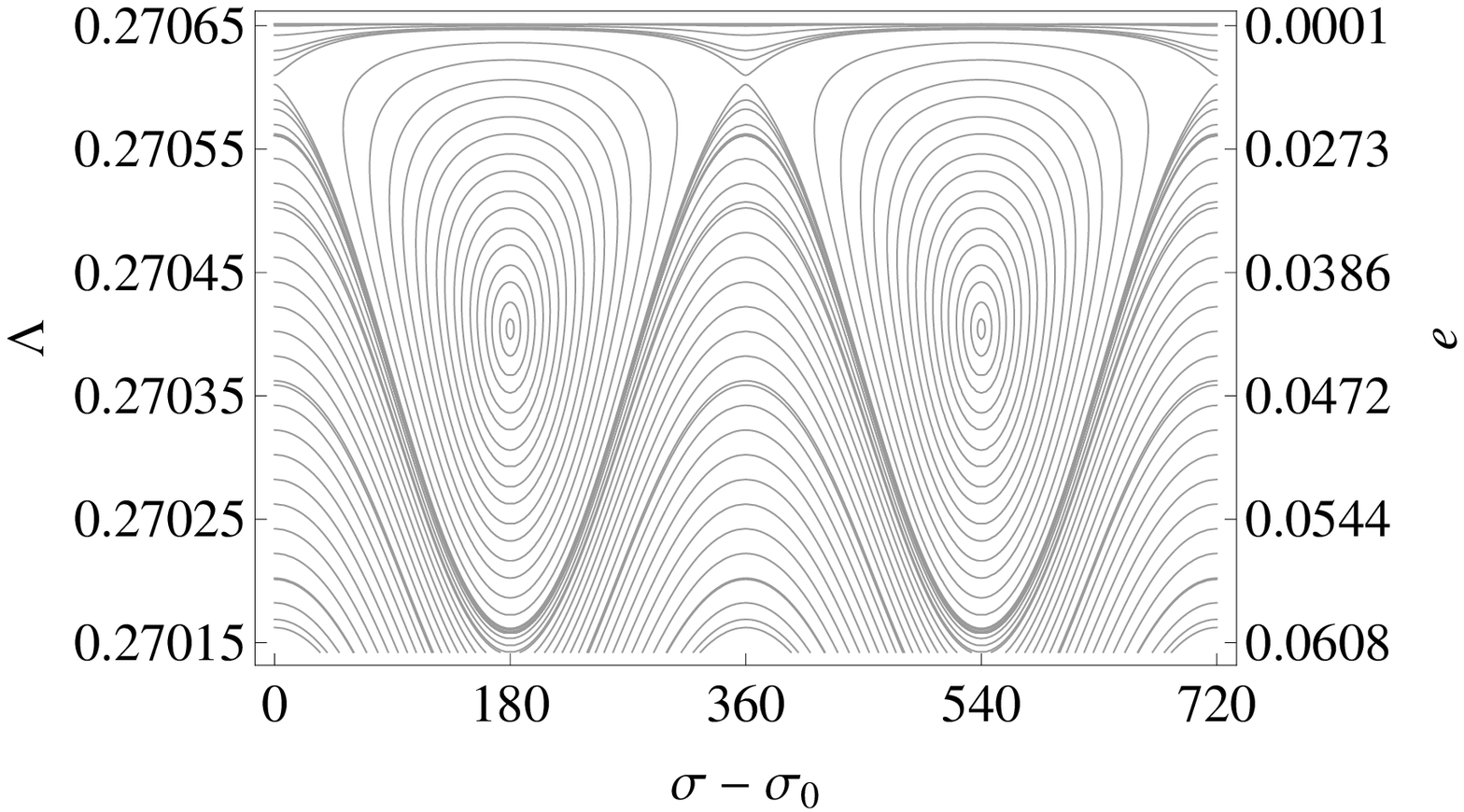}
\vspace{0.1cm} \caption{Phase portrait of the resonance in the
($\sigma-\sigma_0,\Lambda$)-plane, showing a bifurcation of equilibria for the
resonance with $\alpha=\beta=\gamma=2$: stable equilibrium points
at $\sigma-\sigma_0=180^o+360^o k$, $k \in \mathbb{Z}$ (left
plots), and respectively stable and unstable equilibria at
$\sigma-\sigma_0=180^o k$, $k \in \mathbb{Z}$ (right plot). The
plots are obtained for some fixed values of $L$ and $\Gamma$,
computed via \eqref{LGH_aei} and
\eqref{eq:canonical_transformation} by giving $a$, $e$ and $i$.
For the left plot $L$ and $\Gamma$ correspond to  $a=1.937 R_E$,
$e=0.1$, $i=1.51^o$, while the right plot is obtained for $a=1.937 R_E$,
$e=0.1$, $i=2.01^o$.
The values of $\Lambda$ are
expressed with respect to the following units of length and time:
the geostationary distance is unity (it amounts to 42\,164.1696 km)
and the period of Earth's rotation is equal to $2\pi$. On the left plot, all points on the
horizontal line $e=0$ (red line) are equilibrium points, that is the
conditions of the case $(a)_1$ are satisfied. On the
right part we provide the value of the eccentricity computed via
the relation $2 \Lambda = L \sqrt{1-e^2}$. }
\label{f:ssr_phase_portrait}
\end{figure}

In this Section we introduce a model that describes the dynamics
in a neighborhood of a Solar semi-secular resonance of the form
\beq{SRES} \alpha\dot\omega+\beta\dot\Omega-\gamma\dot M_S=0 \eeq
for $\alpha,\beta \in \mathbb{Z}$, $\gamma\in
\mathbb{Z}\backslash\{0\}$. Since for small and moderate
eccentricities, that is $e\leq 0.5$, the semi--secular resonances
occur in LEO and in vicinity of the LEO region (see
Figures~\ref{f:prop11}--\ref{f:sss}), from
\eqref{MomeagaOmega_var} it follows that $\omega$ and $\Omega$, as
well as $M_S$, $M_M$ and $M$ are fast angles in comparison with
the resonant angle $\sigma= \alpha \omega +\beta \Omega- \gamma
M_S$. Considering a specific semi-secular resonance, its dynamics
can be described by a reduced model, which is obtained by
averaging the full Hamiltonian over $M$, $M_S$, $M_M$, $\omega $
and $\Omega$, and retaining only the secular and resonant terms.
Thus, we deduce that the reduced Hamiltonian has the form:
$$
\mathcal{H}_{\alpha \beta \gamma} (G,H,\Phi, \omega, \Omega, M_S; L)=h_0(G,H; L)+ h_1(G,H; L) \cos (\alpha \omega +\beta \Omega - \gamma M_S -\sigma_0)+\dot{M}_S \Phi\,,
$$
where the \sl dummy \rm action $\Phi$, conjugated to $M_S$, was
introduced to make $\mathcal{H}_{\alpha \beta \gamma}$ autonomous,
$\sigma_0$ is a constant, $L$ is also a constant, while $h_0$ and
$h_1$ are known functions.

As noticed in Figure~\ref{f:cart_semi_sec}, from a dynamical
perspective the semi-secular resonances can be grouped in two
classes: resonances with $\alpha \neq 0$ and resonances with
$\alpha = 0$, respectively. Let us consider first the resonances
with $\alpha \neq 0$. We introduce the canonical change of
coordinates $(G,H,\Phi, \omega, \Omega, M_S) \longrightarrow
(\Lambda, \Gamma, Z, \sigma, \Omega, M_S)$, where
\begin{equation}\label{eq:canonical_transformation}
\Lambda=\frac{1}{\alpha} G\,, \qquad \Gamma=H-\frac{\beta}{\alpha} G\,, \qquad Z=\Phi+\frac{\gamma}{\alpha} G\,.
\end{equation}
Clearly, $\Omega$ and $M_S$ are ignorable variables in the new Hamiltonian, so $\Gamma$ and $Z$ are constants of motion.
After neglecting constant terms, we obtain the one degree-of-freedom Hamiltonian:
$$
\mathcal{K}_{\alpha \beta \gamma}(\Lambda, \sigma; \Gamma, L)=f_0(\Lambda; \Gamma, L)-\gamma \dot{M}_S \Lambda +f_1(\Lambda; \Gamma, L) \cos(\sigma-\sigma_0)\,,
$$
where $f_0(\Lambda; \Gamma, L) =h_0(\alpha \Lambda, \Gamma+\beta \Lambda; L)$ and $f_1(\Lambda; \Gamma, L) =h_1(\alpha \Lambda, \Gamma+\beta \Lambda; L)$.

The equilibria are obtained by solving the equations
\begin{equation}\label{eq:equilibria_semi_secular}
\frac{\partial f_0}{\partial \Lambda} -\gamma \dot{M}_S +\frac{\partial f_1}{\partial \Lambda} \cos (\sigma-\sigma_0)=0, \qquad f_1 \sin(\sigma-\sigma_0)=0\,.
\end{equation}

A closer look on the Solar disturbing function \eqref{RSUN}
reveals the fact that for all semi-secular resonances of this
class the function $f_1$ is of second order in the eccentricity.
Then, an analysis of the equations
\eqref{eq:equilibria_semi_secular} shows that we can distinguish
the following cases.

Case $(a)$:  if there exist $L_0$ and $\Gamma_0$ so that
$f_1(\Lambda; \Gamma_0, L_0)=0$ admits solutions, let $\Lambda_0$
be one of them (for instance, for the resonance with
$\alpha=\beta=\gamma=2$ the function $f_1$ vanishes either for
$e=0$ or $i=180^0$). In addition, one has either $(a)_1:$
$\frac{\partial f_1}{\partial \Lambda}(\Lambda_0; \Gamma_0,
L_0)=0$ and $\frac{\partial f_0}{\partial \Lambda}(\Lambda_0;
\Gamma_0, L_0) -\gamma \dot{M}_S=0$ or $(a)_2:$ $\frac{\partial
f_1}{\partial \Lambda}(\Lambda_0; \Gamma_0, L_0) \neq 0$ and
$\Bigl|\frac{\partial f_0}{\partial \Lambda}(\Lambda_0; \Gamma_0,
L_0) -\gamma \dot{M}_S \Bigr| \leq \Bigl|\frac{\partial
f_1}{\partial \Lambda}(\Lambda_0; \Gamma_0, L_0)\Bigr|$, then the
canonical equations admit some equilibrium points that are not
similar to the pendulum-like equilibria. Indeed, in the case
$(a)_1$ all points of the form $(\Lambda_0, \sigma)$ with $\sigma
\in \mathbb{R}$ are equilibrium points; we are dealing with a more
complex problem than the pendulum (compare with \cite{HL1983}).
These points define a line in the $(\Lambda, \sigma)$ plane. On
the other hand, the condition of the case $(a)_2$ assures that in
the domain $\sigma \in [0^o, 360^o)$ there exists at least one
value $\sigma_1$, or at most two values $\sigma_1$ and $\sigma_2$,
such that $(\Lambda_0, \sigma_1)$, respectively $(\Lambda_0,
\sigma_1)$ and $(\Lambda_0, \sigma_2)$, are equilibrium points.

Case $(b)$: if $\sigma=\sigma_0+ 180^o k$, where $k \in
\mathbb{Z}$, then the second equation of
\eqref{eq:equilibria_semi_secular} is identically satisfied.
Substituting $ \sigma $ in the first equation of
\eqref{eq:equilibria_semi_secular} we deduce the values of
$\Lambda$ corresponding to equilibria, provided they exist. In
this case, we can obtain either both stable and unstable
equilibrium points as in the case of a pendulum, let us label this
case by $(b)_1$, or just stable equilibria, denoted hereafter by
$(b)_2$. Indeed, after substituting $\sigma=\sigma_0+ 180^o k$ in
the first equation of \eqref{eq:equilibria_semi_secular}, we find
that for some values of $L$ and $\Gamma$ this equation has
solutions for any $k \in \mathbb{Z}$; in some cases the first
equation of \eqref{eq:equilibria_semi_secular} can be solved only
when $k$ has the form $k=2k_1+1$ with $k_1 \in \mathbb{Z}$.

These aspects are depicted in Figures~\ref{f:cart_semi_sec} and
\ref{f:ssr_phase_portrait} which address the semi-secular
resonance with $\alpha=\beta=\gamma=2$. More precisely, for large
enough eccentricities the phase space is similar to a pendulum as
it is shown in Figure~\ref{f:cart_semi_sec}. However, for small
eccentricities we may have different types of equilibria and
Figure~\ref{f:ssr_phase_portrait} exemplifies the cases described
above. Figure~\ref{f:ssr_phase_portrait} shows the
phase portrait of the resonance. On the left all points on the
horizontal line $e=0$ (red line) are equilibrium points, that is the
conditions of the case $(a)_1$ are satisfied. Indeed, since
$f_1=O(e^2)$, one has $f_1|_{e=0}=0$, $\frac{\partial
f_1}{\partial e}|_{e=0}=0$ and from $\frac{\partial f_1}{\partial
\Lambda}= \frac{\partial f_1}{\partial e} \frac{\partial
e}{\partial \Lambda}$ we deduce $\frac{\partial f_1}{\partial
\Lambda}\mid_{e=0}=0$. Moreover, the left plot of
Figure~\ref{f:ssr_phase_portrait} illustrates that the conditions
of the case $(b)_2$ are satisfied, thus providing stable
equilibria at $\sigma=\sigma_0+ 180^o k$ and $e=0.0359$ (or equivalently $\Lambda=0.270476$), but there
do not exist hyperbolic points similar to those of a pendulum.

The left panel of Figure~\ref{f:ssr_phase_portrait} is obtained
for some fixed values of $L$ and $\Gamma$, computed via
\eqref{LGH_aei} and \eqref{eq:canonical_transformation} by taking
$a=1.937 R_E$, $e=0.1$ and $i=1.51^o$. A small change of the
parameters $L$ and $\Gamma$ can lead to a {\it bifurcation of
equilibria}. Indeed, by considering a slightly different value of
$\Gamma$, which corresponds to $a=1.937 R_E$, $e=0.1$ and
$i=2.01^o$, we obtain a different phase portrait as it is
illustrated by the right panel of
Figure~\ref{f:ssr_phase_portrait}. The conditions of the case
$(a)$ are not satisfied. However, the conditions of the case
$(b)_1$ hold true; the phase space is topologically equivalent to
a pendulum, but the stable and unstable equilibria are located at
different values of $\Lambda$.

Let us consider now the class of resonances with $\alpha=0$. Clearly, it is necessary that $\beta\neq 0$ since otherwise the resonance condition \eqref{SRES} cannot be fulfilled. We proceed as above and we consider the canonical change of coordinates $(G,H,\Phi, \omega, \Omega, M_S) \longrightarrow (G, S, Y, \omega, \sigma, M_S)$, where
$$
S=\frac{1}{\beta} H\,, \qquad \qquad Y=\Phi+\frac{\gamma}{\alpha} H\,.
$$
Clearly, $\omega$ and $M_S$ are ignorable variables in the new
Hamiltonian, so $G$ and $Y$ are constants of motion. Implementing
the transformation and neglecting constant terms, we obtain the
one degree-of-freedom Hamiltonian:
$$
\mathcal{K}_{0 \beta \gamma}(S, \sigma; G, L)=g_0(S; G, L)-\gamma \dot{M}_S S +g_1(S; G, L) \cos(\sigma-\sigma_0)\,,
$$
where $g_0(S; G, L) =h_0(G, \beta S; L)$ and $g_1(S; G, L) =h_1(G, \beta S; L)$.

The equilibria are obtained by solving the equations
$$
\frac{\partial g_0}{\partial S} -\gamma \dot{M}_S +\frac{\partial g_1}{\partial S} \cos (\sigma-\sigma_0)=0, \qquad g_1 \sin(\sigma-\sigma_0)=0\,.
$$
The analysis of the equilibria can be made similarly to the study
performed for the other group of Solar semi--secular resonances. A
key point of the discussion is related to the equation $g_1(S; G,
L)=0$. We underline that there is a difference with respect to the
other group of semi-secular resonances, since the equation
$f_1(\Lambda; \Gamma, L)=0$ always has solutions, while the
equation $g_1(S; G, L)=0$ can have no solutions. Let
us consider the resonance with $\alpha=0$, $\beta=2$, $\gamma=2$.
From \eqref{RSUN} we deduce that the function $g_1$, expressed in
terms of the orbital elements, has the form $g_1=c\, a^2
(1+\frac{3}{2} e^2) \sin^2 i$, where $c$ is a constant. From
Propositions~\ref{pro2S} and \ref{pro3S} (see also the left plots
of Figure~\ref{f:prop11}), this resonance occurs for $i\geq 90^o$;
as a consequence, equation $g_1(S; G, L)=0$ cannot be satisfied, unless $i=180^o$
and we conclude that this resonance is topologically equivalent to
a pendulum as long as $i<180^o$ (see Figure~\ref{f:cart_semi_sec}, bottom panels). This
might not be the case for other resonances and each of them should
be studied individually.

\subsection{Lunar semi-secular resonances}\label{sec:lunarsemi}

Lunar semi-secular resonances are characterized by a relation of the form
\beq{LLLR}
\alpha\dot\omega+\beta\dot\Omega+\alpha_M\dot\omega_M+\beta_M\dot\Omega_M-\gamma\dot M_M=0\ ,
\eeq
where, in the quadrupolar approximations
for the Lunar and Solar expansions, we consider
$\alpha,\alpha_M=0,\pm 2$, $\beta,\beta_M=0,\pm 1,\pm 2$, $\gamma\not=0$.
We can set $\dot M_M=13.06$ deg/day, $\dot \omega_M=0.164$ deg/day,
$\dot \Omega_M=-0.053$ deg/day. Hence, the resonance \equ{LLLR} becomes:
\beqa{resL}
&&4.98\alpha({R_E\over a})^{7\over 2}(1-e^2)^{-2}(5\cos^2i-1)
-9.97\beta({R_E\over a})^{7\over 2}(1-e^2)^{-2}\cos i\nonumber\\
&+&(0.164\alpha_M-0.053\beta_M-13.06\gamma)=0\ .
\eeqa

\begin{figure}
\centering
\includegraphics[width=0.45\textwidth]{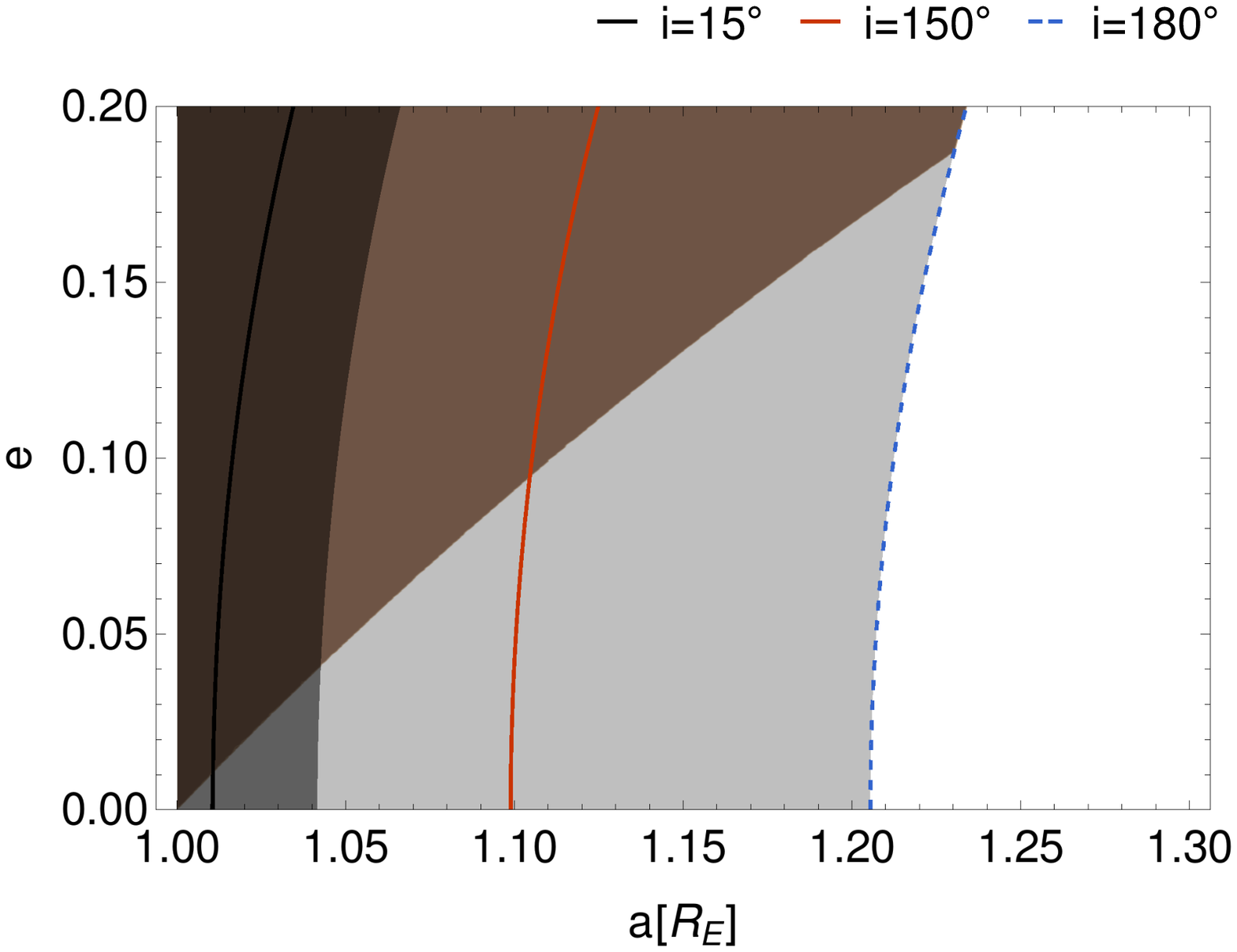}
\includegraphics[width=0.45\textwidth]{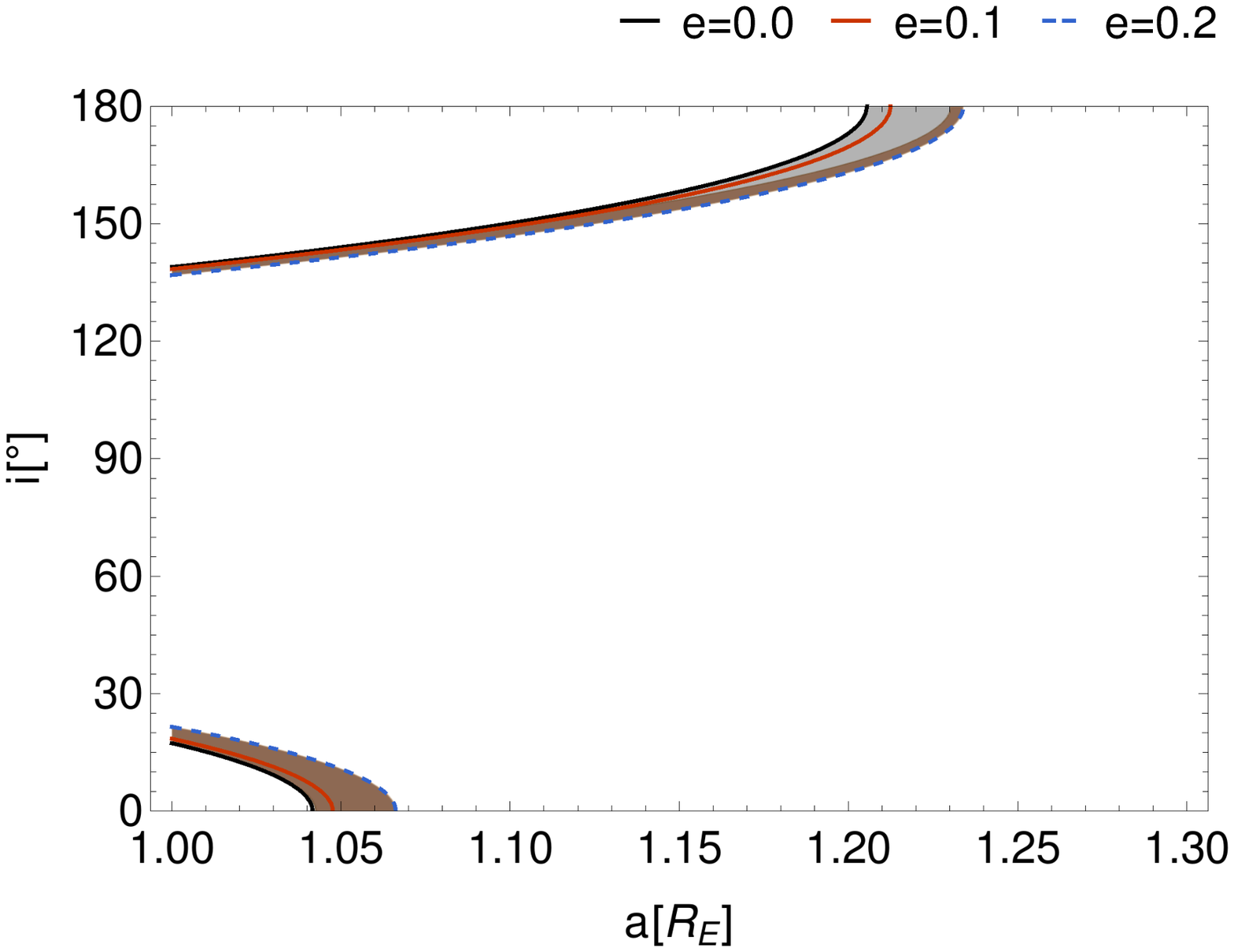}

    \vspace{0.3cm}

\includegraphics[width=0.45\textwidth]{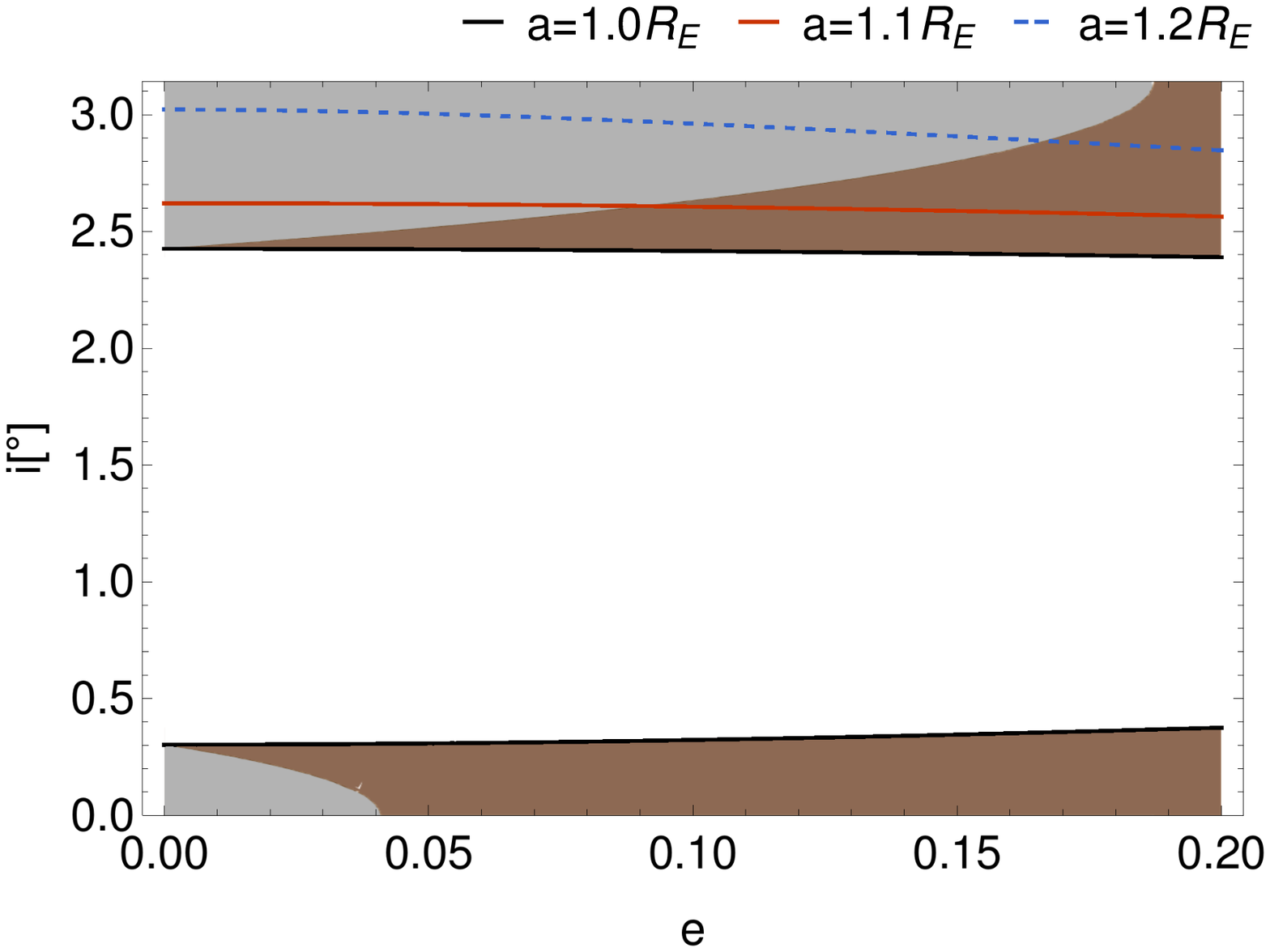}
\includegraphics[width=0.45\textwidth]{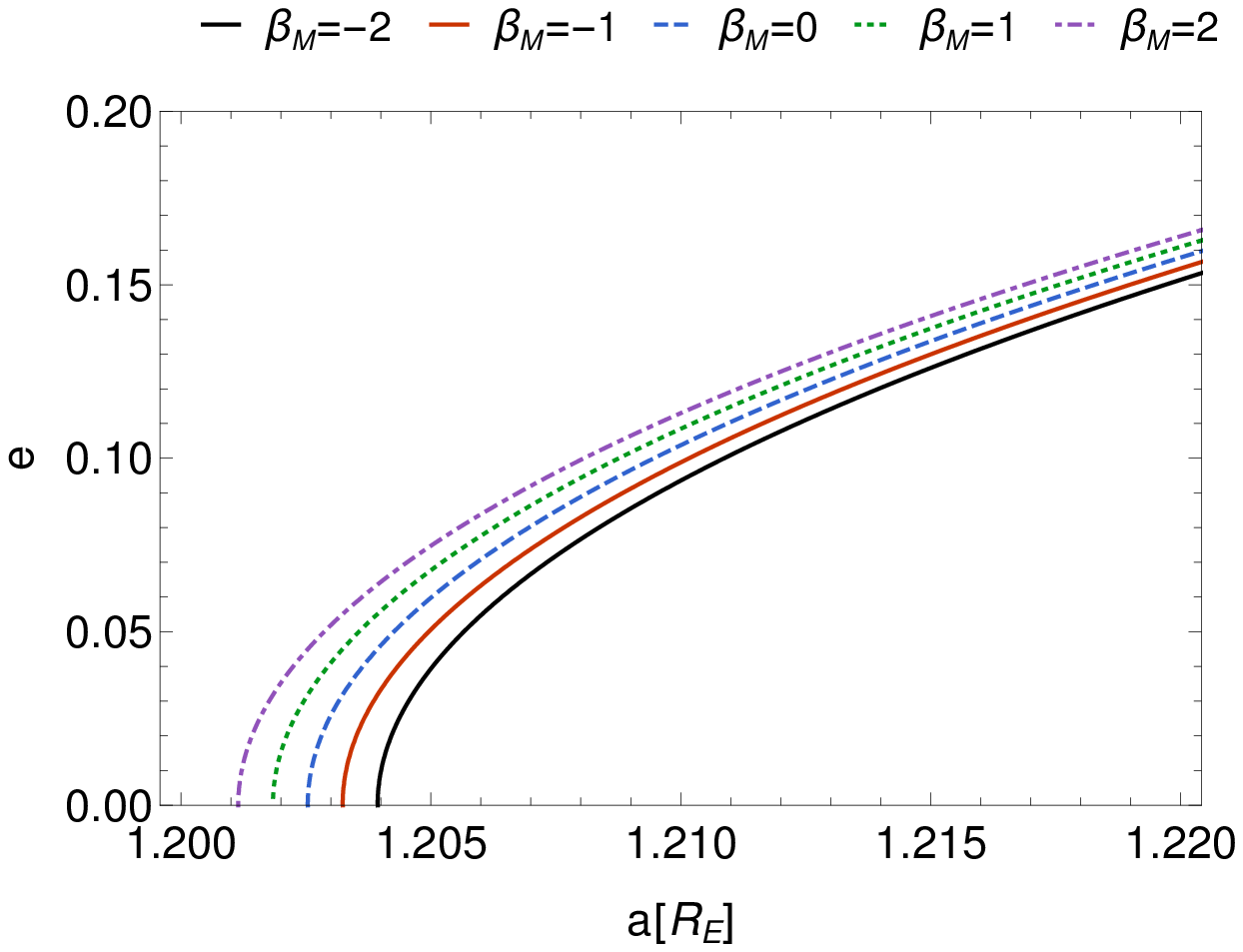}
\caption{Location of Lunar semi-secular resonances with $\alpha=2$, $\beta=1$, $\alpha_M=2$, $\beta_M=2$, $\gamma=2$ (top-left,
top-right, bottom-left) and $\alpha_M=0$, $\beta_M=-2,\dots,2$ (bottom-right).
    \red{The full solution space is shown in grey,} \red{brown regions define
    parameters that lead to collision with Earth. Upper} \red{left, dark: 2 solutions,
    light: 1 solution.}}
\label{f:lss}
\end{figure}

 The possible solutions for $\alpha=2$, $\beta=1$, $\gamma=2$, $\alpha_M=2$,
 $\beta_M=2$ and $\alpha_M=0$, $|\beta_M|\leq2$ are shown in
 Figure~\ref{f:lss}. In the upper left panel we show the $(a,e)$ plane for the
 case with $\alpha_M=2$, where \equ{LLLR} admits one solution (dark grey), two
 solutions (grey), and no solution (white), where the regions have been
 obtained using Proposition~\ref{pro:1L} (see just below). In the brown region
 the orbiting object collides with the Earth at perigee. Three contours are
 shown for fixed values of inclination in black, red, and dashed-blue
 respectively.  In the upper-right and lower-left panel we show the possible
 solutions of \equ{LLLR} for the same case ($\alpha_M=2$) in the $(a,i)$ and
 $(e,i)$ plane, respectively. Here, \red{grey} marks regions, where solutions
 can be found, solutions with one fixed orbital parameter (see legends \red{at
 the top}) are shown in color code. To demonstrate the effect of varying
 $\beta_M$ on the contours, we plot a magnification of the $(a,e)$ plane of the
 previous case, but with $\alpha_M=0$ and $\beta_M=-2,\dots,2$ in the
 bottom-right panel of Figure~\ref{f:lss}. A slight shift of the contours with
 increasing $\beta_M$ to smaller values of the semi-major axis $a$ is clearly
 visible.  As it can be inferred from Figure~\ref{f:lss}, Lunar semi-secular
 resonances are very close to the Earth.  Therefore, from a practical
 perspective these resonances are less relevant for the Earth-Moon system;
 resonant orbits are either colliding orbits or escape orbits, as effect of the
 atmospheric drag.

\begin{proposition}\label{pro:1L}
Within the quadrupolar approximation \equ{quad}, consider the resonance relation \equ{resL} with
given $\alpha$, $\beta$, $\alpha_M$, $\beta_M$, $\gamma$. For \red{given} values of $a$, $e$, let us introduce the quantities
$$
A=4.98({R_E\over a})^{7\over 2}(1-e^2)^{-2}\ ,
$$
and let $\Gamma$ be defined as
$$
\Gamma=-0.164\alpha_M+0.053\beta_M+13.06\gamma\ .
$$
Then, the conclusions of Proposition~\ref{pro1S} hold with $\gamma$ replaced by $\Gamma$ (which depends
on the free parameters $\alpha_M$, $\beta_M$, $\gamma$). Hence, for $\Delta_L$ defined as
$$
\Delta_L=\beta^2 A^2+5\alpha (\alpha A+\Gamma)A\ ,
$$
let us consider the following inequalities:
\beq{C1L}
\beta^2 A^2+5\alpha (\alpha A+\Gamma)A\geq 0
\eeq
\beq{C2L}
-(5\alpha+\beta)A\leq \sqrt{\Delta_L}\leq (5\alpha-\beta)A
\eeq
\beq{C3L}
-(5\alpha-\beta)A\leq \sqrt{\Delta_L}\leq (5\alpha+\beta)A .
\eeq
Then, we have the following cases:
$(i)$ if $\Delta_L<0$ or $|{{\beta A\pm\sqrt{\Delta_L}}\over {5\alpha A}}|>1$, then \equ{resL} admits
no solutions;
$(ii)$ if \equ{C1L} and just one of the conditions \equ{C2L} and \equ{C3L} are satisfied,
then \equ{resL} admits one solution;
$(iii)$ if \equ{C1L}, \equ{C2L}, \equ{C3L} are satisfied,
then \equ{resL} admits two solutions.
\end{proposition}
In a similar way, one can prove the following result.
\begin{proposition}\label{pro:2L}
Within the quadrupolar approximation \equ{quad}, the resonance relation \equ{resL} with
given $\alpha$, $\beta$, $\gamma$ admits solutions for
$$
a=R_E ({A\over\Gamma})^{2\over 7}
$$
with
$$
A=4.98\alpha(1-e^2)^{-2}(5\cos^2i-1)-9.97\beta(1-e^2)^{-2}\cos i\ ,\qquad
\Gamma=-0.164\alpha_M+0.053\beta_M+13.06\gamma\ ,
$$
provided
$$
4.98\alpha(1-e^2)^{-2}(5\cos^2i-1)-9.97\beta(1-e^2)^{-2}\cos i>\Gamma\ .
$$
\end{proposition}
We conclude this section with the following Proposition.
\begin{proposition}\label{pro:3L}
Within the quadrupolar approximation \equ{quad}, the resonance relation \equ{resL} with
given $\alpha$, $\beta$, $\gamma$ admits solutions for
$$
e=\sqrt{1-{\sqrt{A\over\Gamma}}}
$$
with
$$
A=(4.98\alpha(5\cos^2 i-1)-9.97\beta\cos i)({R_E\over a})^{7\over 2}\ ,\qquad
\Gamma=-0.164\alpha_M+0.053\beta_M+13.06\gamma\ ,
$$
provided
$$
(4.98\alpha(5\cos^2i-1)-9.97\beta\cos i)\ ({R_E\over a})^{7\over 2}<\Gamma\ .
$$
\end{proposition}

\section{A characterization of secular resonances}\label{sec:secular}
Secular resonances depend on the rates of variation of slowly varying quantities, typically
the argument of perigee and the longitude of the ascending node. We distinguish again
between Solar and Lunar resonances which are, respectively, described in Sections~\ref{sec:solarsec}
and \ref{sec:lunarsec}.
\red{
In Section~\ref{sec:solarsec} we obtain that the commensurability relation describing the Solar
secular resonances turns out to be a condition on the inclination in terms of the coefficients
entering the resonance relation. In Section~\ref{sec:lunarsec} we discuss the Lunar secular resonances,
which can be treated as in Propositions~\ref{pro:1L}, \ref{pro:2L}, \ref{pro:3L} describing the Lunar
semi-secular resonances.
}

\subsection{Solar secular resonances}\label{sec:solarsec}
Given that the Solar rates of variation of the argument of perigee
and the longitude of the ascending node are zero, the Solar
secular resonances occur whenever \beq{SSRES}
\alpha\dot\omega+\beta\dot\Omega=0 \eeq for some $\alpha$,
$\beta\in \mathbb{Z}$. It might look strange to refer to \equ{SSRES} as
\sl Solar \rm secular resonances, since the elements of the Sun do
not occur at all in this expression, due to the fact that we have
set to zero both $\dot\omega_S$ and $\dot\Omega_S$. However, we
will keep this terminology for consistency with the previous and
following sections.

The relation \equ{SSRES} can be written as
$$
({R_E\over a})^{7\over 2}(1-e^2)^{-2}\ [4.98\alpha(5\cos^2i-1)-9.97\beta\cos i]=0\ .
$$
The above expression depends just on the inclination and it is
satisfied whenever we have
$$
\cos i={{9.97\beta\pm\sqrt{99.4009\beta^2+496.506\alpha^2}}\over {49.8\alpha}}\ .
$$
Since $99.4009\beta^2+496.506\alpha^2\geq 0$, we have two real
solutions, which have a physical meaning provided that $|\cos
i|\leq 1$, which gives a condition on $\alpha$, $\beta$ to have
Solar secular resonances.

For a detailed investigation of Solar secular resonances depending
just on the inclination, we refer the reader to \cite{CGPSIADS}
and \cite{HughesI}. Here, it is worthwhile to mention that the
resonances \equ{SSRES} include the cases $\dot\omega=0$,
corresponding to the critical inclinations ($i=63.4^o$,
$i=116.4^o$), and $\dot\Omega=0$, corresponding to polar orbits.

\subsection{Lunar secular resonances}\label{sec:lunarsec}
The Lunar secular resonances occur whenever
$$
\alpha\dot\omega+\beta\dot\Omega+0.164\alpha_M-0.053\beta_M=0
$$
for some $\alpha$, $\beta, \alpha_M, \beta_M\in \mathbb{Z}$; this relation can be written as
$$
({R_E\over a})^{7\over 2}(1-e^2)^{-2}\ [4.98\alpha(5\cos^2i-1)-9.97\beta\cos i]=0.053\beta_M-0.164\alpha_M\ .
$$
The same result of Proposition~\ref{pro:1L} holds with $\gamma=0$ to find the inclination, having fixed $a$, $e$,
Proposition~\ref{pro:2L} with $\gamma=0$ to determine the semimajor axis,
having fixed $e$, $i$, Proposition~\ref{pro:3L} with $\gamma=0$ to find the eccentricity, having fixed $a$, $i$.

An example of the location of the Lunar secular resonances is
given in Figure~\ref{f:lunsec} as a function of the orbital
elements $a$, $e$, $i$. In the example we choose $\alpha=2$,
$\beta=1$, $\alpha_M=0$ and $\beta_M=\pm1$. In the upper plots
dark grey marks regions where two distinct solutions of above
equations are possible (in the sense of Proposition~\ref{pro:1L}).
For our parameters and the case $\beta_M=-1$ this region can be
found in the space $(a,e)$ for values up to $a\simeq4.5-5.0R_E$.
For the case $\beta=1$ this region of two distinct solutions is
increased up to values $a\simeq6.0-7.0R_E$. Moreover, an
additional region in the space $(a,e)$ is present, starting from
$a\simeq6.0$ where only one solution to the condition for Lunar
secular resonances exists (shown in light grey in the upper right
plot of Figure~\ref{f:lunsec}). The reason for the different
topologies in the solution space becomes clear when looking at the
regions of possible solutions in the space $(a,i)$ (shown at the
bottom of Figure~\ref{f:lunsec}). Two symmetric solutions exist
for the case $\beta=-1$ up to a certain value in $a$, while in
case of $\beta=1$ the two distinct solution branches are
asymmetric. At the moment, when the lower branch stops at some
value in $a$ (that depends on the specific value chosen for
eccentricity $e$) the upper branch still extends farther to larger
values in $a$, which yields the region in solution space with one
solution, according to Proposition~\ref{pro:1L}.


\begin{figure}
\centering
\includegraphics[width=0.45\textwidth]{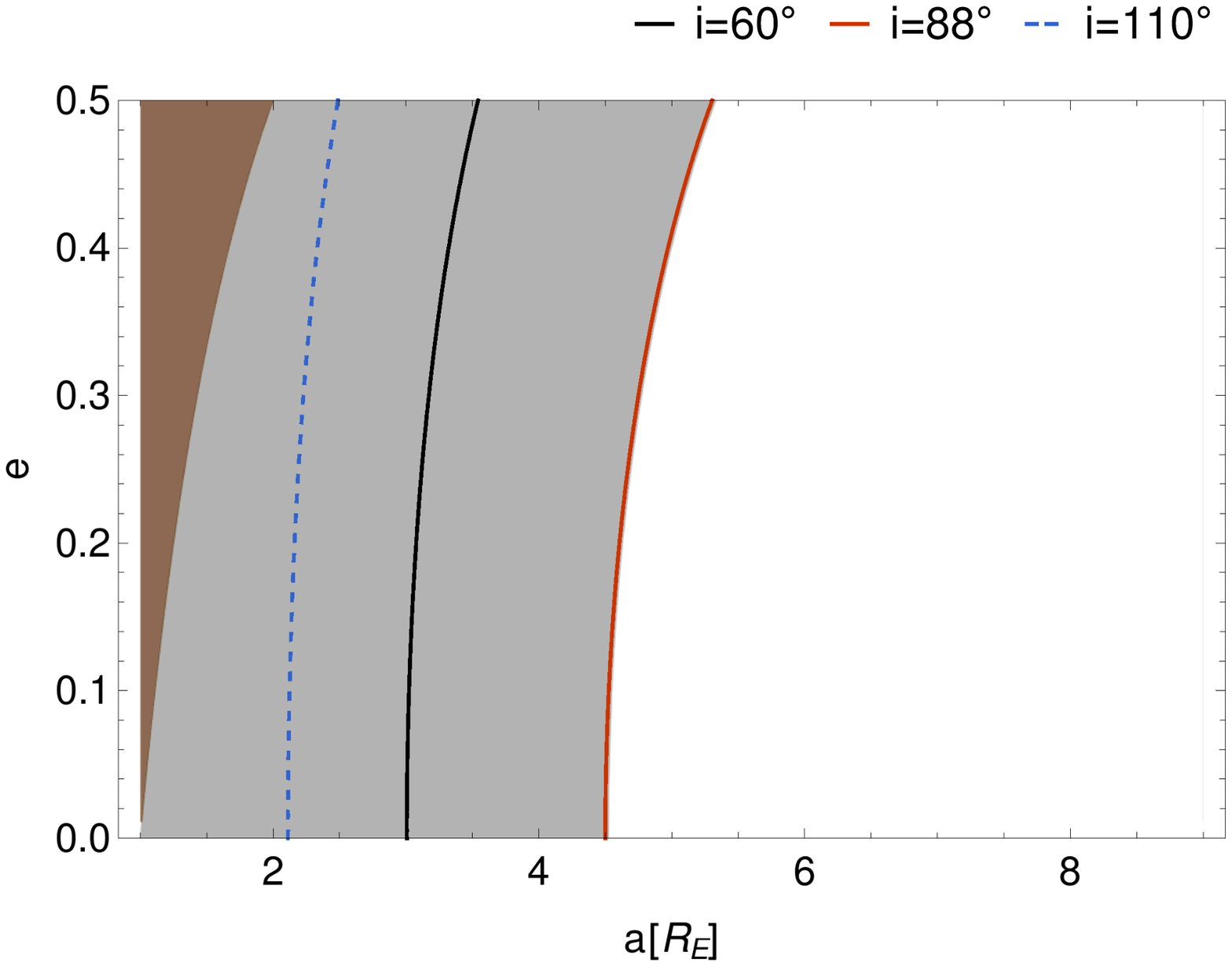}
\includegraphics[width=0.45\textwidth]{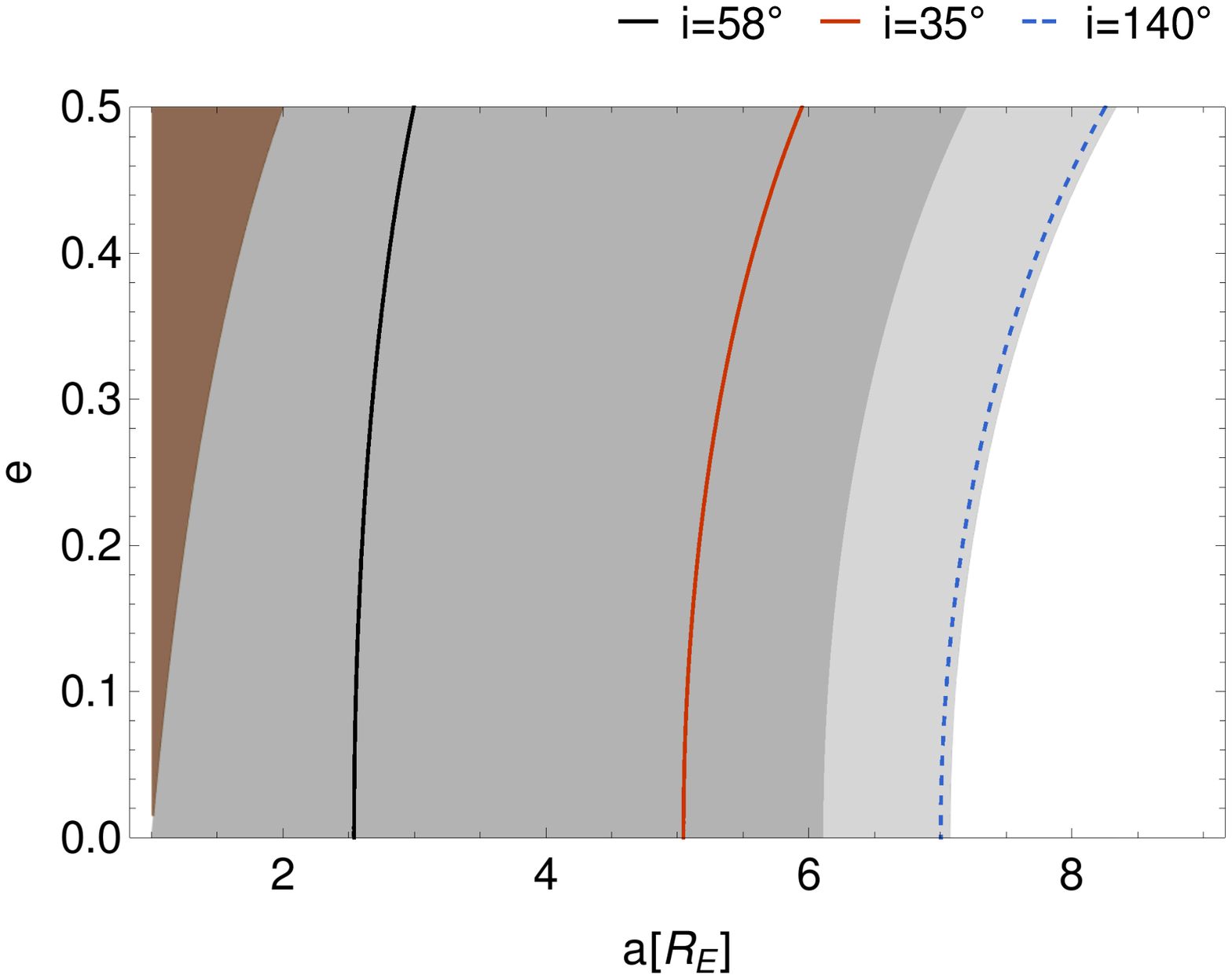}

    \vspace{0.3cm}

\includegraphics[width=0.45\textwidth]{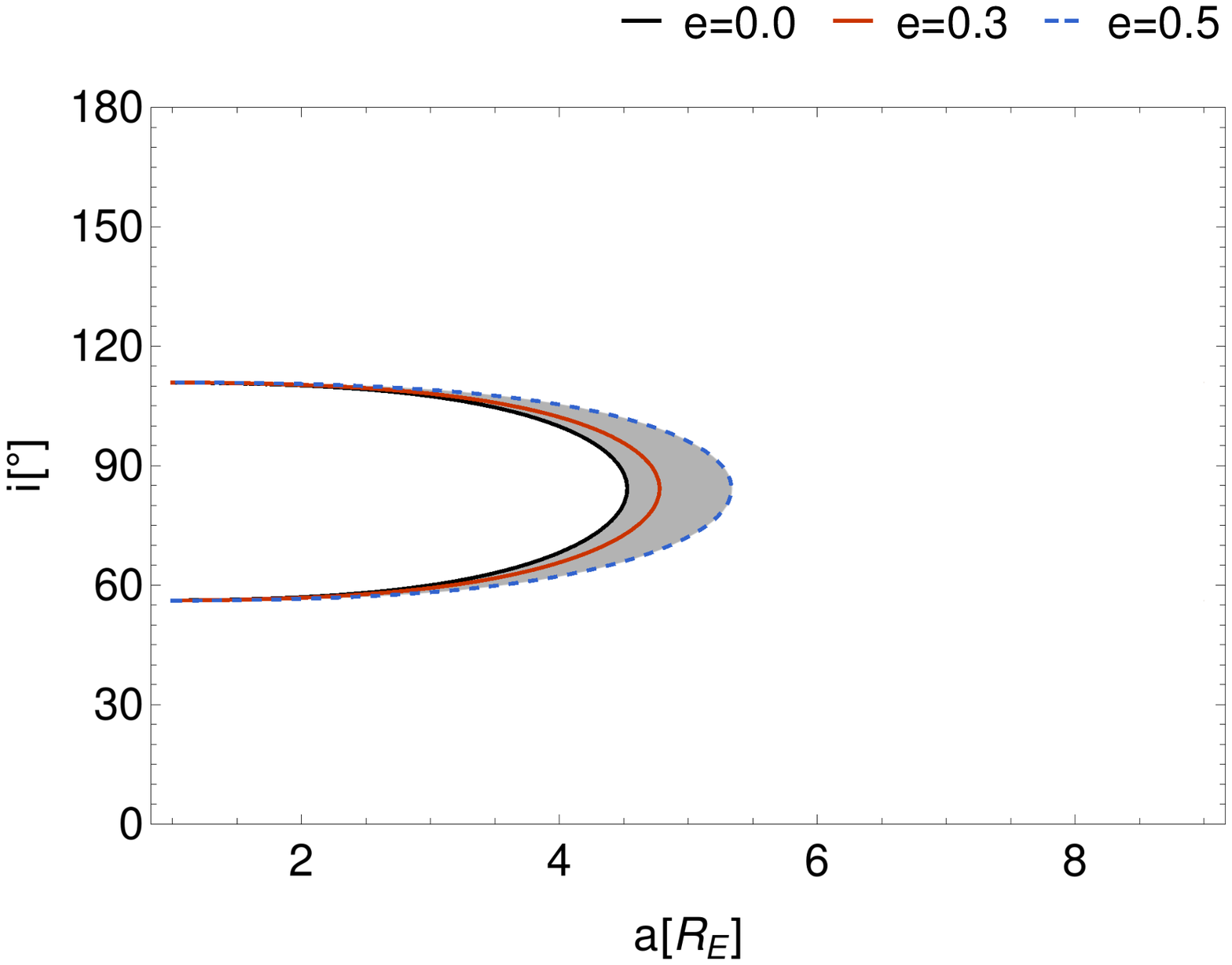}
\includegraphics[width=0.45\textwidth]{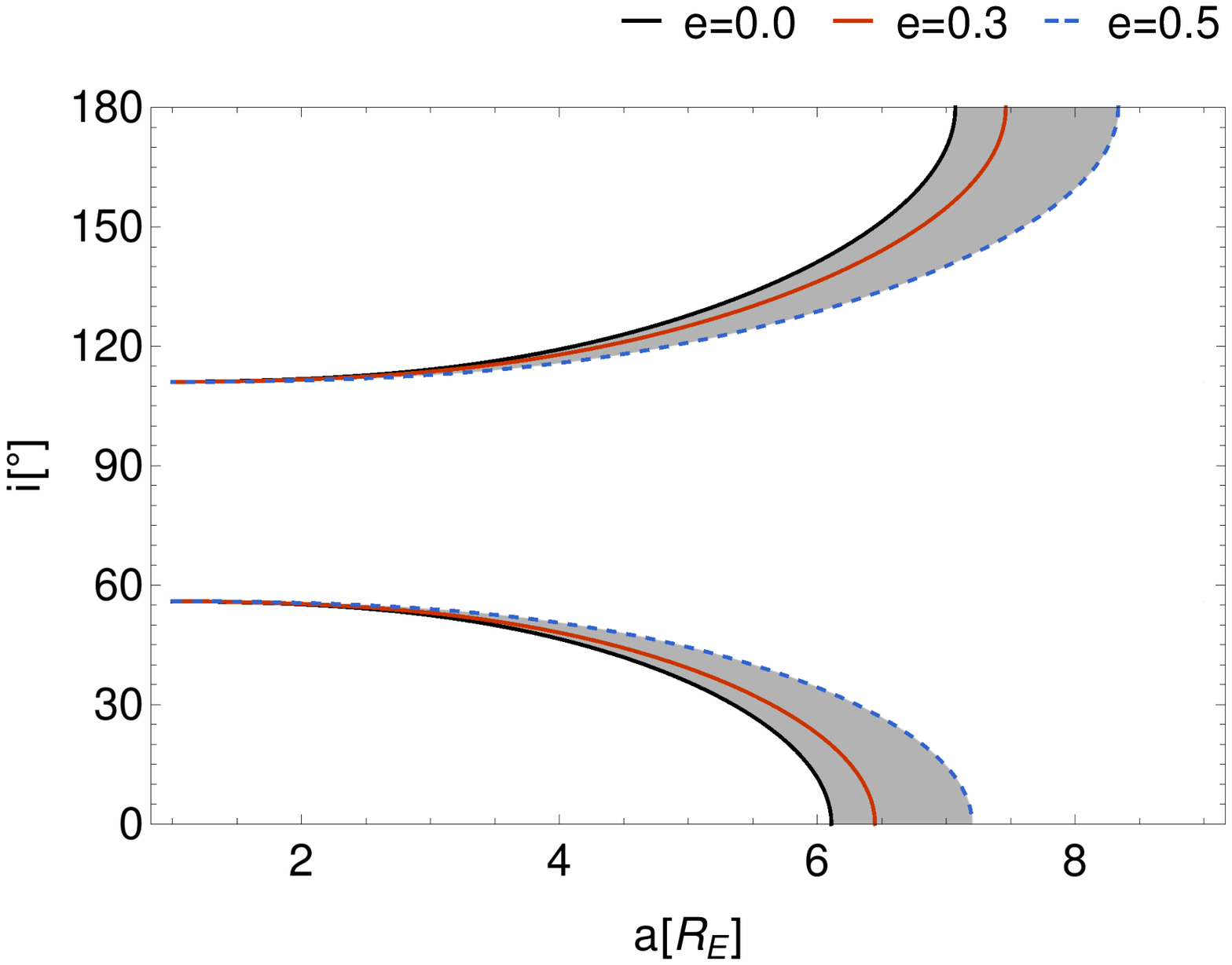}
\caption{Location of Lunar secular resonances in $(a,e)$-space
(top) and $(a,i)$-space (bottom) for $\alpha=2$, $\beta=1$,
$\alpha_M=0$, and $\beta_M=-1$ (left) and respectively
$\beta_M=+1$ (right). In the top panel, dark grey marks regions
with two solutions, light grey marks regions with one solution. At
the bottom panel, grey defines the space where solutions can be
found. Contours are shown for specific parameters provided in the
plot legends. \red{Parameters that lead to collisions are
highlighted in brown}.} \label{f:lunsec}
\end{figure}


\vskip.1in

We remark that another class of resonances is given by the \sl
mean motion resonances between the orbital period of the object and
the Moon, \rm which correspond to solutions of the
equation $\alpha\dot M+\beta\dot M_M=0$ for suitable
$\alpha,\beta\in \mathbb{Z}$; however, here we will not be interested to
such resonances, since they typically occur far from the Earth,
well outside the geostationary orbit.

\red{
\section{Conclusions}\label{sec:conclusions}
}

\red{
Despite the fact that several decades have lapsed since the first satellite was launched in space,
the dynamics of an object moving around the Earth is still an intriguing subject, especially when
considering the large number of space debris populating our sky. Within such context, it is
remarkable that the problem presents different time scales, since the Hamiltonian function describing the dynamics
depends on short-period angles (the mean anomaly of the object and the sidereal time),
secular angles (the arguments of perigee and the longitudes of the ascending node of the object, Sun and Moon),
semi-secular angles given by the mean anomalies of Moon and Sun (which have, respectively, a period of one
month and one year). This remark motivates the study of different types of resonances,
according to the quantities which are involved in the commensurability relation defining the resonance.
Based on elementary computations, our aim is to characterize in the quadrupolar approximation the
different resonances according to the values of the orbital elements, namely semimajor axis, eccentricity
and inclination. This \red{allows us} to specify the orbital elements regions where the resonances
can be found. We summarize below our main results.
}

\red{
        In Section~\ref{sec:tesseral} we study tesseral resonances of \red{order $j:\ell$:}
\begin{itemize}
    \item for fixed values of $a$, $e$, we obtain an expression for the inclination (Proposition~\ref{ref:pro1});
    \item for fixed values of $e$, $i$, we obtain an equation for the semimajor axis (Proposition~\ref{ref:pro2});
    \item for fixed values of $a$, $i$, we get an equation for the eccentricity (Proposition~\ref{ref:pro3}).
\end{itemize}
}

\red{
In Section~\ref{sec:solarsemi} we study the Solar semi-secular resonances and we obtain the following results:
\begin{itemize}
    \item for given values of $a$, $e$, we obtain an equation for $\cos i$, whose discussion leads to
    obtain bounds on the elements, ensuring the existence of zero, one, two solutions (Proposition~\ref{pro1S});
    \item for given values of $e$, $i$, we obtain a value for the semimajor axis and bounds on $e$, $i$,
    ensuring the existence of solutions (Proposition~\ref{pro2S});
    \item for given values of $a$, $i$, we obtain an expression for the eccentricity and bounds on $a$, $i$,
    ensuring the existence of solutions (Proposition~\ref{pro3S}).
\end{itemize}
}

\red{
Similar results are obtained in Section~\ref{sec:lunarsemi} for the Lunar semi-secular resonances as provided
by Propositions~\ref{pro:1L}, \ref{pro:2L}, \ref{pro:3L}.
}

\red{
Finally, in Section~\ref{sec:secular} we shortly analyze the Solar (Section~\ref{sec:solarsec}) and Lunar
(Section~\ref{sec:lunarsec}) secular resonances, which are extensively treated in \cite{CGPSIADS}
and \cite{HughesI}.
}

\vglue1cm

\red{
\section*{Appendix: The Fast Lyapunov Indicator (FLI)} \label{app:chaosIndicators}}

\red{FLI is a familiar chaos tool used to numerically investigate the stability of a dynamical system.
Comparing the values of the FLIs as the initial conditions or parameters are varied, one can distinguish between regular,
resonant or chaotic motions (see \cite{froes},  \cite{GL2018}). Here we briefly recall the definition of FLI and we describe how it is computed for our particular system.}

\red{The Lyapunov Characteristic Exponent (see \cite{BGS}) provides evidence of the chaotic character of the dynamics of a given dynamical system, since it measures the divergence of nearby trajectories. For a phase space of dimension $N$, there exist $N$
Lyapunov exponents, although the largest one is the most significative and is what we
refer to as the \sl Lyapunov exponent. \rm
This choice is motivated by the exponential rate of divergence, since
the greatest exponent dominates the overall separation.}

\red{Lyapunov exponents can be computed as follows (\cite{LAVW}):
let ${\underline \xi}=(L,G,H,M,\omega,\Omega)$ be the phase state associated with the Hamiltonian,
say $\mathcal{H}$ (see Section~\ref{sec:hamiltonian}). We can generically denote the evolution
in phase space as determined by the vector field
$$
\dot{\underline \xi}={\underline f}({\underline \xi})\ ,\qquad {\underline{\xi}}\in \mathbb{R}^6\ ,
$$
and the evolution on the tangent space by the corresponding variational equations
$$
\dot{\underline \eta}=\left({{\partial \underline{f}(\underline{\xi})} \over {\partial
\underline{\xi}}}\right)\ {\underline \eta}\ ,\qquad {\underline{\eta}}\in \mathbb{R}^6\ .
$$
We can assign the initial conditions by choosing $\underline{\xi}(0)$ and
each component of $\underline{\eta}(0) = {\eta}_j  (0) \, \underline{{\hat e}}_j$ in a basis $\underline{{\hat e}}_j$ of the tangent space. Then,
we can compute the quantities
$$
\chi_j \equiv \lim_{t\to\infty}  \lim_{|| \underline{\eta} (0)|| \to 0}{1\over t}
\log {{|{\eta}_j (t)|}\over {|{\eta}_j (0)|}}\ , \qquad j=1,...,6 \ ,
$$
where $\|\cdot\|$ denotes the Euclidean norm.}

\red{When dealing with a {\it Hamiltonian dynamical system} only $N/2$ of the $\chi_j$ are actually meaningful, so in our case we would have three exponents. In view of the exponential rate of divergence, we can concentrate on the greatest of them and estimate it by means of the formula
$$
\chi \equiv \lim_{t\to\infty} {1\over {t}} \log {{|| \underline{\eta} (t)||}\over {|| \underline{\eta} (0)||}}\ ,
$$
where $|| \underline{\eta} (t)||$ is the phase-space Euclidean distance at time $t$ between
trajectories at initial distance $|| \underline{\eta} (0)||$.}

\red{In order to investigate the stability of the dynamics for the models described
in the previous sections, we compute the so-called Fast Lyapunov Indicator, which is defined as the value of the largest Lyapunov characteristic exponent
{\it at a fixed time} (see \cite{froes}). By comparing the values of the FLIs as initial conditions or parameters are varied, one obtains
an indication of the dynamical character of the phase-space trajectories as well as of their chaoticity/regularity behaviour.
The explicit computation of the FLI proceeds as follows: the FLI at a given time $T\geq 0$ is obtained by the expression
$$
{\rm FLI}(\underline{\xi}(0),  \underline{\eta}(0), T) \equiv \sup _{0 < t\leq T}
 \log || \underline{\eta}(t)||\ .
$$
In practice, a \sl reasonable \rm choice of $T$ makes faster the computation of the FLI when compared with
previous expressions for the $\chi$'s where, in principle, very long integration times are required
to obtain a reliable convergence process. Analyzing the frequencies of the model, one can set the integration time $T$ as a good compromise between the accuracy of the computations and the length of the computer programs. For instance, the plots of Figure~\ref{f:tescart} are obtained by integrating, using
a Runge-Kutta fourth order method, the canonical equations and the associated variational equations for an interval of 5\,000 sidereal days. Such a value was calibrated against the frequency $\dot{M}$ of the fastest angle of the system (or equivalently the number of revolutions of the satellite around the Earth),
and it is comparable with the number of revolutions that were taken in
other papers, e.g. \cite{CGmajor, CGminor}.}

\end{document}